\documentclass[nonacm,sigconf]{acmart}
\pdfoutput=1

 \author{ Michael Benedikt}
  \affiliation{
\institution{University of Oxford}
\country{United Kingdom}
}
\author{C\'ecilia Pradic}
\affiliation{
\institution{University of Swansea}
\country{United Kingdom}
}
\author{Christoph Wernhard}
\affiliation{
\institution{University of Potsdam}
\country{Germany}
}

\title{Synthesizing nested relational queries from implicit specifications}

\usepackage{amssymb,amsthm}
\newtheorem*{theorem*}{Theorem}
\newtheorem{theorem}{Theorem}

\newcommand{\nrcwget}{\nrc\kw{[}\nrcget\kw{]}}

\renewcommand{\phi}{\varphi}

\newcommand{\vflat}{V}
\newcommand{\basenest}{B}
\newcommand{\lossless}{\kw{lossless}}
\newcommand{\modelssem}{\models_{\kw{nested}}}
\newcommand{\equi}{\leftrightarrow}
\newcommand{\imp}{\rightarrow}

\newcommand{\inq}{\in}

\newcommand{\memmac}{\mathrel{\hat{\in}}}

\usepackage{xspace}
\newcommand{\kw}[1]{{\mathsf{#1}}\xspace}
\usepackage{color}
\usepackage{mathpartir}
\usepackage{cleveref}
\usepackage{booktabs}
\usepackage{tikz}
\usepackage{plabels}
\usepackage{float}

\makeatletter
\newcommand{\leqnomode}{\tagsleft@true\let\veqno\@@leqno}
\newcommand{\reqnomode}{\tagsleft@false\let\veqno\@@eqno}
\makeatother
\definecolor{colch}{rgb}{0.1,0.5,0.2}

\newcommand{\toconsider}[1]{} %

\newcommand{\la}{\langle}
\newcommand{\<}{\la}
\newcommand{\ra}{\rangle}
\renewcommand{\>}{\ra}

\newcommand{\ipol}[2]{\; :\; \la #1, #2\ra}
\newcommand{\ipolnarrow}[2]{: \la #1, #2\ra}

\newcommand{\defname}[1]{\emph{#1}}
\newcommand{\name}[1]{\emph{#1}}
\newcommand{\valid}{\models}
\newcommand{\entails}{\models}
\newcommand{\sep}{;\;\;}
\newcommand{\D}{\mathcal{D}\xspace}
\newcommand{\vs}{\vv{v}}
\newcommand{\folfocused}{FO-focused\xspace}

\newcommand{\tabrulename}[1]{\raisebox{-2.0ex}[0pt][0pt]{\small #1}}
\newcommand{\tabrulecond}[1]{\raisebox{-2.0ex}[0pt][0pt]{$#1$}}
\newcommand{\rulename}[1]{\textsc{#1}}
\newcommand{\SL}{\Gamma_L}
\newcommand{\SR}{\Gamma_R}

\newcommand{\ADDVARSLAM}{\vv{l}}
\newcommand{\ADDVARSRHO}{\vv{r}}
\newcommand{\ADDPREDSLAM}{\pred(\lambda)}
\newcommand{\ADDPREDSRHO}{\pred(\rho)}

\newcommand{\pdepd}{\kw{PDEPD}}

\newcommand{\ADDLVARS}{\FV^{\mathsf{RHS}}}
\newcommand{\PREDD}{\pred^{\mathsf{RHS}}}
\newcommand{\exdef}{\exists^{\mathsf{RHS}}}
\newcommand{\substdef}[2]{[#2/#1]^{\mathsf{RHS}}}%

\newcommand{\pospolarity}{\kw{EL}}%
\newcommand{\negpolarity}{\kw{AL}}
\newcommand{\pospolaritysuper}{\kw{EL}}%
\newcommand{\negpolaritysuper}{\kw{AL}}

\newcommand{\subst}[2]{[#2/#1]}%

\newcommand{\prooftreesymbol}
{{\footnotesize
\setlength{\arraycolsep}{0pt}
\begin{array}{rcl}
.\;\;\;\;\;&.&\;\;\;\;.\\[-1.8ex]  
.\;\;\;\;&.&\;\;\;.\\[-1.8ex]
.\;\;\;&.&\;\;.\\[-1.8ex]
.\;\;&.&\;.\\[-1.8ex]
.\;&.&.\\[-1.8ex]
&.&
\end{array}}}

\usepackage{tabularx}
\newcolumntype{Z}[1]{>{\raggedleft\let\newline\\\arraybackslash\hspace{0pt}$}p{#1}<$}
\newcolumntype{X}[1]{>{\raggedright\let\newline\\\arraybackslash\hspace{0pt}$}p{#1}<$}                                      \newcolumntype{Y}[1]{>{\centering\let\newline\\\arraybackslash\hspace{0pt}$}p{#1}<$}          

\newenvironment{arrayprf}
               {\begin{array}{Z{1.8em}@{\hspace{1em}}l@{\hspace{0.5em}}l}}
               {\end{array}}

\newcounter{subprop}[theorem]

\newcommand{\seq}{\vdash}

\newcommand{\pred}{\mathit{PRED}}

\usepackage{color}

\usepackage[utf8]{inputenc} %

\usepackage{cellspace, paralist}

\setdefaultleftmargin{10pt}{10pt}{}{}{}{}

\usepackage[linesnumbered,vlined,ruled]{algorithm2e}
\usepackage[utf8]{inputenc} %
\SetAlFnt{\fontsize{9.5}{11}\selectfont} 

\definecolor{darkred}{rgb}{0.7,0,0}
\definecolor{darkblue}{rgb}{0,0,0.7}

\definecolor{colch}{rgb}{0.1,0.5,0.2}

\newcommand\colL[1]{\textcolor{darkred}{#1}}
\newcommand\colR[1]{\textcolor{darkblue}{#1}}
\newcommand\IH{^{\mathsf{IH}}}

\newcommand{\proves}{\vdash}

\usepackage{bussproofs}  \EnableBpAbbreviations
\usepackage{url}
\usepackage{balance}
\usepackage{thm-restate}
\usepackage{cleveref}
\usepackage{upgreek}

\setdefaultleftmargin{10pt}{10pt}{}{}{}{}

\usepackage[linesnumbered,vlined,ruled]{algorithm2e}
\SetAlFnt{\fontsize{9.5}{11}\selectfont} %

\newcommand{\myparagraph}[1]{
  \noindent\textbf{#1.}}

\renewcommand{\myparagraph}[1]{\noindentparagraph{\textbf{\upshape{#1.}}}}

\newcommand{\myeat}[1]{}
\newcommand{\deltazero}{\Delta_0}
\newcommand{\incontext}{\in}
\newcommand\bnfeq{\mathrel{::=}}
\newcommand\bnfalt{\; | \;}
\newcommand\sett{{\sf Set}}

\newcommand{\nrc}{\kw{NRC}}

\newenvironment{myexmp}{\refstepcounter{myexmp}\par\medskip
\noindent\textbf{Example~\themyexmp.}}{\null\hfill$\triangleleft$\medskip}
\newcounter{myexmp}[section]
\renewcommand{\themyexmp}{\thesection.\arabic{myexmp}}

\newcommand{\unit}{\kw{Unit}}

\newcommand\lam\lambda

\newcommand\eqdef{\mathrel{:=}}

\newcommand{\nrcget}{\textsc{get}}

\newcommand\ur{\mathfrak{U}}

\ifdefined\cA
\else
\newcommand\cA{\mathcal{A}}
\fi
\ifdefined\cB
\else
\newcommand\cB{\mathcal{B}}
\fi
\ifdefined\cC
\else
\newcommand\cC{\mathcal{C}}
\fi
\ifdefined\cD
\else
\newcommand\cD{\mathcal{D}}
\fi
\ifdefined\cE
\else
\newcommand\cE{\mathcal{E}}
\fi
\ifdefined\cF
\else
\newcommand\cF{\mathcal{F}}
\fi
\ifdefined\cG
\else
\newcommand\cG{\mathcal{G}}
\fi
\ifdefined\cH
\else
\newcommand\cH{\mathcal{H}}
\fi
\ifdefined\cI
\else
\newcommand\cI{\mathcal{I}}
\fi
\ifdefined\cJ
\else
\newcommand\cJ{\mathcal{J}}
\fi
\ifdefined\cK
\else
\newcommand\cK{\mathcal{K}}
\fi
\ifdefined\cL
\else
\newcommand\cL{\mathcal{L}}
\fi
\ifdefined\cM
\else
\newcommand\cM{\mathcal{M}}
\fi
\ifdefined\cN
\else
\newcommand\cN{\mathcal{N}}
\fi
\ifdefined\cO
\else
\newcommand\cO{\mathcal{O}}
\fi
\ifdefined\cP
\else
\newcommand\cP{\mathcal{P}}
\fi
\ifdefined\cQ
\else
\newcommand\cQ{\mathcal{Q}}
\fi
\ifdefined\cR
\else
\newcommand\cR{\mathcal{R}}
\fi
\ifdefined\cS
\else
\newcommand\cS{\mathcal{S}}
\fi
\ifdefined\cT
\else
\newcommand\cT{\mathcal{T}}
\fi
\ifdefined\cU
\else
\newcommand\cU{\mathcal{U}}
\fi
\ifdefined\cV
\else
\newcommand\cV{\mathcal{V}}
\fi
\ifdefined\cW
\else
\newcommand\cW{\mathcal{W}}
\fi
\ifdefined\cX
\else
\newcommand\cX{\mathcal{X}}
\fi
\ifdefined\cY
\else
\newcommand\cY{\mathcal{Y}}
\fi
\ifdefined\cZ
\else
\newcommand\cZ{\mathcal{Z}}
\fi
\ifdefined\bA
\else
\newcommand\bA{\mathbb{A}}
\fi
\ifdefined\bB
\else
\newcommand\bB{\mathbb{B}}
\fi
\ifdefined\bC
\else
\newcommand\bC{\mathbb{C}}
\fi
\ifdefined\bD
\else
\newcommand\bD{\mathbb{D}}
\fi
\ifdefined\bE
\else
\newcommand\bE{\mathbb{E}}
\fi
\ifdefined\bF
\else
\newcommand\bF{\mathbb{F}}
\fi
\ifdefined\bG
\else
\newcommand\bG{\mathbb{G}}
\fi
\ifdefined\bH
\else
\newcommand\bH{\mathbb{H}}
\fi
\ifdefined\bI
\else
\newcommand\bI{\mathbb{I}}
\fi
\ifdefined\bJ
\else
\newcommand\bJ{\mathbb{J}}
\fi
\ifdefined\bK
\else
\newcommand\bK{\mathbb{K}}
\fi
\ifdefined\bL
\else
\newcommand\bL{\mathbb{L}}
\fi
\ifdefined\bM
\else
\newcommand\bM{\mathbb{M}}
\fi
\ifdefined\bN
\else
\newcommand\bN{\mathbb{N}}
\fi
\ifdefined\bO
\else
\newcommand\bO{\mathbb{O}}
\fi
\ifdefined\bP
\else
\newcommand\bP{\mathbb{P}}
\fi
\ifdefined\bQ
\else
\newcommand\bQ{\mathbb{Q}}
\fi
\ifdefined\bR
\else
\newcommand\bR{\mathbb{R}}
\fi
\ifdefined\bS
\else
\newcommand\bS{\mathbb{S}}
\fi
\ifdefined\bT
\else
\newcommand\bT{\mathbb{T}}
\fi
\ifdefined\bU
\else
\newcommand\bU{\mathbb{U}}
\fi
\ifdefined\bV
\else
\newcommand\bV{\mathbb{V}}
\fi
\ifdefined\bW
\else
\newcommand\bW{\mathbb{W}}
\fi
\ifdefined\bX
\else
\newcommand\bX{\mathbb{X}}
\fi
\ifdefined\bY
\else
\newcommand\bY{\mathbb{Y}}
\fi
\ifdefined\bZ
\else
\newcommand\bZ{\mathbb{Z}}
\fi

\newcommand{\ptime}{\kw{PTIME}}
\newcommand{\booltype}{\kw{Bool}}

\newcommand{\sig}{{\mathcal SIG}}

\newcommand{\ursort}{\ur}

\newcommand\FV{FV}

\newcommand\tuple[1]{\langle #1 \rangle}

\newcommand\myLeftLabel[1]{\LeftLabel{\scriptsize #1}}

\begin{document}
\begin{abstract}
Derived datasets can be defined implicitly or explicitly.
An \emph{implicit definition} (of dataset $O$ in terms of datasets $\vec I$)
 is a logical specification 
involving the source data $\vec I$ and the interface data $O$.  It is a valid
definition  of $O$ in terms of $\vec I$,
if  any two models of the specification agreeing on $\vec I$ agree on $O$. In contrast, an
\emph{explicit definition} is a query that produces $O$ from $\vec I$.
Variants of \emph{Beth's
theorem} \cite{beth}
state  that one can convert implicit definitions to explicit ones. Further, this
conversion can be done effectively given a proof witnessing implicit definability
in a suitable proof system.
We prove the analogous effective implicit-to-explicit result for nested relations:
implicit definitions, given in the natural logic for nested relations, can be
effectively converted to explicit definitions
 in the nested relational calculus ($\nrc$). As a  consequence,
we can effectively extract
rewritings of $\nrc$ queries in terms of $\nrc$ views, given a proof
witnessing that the query is determined by the views.
\end{abstract}

\maketitle

\section{Introduction}

One way of describing a virtual datasource is via \emph{implicit definition}:
a  specification $\Sigma$ -- e.g. in logic --  involving
 symbols for the  ``virtual'' object $O$ and the stored ``input'' data $\vec I$. The specification
may mention other data objects (e.g. auxiliary views). But to be an implicit definition,
any two models of $\Sigma$ that agree on $\vec I$ must agree on $O$.
In the case where $\Sigma$ is in first-order logic, 
this hypothesis can be expressed as a first-order entailment, using two copies of the vocabulary, primed
and unprimed, representing the two models:
\begin{equation}
  \tag{$\star$}
\Sigma \wedge \Sigma' \wedge \bigwedge_{I_i \in \vec I} \forall \vec x_i ~ [I_i(\vec x_i) \leftrightarrow I'_i(\vec x_i)] \models 
\forall \vec x ~ [O(\vec x) \leftrightarrow O'(\vec x)]
\end{equation}
Above  $\Sigma'$ is a copy
of $\Sigma$ with primed versions of each predicate.

A fundamental result in logic states that we can replace an implicit definition
with an \emph{explicit definition}: a first-order query $Q$ such that whenever
$\Sigma(\vec I, O, \ldots)$ holds, $O=Q(\vec I)$. The original result of this kind
is \emph{Beth's theorem} \cite{beth}, which deals with classical first-order logic.
Segoufin and Vianu's \cite{svconf} looks at the case where
$\Sigma$ is in \emph{active-domain first-order logic}, or equivalently a Boolean
relational algebra expression. Their conclusion is that one can produce an
explicit definition of $O$ over $\vec I$ in  relational
algebra. \cite{svconf} focused on the special case where $\Sigma(I_1 \ldots I_j, 
\vec B, O)$ 
specifies each $I_i$ as  a view defined
by an active-domain first-order formula $\phi_{V_i}$ over base
data $\vec B$, and also defines $O$ as an active-domain first-order query $\phi_Q$ over
$\vec B$. In this case, $\Sigma$ implicitly defining $O$ in terms of $\vec I$
is called ``determinacy of the query  by the views''.
Segoufin and Vianu's result implies that \emph{whenever  a relational algebra  query
$Q$ is determined by relational algebra views $\vec V$,
then  $Q$ is rewritable over the views by a relational algebra query}.

Prior Beth-style results like \cite{beth,svconf}  are effective. From a proof of the entailment $(\star)$
 in a suitable
proof system, one can extract an explicit definition effectively,
even in polynomial time. 
In early proofs of Beth's theorem, the proof systems were custom-designed 
for  the task of proving implicit definitions, and the bounds were not stated.
Later on 
standard proof systems such as tableaux
\cite{smullyan} or resolution \cite{huang} were employed, and the polynomial claim was explicit. 
It is important that in our definition of implicit definability, we require the
existence of a proof witness. By the completeness theorem for first-order logic,
requiring such a proof witness is equivalent to demanding that implicit definability
of $O$ over $\vec I$ holds for all instances, not just finite ones.

This paper deals with the situation for \emph{nested relations}, a data model heavily
explored in the database community. There is a natural analog of active domain first-order
logic, suitable for implicit specification.
These are the $\deltazero$ formulas,  logical expressions
where quantification is over elements within nested sets defined by terms.
The notion of a $\deltazero$ specification $\Sigma(\vec i, o, \ldots)$ implicitly defining nested relation
$o$ in terms of $\vec i$ is the obvious one: for any two nested relations
satisfying $\Sigma$, and  agreeing on $\vec i$, they must agree on $o$.
There is also a natural notion of proof witness for determinacy, using a proof system
for $\deltazero$ formulas. The analog of relational algebra for explicit definitions
is \emph{nested relational calculus} $\nrc$ \cite{limsoonthesis}, which is the standard
query language for nested relations.
Our main result is:

\smallskip

From a proof $p$ that $\Sigma$ implicitly defines $o$ 
in terms of $\vec i$, we can obtain, in $\ptime$,  an $\nrc$ expression $E$ that explicitly
defines $o$ from $\vec i$, relative to $\Sigma$.

\smallskip

A special case of this result concerns $\nrc$ views and queries.
Our result implies that \emph{if we have $\nrc$ views $\vec V$ that determine an $\nrc$ query
$Q$, then we can generate -- from a suitable proof -- an $\nrc$ rewriting of $Q$
in terms of $\vec V$}.

The fact that such an $\nrc$ rewriting \emph{exists} whenever there is a functional
relationship was
proven in \cite{benediktpradicpopl}. But the argument for existence given there is model-theoretic.
Indeed,  \cite{benediktpradicpopl} listed coming up with an algorithm for constructing
such a rewriting for a classically-complete proof system as a major open issue.

\begin{myexmp} \label{ex:views}
We consider the case where our  specification $\Sigma(Q,V,B)$ 
describes  a view  $V$, a query $Q$, as well as some constraints on
the base data $B$.
Our base data $B$ is
of type $\sett(\ur \times \sett(\ur))$, 
where $\ur$
refers to the basic set of elements, the ``Ur-elements''.
That is, $B$ is a set of pairs, where the first item is a data item
and the second is a set of data items.  
View  $V$ is of type $\sett(\ur \times \ur)$, a set of pairs, given
by the query that is the usual ``flattening'' of $B$: in $\nrc$ this
can be expressed as $\{ \< \pi_1(b) , c \> \mid c \in \pi_2(b) \mid b \in B\}$.
The view  definition can be converted to a  specification in our logic.

A query $Q$ might ask for a selection of the pairs in $B$, those whose first
component is contained in the second:
$\{ b \in B ~ | ~ \pi_1(b) \in \pi_2(b) \}$.
The definition of $Q$ can also be incorporated
into our specification.

View $V$ is not sufficient to answer $Q$ in general. This is
the case if 
we assume  as part of $\Sigma$ an integrity constraint stating that the 
first component of $B$ is a key.
We can prove that 
$\Sigma(Q,V,B)$ implicitly defines $Q$ in terms of $V$, and from this proof
our algorithm can produce an $\nrc$  rewriting of $Q$ in terms of $V$.
\end{myexmp}

\myparagraph{Organization}
We overview related work in Section \ref{sec:related} and
provide preliminaries in Section \ref{sec:prelims}. 
Section \ref{sec:statement} presents our main result. It is proven
in Section \ref{sec:proof:ain}, making use of infrastructure from Section~\ref{sec:tools}.
We close with discussion in Section \ref{sec:conc}.
Due to space constraints, many proofs are deferred to the appendix.

\section{Related work} \label{sec:related}
In addition to the theorems of  Beth and Segoufin-Vianu mentioned in the introduction,
there are numerous works on effective Beth-style results for other logics.
Some concern fragments of classical first-order logic, such as the guarded fragment
\cite{HMO,guardedinterpj};
others deal with non-classical logics such as description logics \cite{balderbethdl}.  The Segoufin-Vianu result is
closely related to variations of Beth's theorem and Craig interpolation  for relativized quantification, such as Otto's
interpolation theorem
\cite{otto}.
There are also effective interpolation and definability results for logics richer
than or incomparable to
first-order logic,
such as fragments of fixpoint logics \cite{interpolationmucalc,
lmcsusinterpolfixedpoint}. There are even Beth-style results for full infinitary logic
\cite{lopezescobar}, but there one can not hope for effectivity.
The connection between Beth-style results and view rewriting originates in \cite{svconf,NSV}.
The idea of using effective Beth results to generate view rewritings from proofs
appears in \cite{franconisafe}, and is explored in more detail first in
\cite{tomanweddell} and later in \cite{interpbook}.

Our main result relates to Beth theorems  ``up-to-isomorphism''.  Our implicit
definability hypothesis is that two models that satisfy a  specification and
agree on the inputs must agree on the output nested relations, where ``agree on the output''
means up to extensional equivalence
of sets, which is a special (definable) kind of isomorphism.
Beth-like theorems up to isomorphism originate in Gaifman's \cite{gaifman74} and are studied
extensively by Hodges and his collaborators (e.g. \cite{hodgesnormal,hodgesdugald,hodgesbook}).
  The focus is model-theoretic, with emphasis
on  connections with categoricity and classification in classical model theory.

Our effective Beth-like theorem for nested relations extends two results in
 \cite{benediktpradicpopl}. One is an ineffective result, which makes use
of an idea in \cite{gaifman74},  but without any effectivity.
Another is an effective result, but only for a 
very restricted notion of ``constructive proof'', which
is not complete for classical logic. \cite{benediktpradicpopl} lists coming up with
a general effective notion as the major unfinished business of the paper. Our resolution of this
problem requires a new proof system and much more complex tools than those employed in \cite{benediktpradicpopl}.

\section{Preliminaries} \label{sec:prelims}

\myparagraph{Nested relations}
We deal with schemas that describe objects of various
 \emph{types} given by the following grammar.
$$T, \; U \bnfeq \ur \bnfalt T \times U \bnfalt \unit \bnfalt \sett(T)$$
For simplicity throughout the remainder
we will assume only two basic types. There is the one-element type $\unit$, which will
be used to construct Booleans.
And there is  $\ur$, the ``scalars'' or \emph{Ur-elements} whose inhabitants are not specified further.
From the Ur-elements and a unit type we can build up the set of types via
product and the power set operation.  
We use standard conventions for abbreviating types, with the $n$-ary product abbreviating
an iteration of binary products.
A \emph{nested relational schema} consists of declarations of 
variable names associated to objects of given types. 

\begin{myexmp} \label{ex:aschema}
An example nested relational schema declares two objects
$R: \sett(\ur \times \ur)$ and $S: \sett( \ur \times \sett(\ur))$.
That is, $R$ is a set of pairs of Ur-elements: a standard ``flat'' binary relation.
$S$ is a collection of pairs whose first elements are Ur-elements
and whose second elements are sets of Ur-elements.
\end{myexmp}

The types have a natural interpretation.
The unit type has a unique member and the members of $\sett(T)$
are the sets of members of $T$.
An \emph{instance} of such a schema is defined in the obvious way.

For the schema in Example \ref{ex:aschema} above, assuming that $\ur = \mathbb{N}$, one possible instance has
$R = \{ \<4,  6\>, \<7,  3\>\}$ and 
$S = \{ \<4 , \{  6,  9 \} \> \}$.

\myparagraph{$\deltazero$ formulas}
We need a logic appropriate for talking about nested relations.
A natural and well-known subset of first-order logic formulas with
a set membership relation are the $\deltazero$ formulas.
They are built up from equality of Ur-elements via Boolean operators
as well as relativized existential and universal quantification. All terms involving
tupling and projections are allowed. %

Formally, we deal with multi-sorted first-order logic, with sorts corresponding
to each of our types. We use the following syntax for $\deltazero$ formulas and terms.
Terms are built from variables using tupling and projections.
All formulas and terms are assumed to be well-typed in the obvious way, with the
expected sort of $t$ and $u$ being $\ursort$ in expressions
$t =_\ursort u$ and $t \neq_\ursort u$, and
in $\exists t \in_T u ~ \phi$  the sort of $t$ is $T$ and the sort of $u$ is  $\sett(T)$.

$$
\begin{array}{lcl}
t, u &\bnfeq& x \bnfalt () \bnfalt \< t, u \>  \bnfalt \pi_1(t) \bnfalt \pi_2(t)
\\
\varphi, \psi &\bnfeq& t =_\ursort t' \bnfalt t \neq_\ursort t' \bnfalt \top \bnfalt \bot \bnfalt \varphi \vee \psi \bnfalt \phi \wedge \psi \bnfalt \\
& & \forall x \in_T t~ \varphi(x) \bnfalt \exists x \in_T t~ \varphi(x)
\end{array}
$$
Note that there is  no primitive negation, and no equalities for sorts other than $\ursort$.
Negation $\neg \phi$ will be  defined as a macro
by induction on $\phi$ by dualizing every connective.
Other connectives can be derived in the usual way on top of negation:
 $\phi \rightarrow \psi$ by $\neg \phi \vee \psi$. 

More crucial is the fact that a
$\deltazero$ formula does not allow membership atoms.
An \emph{extended $\deltazero$ formula} allows membership literals $x \in_T y$,
$x \notin_T y$
 at every type $T$. 

The notion of an extended $\deltazero$  formula $\phi$ \emph{entailing} another formula
$\psi$ is the standard one in first-order logic,
meaning that every model of $\phi$ is a model of $\psi$.
We emphasize here that by every model, we include models where membership is not
extensional.
An important point is that:
\emph{when $\phi$ and $\psi$ are $\deltazero$, rather than extended $\deltazero$, ``every model'' can
be replaced by ``every nested relation''}.  When we consider formulas that
are  $\deltazero$, we write  $\phi \modelssem \psi$ for entailment.
The point above is due to two facts. First,
we have neither $\in$ nor  equality at higher types as an atomic predicate.
This guarantees that any model 
can be modified, without changing the truth value of $\deltazero$ formulas,
into a model satisfying extensionality: if we have $x$ and $y$ with
$(\forall z \in_T x~~ z \in_T y) ~~\wedge~~ (\forall z \in_T y~~ z \in_T x)$
then  $x$ and $y$ must be the same.
Secondly, a well-typed extensional  model is isomorphic to a nested relation,
by the well-known
Mostowski collapse construction that iteratively identifies elements that
have the same members.
The lack of  primitive membership and equality relations in $\deltazero$ formulas
allows us to avoid having to consider extensionality axioms, which would require
special handling in our proof system.  

Equality, inclusion and membership predicates ``up to extensionality''
may be defined as macros by induction on the involved types, while staying within $\deltazero$
formulas.
$$
\begin{array}{rcl}
t \memmac_T u &\eqdef& \exists z' \in u\; t \equiv_T z' \\
t \subseteq_T u &\eqdef& \forall z \in t ~~ z \memmac_T u 
\\
t \equiv_{\sett(T)} u &\eqdef& {t \subseteq_T u} ~~\wedge ~~{u \subseteq_T t} \\
t \equiv_\unit u &\eqdef& \top
\qquad
\qquad
\qquad
t \equiv_\ur u ~~\eqdef~~ t =_\ur u
\\
t \equiv_{T_1 \times T_2} u &\eqdef& {\pi_1(t) \equiv_{T_1} \pi_1(u)}~~ \wedge~~ {\pi_2(t) \equiv_{T_2} \pi_2(u)
}\\
\end{array}
$$
We will use small letters for variables in $\deltazero$
formulas, except in examples when we sometimes use capitals
to  emphasize that an object is of set type.
We drop the type subscripts $T$ in bounded quantifiers, primitive memberships,
and macros $\equiv_T$ when clear.
Of course  membership-up-to-equivalence $\memmac$ and  membership  $\in$ agree
on extensional models. But $\in$ and $\memmac$ are not interchangeable on general models,
and hence are not interchangeable in $\deltazero$ formulas. For example:
\[
x \in y ,  x \in y' \models \exists z \in y ~ z \in y'
\]
But we do not have
\[
x \memmac y, x \memmac y' \models  \exists z \in y ~ z \in y'
\]
A set of primitive membership expressions $t \incontext u$  (i.e.
extended $\deltazero$ formulas)
will be called an $\in$\emph{-context}.

Let us now introduce notation for instantiating a block of bounded quantifiers at a time.
A \emph{variable membership atom} is a membership atom $x \in y$ where $x, y$ are variables
An  \emph{ordered variable $\in$-context} is a list of variable membership atoms.

Given a variable membership atom $x \in y$ and a $\deltazero$ formula
of the form $\phi_0 = \exists w \in y ~ \phi_1$, the \emph{specialization of $\phi_0$ using $x \in y$}
is simply $\phi_1[x/w]$. We generalize this to specializing $\phi_0$ using
an ordered variable $\in$-context $x_1 \in y_1 \ldots x_i \in y_i$ by induction on $i$:
 when $\phi_0= \exists x_1 \in y_1 ~ \phi_1$, 
 we first let $\phi_1$ be the specialization of $\phi_0$ using $x_1 \in y_1$ and then
let $\phi'_1$ be the specialization of $\phi_1$ using  $x_2 \in y_2 \ldots x_i \in y_i$, the latter given
inductively. If  $\phi_1$ or $\phi'_1$ is not of the required form, then the specialization
is not defined.  A specialization of $\phi$ with respect to a variable $\in$-context is
a specialization with respect to an ordering of some subset of the context.
A \emph{maximal specialization (max. spec.) of $\phi_0$} with respect to a
$\in$-context is a specialization which is not existential leading. That is, no other variable
membership can be applied to perform further specialization.

\myparagraph{Nested Relational Calculus}
We review the main language for declaratively transforming  nested relations, Nested
Relational Calculus ($\nrc$). Variables occurring in expressions
are  typed, and each expression is  associated with 
an \emph{output type}, both of these being in the type system described above.
We let $\booltype$ denote the type $\sett(\unit)$. Then $\booltype$ has exactly
two elements, and will be used to simulate Booleans.
The grammar for $\nrc$ expressions  is presented in
Figure \ref{fig:nrc_type}.
\begin{figure}[t]
  \[
    \begin{array}{rl@{}l@{~~}r}
      E, E' \bnfeq& x& \bnfalt () \bnfalt \tuple{E, E'} \bnfalt \pi_1(E) \bnfalt \pi_2(E) \bnfalt
      & {\footnotesize \text{(variable, (un)tupling)}} 
      \\
       && \{E\} \bnfalt \nrcget_T(E) \bnfalt \bigcup \{ E \mid x \in E' \} 
      & {\footnotesize \text{((un)nesting, binding union)}}
\\
       && \bnfalt \emptyset \bnfalt E \cup E' \bnfalt E \setminus E'
& {\footnotesize \text{(finite unions, difference)}}
  \end{array}\]
  \vspace{-1em}
  \caption{$\nrc$ syntax (typing rules omitted)}
\label{fig:nrc_type}
\vspace{-0.2cm}
\end{figure}

The definition
of the free and bound variables of an expression is standard,
the union operator $\bigcup \{ E \mid x \in R \}$ binding the variable~$x$.
The semantics of these expressions should be fairly evident, see
\cite{limsoonthesis}.
If $E$ has type $T$, and has input (i.e. free) variables $x_1 \ldots x_n$
of types $T_1 \ldots T_n$, respectively, then the semantics associates with $E$
a function that 
given a binding associating each free variable with a value of the appropriate type,
returns an object of type $T$.
For example, the expression $()$ always returns the empty tuple, while
$\emptyset_T$ returns the empty set of type $T$.

The language $\nrc$ as originally defined cannot express certain natural
transformations whose output type is $\ur$.
To get a canonical language for such transformations, above we included in our
$\nrc$ syntax a family of operations
$\nrcget_T : \sett(T) \to T$ that extracts the unique element from a singleton.
$\nrcget$ was considered in \cite{limsoonthesis}.
The semantics are: if $E$ returns a singleton set $\{x\}$, then $\nrcget_T(E)$
returns $x$; otherwise it returns some default object of the appropriate type.
In \cite{suciuthesis}, it is shown that $\nrcget$ at sort $\ursort$ is not expressible
using the other constructs in $\nrc$.
However, $\nrcget_T$ for general $T$ is definable from $\nrcget_\ursort$
and the other $\nrc$ constructs.

As explained in prior work (e.g. \cite{limsoonthesis}), on top of
the $\nrc$ syntax above we can support richer operations as  ``macros''.
For every type~$T$ there is an $\nrc$ expression
$=_T$ of type $\booltype$ representing equality of elements of type $T$.
In particular, there is an expression $=_\ur$ representing equality between
Ur-elements. 
 For every type $T$ there is an $\nrc$ expression $\in_T$ of type
$\booltype$ representing membership between an element of type $T$
in an element of type $\sett(T)$. We can define conditional expressions, 
joins, projections on $k$-tuples, and $k$-tuple formers.
$\nrc$ is efficiently \emph{closed under composition}: given $E(x, \ldots )$ and 
$F(\vec i)$ with output type matching the type of input
variable $x$, we can form an expression
$E(F)$ whose free variables are those of $E$ other than $x$,
unioned with those of $F$.

Finally, we note that $\nrc$ is closed under \emph{$\deltazero$ comprehension}:
if $E$ is in $\nrc$, $\phi$ is a $\deltazero$ formula, then
we can efficiently form an expression  
$\{z \in E \mid \phi \}$ which returns the subset of $E$ such
that $\phi$ holds.
We  make use of these macros freely in our examples of $\nrc$,
such as Example \ref{ex:views}.

\myparagraph{Connections between $\nrc$ queries using $\deltazero$ formulas}
Given an $\nrc$ expression $E$ with input relations $\vec i$,
we can create %
a $\deltazero$ formula $\Sigma_E(\vec i, o)$ that
is an \emph{input-output specification of $E$}:
a formula such that $\Sigma_E$ implies $o=E(\vec i)$ and whenever nested relations $0_0, \vec i_0$ satisfies
$o_0=E(\vec i_0)$,  there is  a model  with $\Sigma_E$ holding whose
restriction to $\vec i, o$ is $\vec i_0, o_0$.
For 
the ``composition-free'' fragment, in
which comprehensions $\bigcup$ can only be over input
variables, this conversion can be done in $\ptime$. But it cannot be done
efficiently for general $\nrc$, under  complexity-theoretic hypotheses
 \cite{koch}.

We  also write entailments that use $\nrc$ expressions, e.g.\\
$\phi(x, \vec c \ldots) \modelssem x \in E(\vec c)$ for $\phi$ $\deltazero$
and $E \in \nrc$. An entailment
with $\modelssem$  involving $\nrc$ expressions
 means that in every \emph{nested relation} satisfying $\phi$, $x$
is in the output of $E$ on $\vec c$. Note that the semantics
of $\nrc$ expressions is only defined on nested relations.

\section{Implicit vs Explicit and the statement of the main result} \label{sec:statement}

We now formalize our implicit-to-explicit result.
A $\deltazero$ formula $\varphi(\vv i, \vv a, o)$ \emph{implicitly defines
variable $o$ in terms of variables $\vv i$ up to extensionality} if we have
\begin{equation}
  \tag{$\star$}
\varphi(\vv i, \vv a, o) \; \wedge \; \varphi(\vv i, \vv a', o')
~\modelssem~ o \equiv_T o'
\end{equation}

Recall that $\equiv_T$ is equivalence-modulo-extensionality. It can be replaced
by equality if we we add  extensionality axioms on the left of the entailment symbol.
$\varphi$ will be called an  \emph{implicit definition up to extensionality} of $o$ in terms of $\vv i$.

An $\nrc$ expression $E$ using free variables in $\vv i$ \emph{explicitly defines  $o$ up to extensionality relative to $\deltazero$
formula $\varphi(\vv i , \vv a , o)$} if for every model of $\varphi$,
$E$ applied to $\vv i$ produces $o'$ with $o' \equiv_T o$.
Assuming extensionality, the conclusion is equivalent to $o'=o$.

In \cite{benediktpradicpopl}, it was shown that 
implicit $\deltazero$ definitions can be converted to $\nrc$ definitions:
\begin{theorem}[\cite{benediktpradicpopl}]  \label{thm:bethnrcmodeltheoretic}
$\deltazero$ formula $\varphi(\vv i, \vv a, o)$ implicitly defines
$o$ with respect to $\vv i$ up to extensionality if and only if there
is an $\nrc$ expression $E(\vv i)$ that explicitly defines $o$ up to extensionality relative to~$\varphi$.
\end{theorem}

We explain how Theorem \ref{thm:bethnrcmodeltheoretic} implies
 Segoufin and Vianu's \cite{svconf} result for relational
algebra. 
Suppose $\Sigma(\vec I,  O, \ldots)$ is a single-sorted
first-order logic formula over predicates that include $\vec I \cup \{O\}$, using
only \emph{active domain quantification} -- quantification over the union of projections of 
predicates --
and suppose that  any two models of $\Sigma$ that agree on $\vec I$ agree on $O$.
Such a $\Sigma$ can be considered a special kind of $\deltazero$ formula, and the hypothesis
implies that $O$ is implicitly defined by $\vec I$ relative to $\Sigma$.
Our conclusion is that there is an $\nrc$ expression that produces $O$ from  $\vec I$.
We now use well-known results about the ``conservativity'' of $\nrc$ over relational algebra
for set-to-set  transformation \cite{conservativity,simulation,limsoonthesis}:
$\nrc$ expressions 
transforming relations to relations can be converted to  relational algebra expressions.

\myparagraph{Proof systems for $\deltazero$ formulas}
Our main result is an effective version of Theorem \ref{thm:bethnrcmodeltheoretic}.
For this we need to formalize our proof system for $\deltazero$ formulas, which
will allow us to talk about proof witnesses for implicit definability.

If we want to talk only about  \emph{effective generation} of $\nrc$ witnesses from
proofs, we can use a basic proof system for $\deltazero$ formulas, whose
inference rules
are shown in Figure
\ref{fig:dzcalc-2sided}.
\begin{figure} 
\begin{mathpar}

\inferrule*
[left={Ax}%
]
{ }{\Theta; \; \Gamma, \phi \seq \phi, \Delta}

\and

\inferrule*
  [left={$\bot$\textsc{-L}}]
  {  }{\Theta; \; \Gamma, \bot \seq \Delta}

\\

\inferrule*
  [left={$\neg$}\textsc{-L}]
{
\Theta; \; \Gamma \seq \neg \phi, \Delta
}{\Theta; \; \Gamma, \phi \seq \Delta}
\and
  \inferrule*
  [left={$\neg$}\textsc{-R}]
{
\Theta; \; \Gamma, \phi \seq \Delta
}{\Theta; \; \Gamma \seq \neg \phi, \Delta}
\\

\inferrule*
  [left={$\wedge$\textsc{-R}}]
{
\Theta; \; \Gamma \seq \phi_1, \Delta
\and
\Theta; \; \Gamma \seq \phi_2, \Delta
}{\Theta; \; \Gamma \seq \Delta, \phi_1 \wedge \phi_2}
\and
\inferrule*
  [left={$\vee$}\textsc{-R}]
{
\Theta; \; \Gamma \seq \phi_1, \phi_2,\Delta
}{\Theta; \; \Gamma \seq \phi_1 \vee \phi_2, \Delta}
  \\

\inferrule*
  [left={$\forall$\textsc{-R}}, right={$y$ fresh}]
{\Theta, y \in b; \; \Gamma \seq \phi[y/x], \Delta}{\Theta; \; \Gamma \seq \forall x \in b \;\phi, \Delta}
\and
\inferrule*
  [left={$\exists$\textsc{-R}}]
{\Theta, t \in b; \; \Gamma \seq \phi[t/x], \exists x \in b\;\phi,\Delta}{\Theta, t \in b; \; \Gamma \seq \exists x \in b\;\phi, \Delta}\\
\inferrule*
[left={Refl}]
{\Theta; \; \Gamma, t =_\ur t \seq \Delta}
{\Theta; \; \Gamma \seq \Delta}
\and
\inferrule*
[left={Repl}%
]
{\Theta; \; \Gamma, t=_\ur u, \phi[u/x], \phi[t/x] \seq \Delta}
{\Theta; \; \Gamma, t =_\ur u, \phi[t/x] \seq \Delta}

\\
\inferrule*
[left={$\times_\eta$}, right={$x_1,x_2$ fresh}]
{\Theta[\tuple{x_1,x_2}/x]; \; \Gamma[\tuple{x_1,x_2}/x] \seq \Delta[\tuple{x_1,x_2}/x]}
{\Theta; \; \Gamma \seq \Delta}
\and
\inferrule*
[left={$\times_\beta$}, right={$i \in \{1,2\}$}]
{\Theta[x_i/x]; \; \Gamma[x_i/x] \seq \Delta[x_i/x]}
{\Theta[\pi_i(\tuple{x_1,x_2})/x]; \; \Gamma[\pi(\tuple{x_1,x_2})/x] \seq \Delta[\pi_i(\tuple{x_1,x_2})/x]}

\end{mathpar}
\caption{Proof rules for a $\deltazero$ calculus, without restrictions for efficient generation of witnesses.}
\label{fig:dzcalc-2sided}
\end{figure}

The node labels are a variation of  the traditional rules for first-order logic, with a couple
of quirks related to  the specifics of $\deltazero$  formulas.
Each node label has shape
$\Theta; \Gamma \vdash \Delta$ where
\begin{compactitem}
\item $\Theta$ is an 
$\in$-context.
Recall that these are sets of membership atoms --- the only formulas in our proof
system that are extended $\deltazero$ but not $\deltazero$.
They will emerge during proofs involving $\deltazero$ formulas when
we start breaking down bounded-quantifier formulas.  
\item $\Gamma$ and $\Delta$ are finite sets of
$\deltazero$ formulas.
\end{compactitem}
For example, \textsc{Repl} in the figure is a ``congruence rule'', capturing
that terms that are equal are interchangeable. Informally, it says that to prove
conclusion $\Delta$ from a hypothesis that includes a formula $\phi$ including
variable $t$ and an equality $t =_{\ur} u$, it suffices to 
add to the hypotheses
a copy of $\phi$ with $u$ replacing some occurrences of $t$.

A proof tree whose root is labelled by $\Theta; \; \Gamma \vdash \Delta$ witnesses that,
for any given meaning for the free variables, if all the membership relations
in $\Theta$ and all formulas in $\Gamma$ are satisfied, then there is a formula in $\Delta$
which is true.
We say that we have a proof of a single formula $\phi$ when we have a proof of $\emptyset ; \emptyset \vdash \phi$.

The proof system is easily seen to be sound:
if $\Theta; \; \Gamma \vdash \Delta$, then $\Theta; \; \Gamma \models \Delta$, where we remind the reader that
$\models$ considers all models, not just extensional ones.
It can be shown to be complete by a standard technique (a ``Henkin construction'', see the appendix).

To generate $\nrc$ definitions \emph{efficiently} from proof witnesses will require a more
restrictive proof system, in which we enforce some ordering on how proof rules can be applied,
depending on the shape of the hypotheses.
We refer to proofs in this system as \emph{focused proofs},
 the terminology being inspired by the proof search literature \cite{mmpvfocusing}.
To this end, we categorize formulas as being either
\emph{existential-leading} ($\pospolarity$) or \emph{alternative-leading} ($\negpolarity$)
according to their top-level connective. 
Only atomic formulas are both $\pospolarity$ and $\negpolarity$,
and the only
other $\pospolarity$ formulas are existentials; all the others are $\negpolarity$.%
\footnote{In the literature on focusing, these are referred to as ``positive'' and ``negative'' formulas,
but we avoid this terminology due to clashes with other uses of those terms.}
\[
\begin{array}{lcl}
\varphi^\pospolaritysuper &\eqdef& t =_\ur u \bnfalt t \neq_\ur u \bnfalt \exists x \in t. \psi \\
\varphi^\negpolaritysuper &\eqdef& t =_\ur u \bnfalt t \neq_\ur u
\bnfalt \varphi \wedge \psi
\bnfalt \varphi \vee \psi
\bnfalt \top \bnfalt \bot
\bnfalt \forall t \in b.~ \psi
\end{array}\]

Our focused proof system is shown in Figure \ref{fig:dzcalc-1sided}.
A superficial difference from Figure \ref{fig:dzcalc-2sided} is that
the focused system is ``almost $1$-sided'': $\deltazero$ formulas only
occur on the right, with only $\in$-contexts on the left. In particular,
a top-level goal $\Theta;\Gamma \proves \Delta$ in the  higher-level system
would be expressed as $\Theta \proves \neg \Gamma, \Delta$ in this system.
We will often abuse notation by referring to focused proofs
of a $2$-sided sequent $\Theta; \Gamma \proves \Delta$, considering them as
``macros'' for the corresponding $1$-sided sequent.
For example, the hypothesis of 
 the $\neq$ rule  
could be written in  $2$-sided notation as
$\Theta,  t =_\ur u \seq \alpha[u/x], \alpha[t/x], \Delta^{\pospolaritysuper}$
while the conclusion could be written as
$\Theta , t =_\ur u \seq  \alpha[t/x], \Delta^{\pospolaritysuper}$.  As with
\textsc{Repl} in the prior system, this rule is about duplicating a hypothesis with
some occurrences of  $t$ replaced by $u$.

A major aspect of the restriction, related to the terminology \emph{focused}, is that the $\exists\textsc{-R}$ rule
enforces that blocks of existentials are instantiated all at once and that
all other formulas in the context are also $\pospolarity$.
\begin{figure} 
\begin{mathpar}

\inferrule*
[left={$=$}]
{  }
{\Theta \seq x =_\ur x,\Delta}
\and

\inferrule*
  [left={$\top$}]
  {  }{\Theta \seq \top, \Delta}

\\
  \inferrule*
[left={$\neq$}%
]
{\Theta \seq t \neq_\ur u, \alpha[u/x], \alpha[t/x], \Delta^{\pospolaritysuper} \and \text{$\alpha$ atomic}}
{\Theta \seq t \neq_\ur u, \alpha[t/x], \Delta^{\pospolaritysuper}}
\\

\inferrule*
  [left={$\wedge$}]
{
\Theta \seq \phi_1, \Delta
\and
\Theta \seq \phi_2, \Delta
}{\Theta \seq \phi_1 \wedge \phi_2, \Delta}
\and
\inferrule*
  [left={$\vee$}]
{
\Theta \seq \phi_1, \phi_2,\Delta
}{\Theta \seq \phi_1 \vee \phi_2, \Delta}
  \\

\inferrule*
  [left={$\forall$}]
{\Theta, y \in b \seq \phi[y/x], \Delta \and
\text{$y$ fresh}}
{\Theta  \seq \forall x \in b. \;\phi, \Delta}
\and
\inferrule*
  [left={$\exists$}]
  {\Theta  \seq
\phi', \phi, \Delta^{\pospolaritysuper} \qquad
\text{$\phi'$ a max. spec. of $\phi$ w.r.t.~$\Theta$}
}
  {\Theta
\seq \phi, \Delta^{\pospolaritysuper}}\\
\\
\inferrule*
[left={$\times_\eta$}]
{\Theta[\tuple{x_1,x_2}/x] \seq \Delta^{\pospolaritysuper}[\tuple{x_1,x_2}/x]
\and \text{$x_1,x_2$ fresh}}
{\Theta  \seq \Delta^{\pospolaritysuper}}
\and
\inferrule*
[left={$\times_\beta$}]
{\Theta[x_i/x] \seq \Delta^{\pospolaritysuper}[x_i/x] \and i \in \{1,2\}}
{\Theta[\pi_i(\tuple{x_1,x_2})/x] \seq \Delta^{\pospolaritysuper}[\pi_i(\tuple{x_1,x_2})/x]}

\end{mathpar}
\caption{Our focused calculus for efficient generation of witnesses.
We assume that formulas $\phi^\negpolaritysuper$ (respectively contexts $\Delta^\pospolaritysuper$)
are $\negpolarity$ (respectively contain only $\pospolarity$ formulas)
and that $x,y,z$ are variables.}
\label{fig:dzcalc-1sided}
\end{figure}

Soundness is evident, since it is a special case of the proof system
above. Completeness is not as obvious, since we are restricting the proof rules.
But we can translate proofs in the more general system of Figure \ref{fig:dzcalc-2sided} into
a focused proof, however with an exponential blow-up: see the appendix for details.

Furthermore, since for $\deltazero$ formulas
equivalence over all structures is the same as equivalence over nested relations, 
a $\deltazero$ formula $\phi$ is provable exactly when $\modelssem \phi$.

\begin{myexmp} \label{ex:simplenesting}
Let us look at  how to formalize a variation of 
Example~\ref{ex:views}.
  The specification $\Sigma(\basenest,\vflat)$ includes two conjuncts 
$C_1(\basenest,\vflat)$ and $C_2(\basenest, \vflat)$.
  $C_1(\basenest, \vflat)$ states that every pair $\tuple{k,e}$ of $\vflat$ corresponds
  to a $\tuple{k, S}$ in $\basenest$ with $e \in S$:
  \[
\forall v \in \vflat ~\exists b \in \basenest. ~
  \pi_1(v)=_\ur \pi_1(b) \wedge \pi_2(v) \memmac \pi_2(b)\]

$C_2(\basenest,\vflat)$ is:
\[
 \forall b \in \basenest ~ \forall e \in \pi_2(b)  ~
  \exists v \in \vflat. ~ \pi_1(v)=_\ur \pi_1(b) \wedge \pi_2(v) =_\ur e \\
\]

Let us assume a stronger constraint, $\Sigma_{\lossless}(\basenest)$,
saying that the first component is a key and second is non-empty:
  \[
    \begin{array}{l@{~}l} &
    \forall b \in \basenest  ~ \forall b' \in \basenest. ~ \pi_1(b) =_{\ur} \pi_1(b')
  \rightarrow b \equiv b'\\
    \wedge & \forall b \in \basenest ~ \exists e \in \pi_2(b). ~ \top
  \end{array}
  \]

With $\Sigma_{\lossless}$ we can show something stronger 
than in Example \ref{ex:views}:
$\Sigma \wedge \Sigma_{\lossless}$ implicitly defines 
$\basenest$ in terms of $\vflat$. That is, the view determines the identity query,
which is witnessed by a proof of
\[
  \Sigma(\basenest,\vflat) \wedge \Sigma_{\lossless}(\basenest) 
\wedge \Sigma(\basenest',\vflat)  \wedge \Sigma_{\lossless}(\basenest')
  \imp \basenest \equiv \basenest'\]
Let's prove this informally. Assuming the premise, it is sufficient
to prove $\basenest \subseteq \basenest'$ by symmetry. So fix $\tuple{k, S} \in \basenest$.
  By the second conjunct of $\Sigma_{\lossless}(\basenest)$, we know there is $e \in S$.
  Thus by $C_2(\basenest, \vflat)$,  $\vflat$ contains the pair $\tuple{k, e}$.
  Then, by $C_1(\basenest',\vflat)$, there is a $S'$ such that $\tuple{k, S'} \in \basenest'$.
To conclude it suffices to show that $S \equiv S'$. There are two similar directions, let us
detail the inclusion $S \subseteq S'$; so fix $s \in S$.
      By $C_2(\basenest,\vflat)$, we have $\tuple{k,s} \in \vflat$.
      By $C_1(\basenest',\vflat)$ 
      there exists $S''$ such that $\tuple{k, S''} \in \basenest'$
      with $s \memmac S''$. But since we also have 
$\tuple{k, S'} \in \basenest'$, the constraint
      $\Sigma_{\lossless}(\basenest)$ implies that $S' \equiv S''$, so $s \in S'$ as desired.
\end{myexmp}

\myparagraph{Main result}
A derivation of $(\star)$  in our proof system
will be referred to as a \emph{witness} to the implicit definability
of $o$ in terms of $\vv i$ up to extensionality.
With these definitions, we now state formally our main result,
the  effective version of Theorem \ref{thm:bethnrcmodeltheoretic}:

\begin{restatable}[Effective implicit to explicit for nested data]{theorem}{thmmainset}
\label{thm:mainset}
Given a 
witness for an implicit definition of $o$ in terms of
$\vv i$ up to extensionality relative to $\deltazero$
$\varphi(\vv i, \vv a, o)$,
one can compute  $\nrc$ expression $E$ such that for any $\vv i$, $\vv a$ and $o$,
if  $\varphi(\vec i, \vv a , o)$ then $E(\vv i)=o$.
Furthermore, if the witness is focused, this can be done in polynomial time.
\end{restatable}

\myparagraph{Application to views and queries}
We have a consequence for rewriting queries over views.
Consider a query given by
 $\nrc$ expression $E_Q$ over inputs $\vec B$ and
$\nrc$ expressions $E_{V_1} \ldots E_{V_n}$ over $\vec B$.
$E_Q$ is \emph{determined} by 
$E_{V_1} \ldots E_{V_n}$, if
every two nested relations (finite or infinite) interpreting $\vec B$
that agree on the output of each $E_{V_i}$
agree on the output of $E_Q$. 
An $\nrc$ rewriting of $E_Q$ in terms of $E_{V_1} \ldots E_{V_n}$ is
an expression $R(V_1 \ldots V_n)$ such that for any nested relation
$\vec B$, if we evaluate each $E_{V_i}$ on $\vec B$ to obtain $V_i$ and
evaluate $R$ on the resulting $V_1 \ldots V_n$, we obtain $Q(\vec B)$.

Given $E_Q$ and $E_{V_1} \ldots E_{V_n}$, let
$\Sigma_{\vec V, Q}(\vec V, \vec B, Q, \ldots)$ conjoin the input-output specifications, as defined in Section \ref{sec:prelims},
for  $E_{V_1} \ldots E_{V_n}$ and $E_Q$. This formula has variables
$\vec B, V_1 \ldots V_n, Q$ along with auxiliary variables for subqueries.
 A proof witnessing determinacy of $E_Q$  by
$E_{V_1} \ldots E_{V_n}$, is a proof that $\Sigma_{V,Q}$ implicitly
defines $Q$ in terms of $\vec V$.

\begin{corollary} \label{cor:effdeterminacy} From a witness that a set of $\nrc$ views $\vec V$ determines
an $\nrc$ query $Q$, we can produce an $\nrc$  rewriting
of $Q$ in terms of $\vec V$. If the witness is focused, this can be done
in $\ptime$.
\end{corollary}
The notion of determinacy of a query over views
relative  to a $\deltazero$ theory (e.g. the key constraint in Example \ref{ex:views}) is a straightforward generalization of the definitions above, and
Corollary \ref{cor:effdeterminacy} extends to this setting.

In the case where we are dealing with flat relations, the effective version
is well-known: see Toman and Weddell's \cite{tomanweddell}, and the discussion
in \cite{franconisafe,interpbook}.

We emphasize that the result involves equivalence up to extensionality,
which underlines the distinction from the classical Beth theorem. If
we wrote  out implicit definability up to extensionality as an entailment involving
two copies of the signature, we would  run into problems in applying the
standard proof of Beth's theorem.

\section{Tools for the main theorem}
  \label{sec:tools}

\myparagraph{Interpolation}
The first tool for our main theorem will
be an interpolation result. Informally, such results
say that if we have an entailment involving
two formulas, a ``left'' formula $\phi_L$ and a ``right'' formula $\phi_R$, we can get an ``explanation''
for the entailment that factors through an expression only involving
non-logical symbols (in our case, variables) that
are common to $\phi_L$ and $\phi_R$.

\begin{theorem}
\label{thm:interpolationdeltafocus}

Let $\Theta$ be an $\in$-context and $\Gamma, \Delta$ finite sets of $\deltazero$ formulas.
Then from any proof of $\Theta; \; \Gamma \vdash \Delta$.
we can compute in linear time a $\deltazero$ formula $\theta$ with $\FV(\theta) \subseteq \FV(\Theta, \Gamma) \cap \FV(\Delta)$
such that $\Theta; \; \Gamma \vdash \theta$ and $\emptyset; \; \theta \vdash \Delta$
\end{theorem}

The $\theta$ produced by the theorem is a \emph{Craig interpolant}.
Craig's interpolation theorem \cite{craig57interp}
states that when $\Gamma \proves \Delta$ with $\Gamma, \Delta$  in first-order logic, such
 a $\theta$ exists in
first-order logic.
Our variant  states one can find $\theta$ in $\deltazero$ efficiently from
a proof of the entailment in either of our $\deltazero$ proof systems.
We have stated the result for the $2$-sided system. It holds also for the $1$-sided
focused system, where the partition of the formulas into left and right of the proof
symbol is arbitrary.
The argument
is induction on proof length, roughly following prior
interpolation algorithms \cite{smullyan}.

\myparagraph{Some admissible rules}
As we mentioned earlier, our focused proof system is extremely low-level, and
so it is convenient to have higher-level proof rules as macros. We formalize
this below.

\begin{definition}
  \label{def:polyadm}
A rule with premise $\Theta' \vdash \Delta'$ and conclusion $\Theta \vdash \Delta$
  \[ \dfrac{\Theta' \vdash \Delta'}{\Theta \vdash \Delta} \]
is \emph{(polytime) admissible} in a given calculus if a proof of
the conclusion $\Theta \vdash \Delta$ in that calculus can be computed from a
proof of the premise $\Theta' \vdash \Delta'$ (in polynomial time).
\end{definition}

Up to rewriting the sequent to be one-sided, all the rules in
Figure~\ref{fig:dzcalc-2sided} are polytime admissible in the focused calculus.
Our main theorem will rely on the polytime admissibility within
the focused calculus of
additional rules that involve chains of existential quantifiers.
To state them, we need
to introduce a generalization of bounded quantification:
``quantifying over subobjects of a variable''. 
For every type $T$, define a set of words over the three-letter 
alphabet $\{1,2,m\}$ of \emph{subtype occurrences of $T$} inductively as follows:
\begin{itemize}
\item The empty word $\varepsilon$ is a subtype occurence of any type
\item If $p$ is a subtype occurence of $T$,
$mp$ is a subtype occurence of $\sett(T)$.
\item If $i \in \{0,1\}$ and $p$ is a subtype occurence of $T_i$,
$ip$ is a subtype occurence of $T_1 \times T_2$.
\end{itemize}
Given subtype occurence $p$ and  quantifier symbol $\mathbf{Q} \in \{\forall, \exists\}$,
define the notation $\mathbf{Q} \; x \in_p t. \phi$ by induction on $p$:
\begin{itemize}
\item $\mathbf{Q} \; x \in_m t. \phi$ is $\mathbf{Q} \; x \in t$
\item $\mathbf{Q} \; x \in_{mp} t. \phi$ is $\mathbf{Q} \; y \in t.
\mathbf{Q} \; x \in_p y. \phi$ with $y$ a fresh variable
\item $\mathbf{Q} \; x \in_{ip} t. \phi$ is $\mathbf{Q} \; x \in_p \pi_i(t). \phi$
  when $i \in \{1,2\}$.
\end{itemize}

Now we are ready to state the results we need on admissibility in the body
of the paper, referring in each case to the focused calculus. Some
further routine rules are used in the appendices.
The first states that if we have proven that there exists a
 subobject of $o'$ equivalent to object $r$, then we can prove
that for each element $z$ of $r$ there is a corresponding equivalent
subobject $z'$ within $o'$.

\begin{restatable}{lemma}{effnestedproofdown}
  \label{lem:effnestedproof-down}
Assume $p$ is a subtype occurence for the type of the term~$o'$. The following is polytime admissible
\[\dfrac{\Theta \vdash \Delta,\; \exists {r' \inq_p o'}. \; r \equiv_{\sett({T'})} r'}{
\Theta, z \in r \vdash \Delta, \; \exists {z' \inq_{mp} o'}. \; z \equiv_{{T'}} z'
  }\]
Furthermore, the size of the output proof is at most the size of the input proof.
\end{restatable}

The second rule states that we can move between an equivalence of $r, r'$
and a universally-quantified  biconditional between memberships in $r$
and $r'$. Because we are dealing with $\deltazero$ formulas, the universal quantification
has to be bounded by some additional variable $a$.
\begin{restatable}{lemma}{equivsettoequivalence}

  \label{lem:equivsettoequivalence}
  The following is polytime admissible (where $p$ is a subtype occurence of the type of $o'$)
  \[ 
\dfrac{ \Theta \seq \Delta, \exists r' \inq_p o'. \; r \equiv_{\sett({T'})} r'}{ \Theta \seq \Delta, \exists r' \inq_p o'. \forall z \inq a. \; z \memmac
  r \leftrightarrow z \memmac r'}
  \]
\end{restatable}

\myparagraph{The $\nrc$ Parameter Collection Theorem}
Our last tool is a kind of interpolation result
connecting $\deltazero$ formulas and $\nrc$:

\begin{restatable}[$\nrc$ Parameter Collection]{theorem}{mainflatcasebounded}
\label{thm:mainflatcasebounded}
Let $L$, $R$ be sets of variables with $C = L \cap R$ and
\begin{itemize}
\item $\phi_{L}$ and $\lambda(z)$ $\deltazero$ formulas over $L$
\item $\phi_{R}$ and $\rho(z,y)$ $\deltazero$ formulas over $R$
\item $r$ a variable of $R$ and $c$ a variable of $C$.
\end{itemize}

Suppose that we have a proof of
\begin{align*}
\phi_L \wedge \phi_R ~  \imp ~ \exists y \in_p r  ~ \forall z \in c  ~
(\lambda(z) \equi \rho(z, y))
\end{align*}

Then one may compute in polynomial time an $\nrc$ expression $E$ with
free variables in $C$
such that
\begin{align*}
\phi_L 
\wedge \phi_R \imp
  \{ z \in c \mid \lambda(z)\} \in E
\end{align*}
\end{restatable}

If $\lambda$ was a ``common formula'' ---
one using only variables in $C$
---  then the nested
relation
$ \{ z \in c \mid \lambda(z)\}$ would be definable over $C$ in $\nrc$ via $\deltazero$-comprehension.
Unfortunately $\lambda$ is a ``left formula'', possibly with variables outside of $C$.
Our hypothesis is that it is equivalent to a 
``parameterized right formula'': a formula with variables in $R$ and parameters that lie below them.
Intuitively,
this can happen only if $\lambda$ can be rewritten to a formula $\rho'(z, x)$
with variables of $C$ and a distinguished $c_0 \in C$ such that
\begin{align*}
  \phi_L \wedge \phi_R ~  \imp ~ \exists x \in_p  c_0 ~ \forall z \in c  ~
(\lambda(z) \equi \rho'(z, x))
\end{align*}

And if this is true, we can use an $\nrc$ expression
over $C$ to define a set that will contain
the correct ``parameter'' value $x$ defining $\lambda$.
From this we can define a set containing the nested
relation $ \{ z \in c \mid \lambda(z)\}$.
A formalization of this rough intuition -- ``when left formulas are equivalent
to parameterized right formulas, they are equivalent to parameterized
common formulas'' -- can be found in the appendix.

\myparagraph{Sketch of the proof of Theorem \ref{thm:mainflatcasebounded}}
To get the desired conclusion, we need to prove a more general statement by induction
over proof trees. Besides making the obvious generalization to handle 
two sets of
formulas instead of the particular formulas  $\phi_L$ and $\phi_R$, as well as some corresponding left and right $\in$-contexts,
that may appear during the proof, we need to additionally generate a new formula~$\theta$ that only uses common variables, which can replace $\phi_R$ in the conclusion.
This is captured in the following lemma:

\begin{restatable}{lemma}{lemnrcparamcoll}
  \label{lem:nrcparamcoll}
Let $L$, $R$ be sets of variables with $C = L \cap R$ and
\begin{itemize}
\item $\Delta_{L}, \lambda(z)$ a set of $\deltazero$ formulas over $L$
\item $\Delta_{R}, \rho(z,y)$ a set of $\deltazero$ formulas over $R$
\item $\Theta_L$ (respectively $\Theta_R$) a $\in$-context over $L$
(respectively over $R$)
\item $r$ a variable of $R$ and $c$ a variable of $C$.
\end{itemize}

Suppose that we have a proof tree with conclusion
\begin{align*}
\Theta_L, \Theta_R \vdash \Delta_L, \Delta_R, \exists y \in_p r  ~ \forall z \in c  ~
(\lambda(z) \equi \rho(z, y))
\end{align*}

Then one may compute in polynomial time an $\nrc$ expression $E$
and a $\deltazero$ formula $\theta$ using only variables from $C$
such that
\begin{align*}
\Theta_L ~~\modelssem~~ \Delta_L, \theta \vee
  \{ z \in c \mid \lambda(z) \} \in E
\quad \text{and} \quad
\Theta_R \modelssem \Delta_R, \neg \theta
\end{align*}
\end{restatable}

In this induction over the size of the proof of 
\[
\Theta_L, \Theta_R \vdash \Delta_L, \Delta_R, \exists y \in_p r  ~ \forall z \in c  ~
(\lambda(z) \equi \rho(z, y))\]
we make a case distinction according to which
rule is applied last. Many of the cases use standard techniques, we focus
in the body of the paper on the most novel case.

Let us write $\cG$ for the formula
$\exists y \in_p r  ~ \forall z \in c.  ~ (\lambda(z) \equi \rho(z, y))$
and $\Lambda$ for the set $\{ z \in c \mid \lambda(z)\}$.
The most difficult inductive case
is where the last rule applied is $\exists$, and where
$\cG$ is the main formula, i.e., when the last step,
where we pick some witness $w$ for $y$ using the $\exists$ rule, has shape
\[
\small  \dfrac{
\Theta_L, \Theta_R \vdash \Delta_L, \Delta_R,
\forall z \in c.  ~ (\lambda(z) \equi \rho(z, w)),
\cG
  }{
 \Theta_L, \Theta_R \vdash \Delta_L, \Delta_R, \cG
}
\]
Now notice that, due to our restriction on the $\exists$ rule, 
all formulas in $\Delta_L$ and $\Delta_R$ are $\pospolarity$. Therefore, the only possible shape of
the proof, when reasoning backward from the goal,
is  via  successive applications of the $\forall$, $\wedge$ rules.
This means we have two strict subproofs with respective conclusions
      \[
        \begin{array}{ll}
          &{\Theta_L, z \in c}, {\Theta_R} \vdash {\lambda(z), \Delta_L}, {\neg \rho(z,w), \Delta_R},\cG\\
          \text{and} & {\Theta_L, z \in c}, {\Theta_R} \vdash {\neg\lambda(z),\Delta_L}, {\rho(z,w), \Delta_R},
        \cG \end{array}
\]
Applying the inductive hypothesis, we obtain $\nrc$ expressions $E_1\IH$, $E_2\IH$
and formulas $\theta_1\IH, \theta_2\IH$ which contain free variables in $C \cup \{z\}$ such that all of the following hold
\begin{align*}
&
  {\Theta_L, z \in c} \modelssem \phantom{\neg}{\lambda(z), \Delta_L}, \theta_1\IH \vee \Lambda  \in E_1\IH
 \\
  \text{and} \qquad &
  {\Theta_L, z \in c} \modelssem {\neg \lambda(z), \Delta_L}, \theta_2\IH \vee \Lambda  \in E_2\IH
   \\
  \text{and} \qquad &
  \phantom{z \in c,}{\Theta_R} \modelssem {\neg\rho(z,w),\Delta_R}, \neg \theta_1\IH
   \\
  \text{and} \qquad &
  \phantom{z \in c,} {\Theta_R} \modelssem \phantom{\neg}{\rho(z,w), \Delta_R}, \neg \theta_2\IH
\end{align*}
With this in hand, we set
\begin{align*}
  &   \theta \eqdef \exists z \in c. ~ \theta_1\IH \wedge \theta_2\IH \\
\text{and} \qquad &       E \eqdef \left\{\left\{ z \in c \mid \theta_2\IH\right\}\right\} ~\cup ~ \bigcup\left\{ E_1\IH \cup E_2\IH \mid z \in c \right\}
\end{align*}
Note in particular that the free variables of $E$ and $\theta$ are contained in $C$, since we bind $z$.
The verification that $E$ and $\theta$ suffice is routine; see the appendix.

Notice that the structure of our construction is: \begin{inparaenum}
\item by induction,
we get  $\nrc$ expressions $E_i\IH$ satisfying the required property over nested relations;
\item using our witness
proof, we get a formula $\theta$ satisfying some invariant
 over all models, and thus
in particular on all nested relations; \item we  perform a construction
combining $\theta$ and $E_i\IH$,
and reason with the na\"ive semantics of $\nrc$ to argue for correctness.
\end{inparaenum}
We use this reasoning template in the proof of our main theorem as well.

\section{Proof of the main result} \label{sec:proof:ain}
We now turn to the proof of our main result:
\thmmainset*

We have as input a proof of
\[ \phi(\vv i, \vv a, o) \wedge \phi(\vv i, \vv a', o') \imp o \equiv_T o'\]
and we want an $\nrc$ expression $E(\vv i)$ such that
\[ \phi(\vv i, \vv a, o) \modelssem E(\vv i) \equiv_T o\]

This will be a consequence of the following theorem.
\begin{restatable}{theorem}{mainsetinter}
\label{thm:mainsetinter}
Given $\deltazero$ $\varphi(\vv i, \vv a, o)$ and $\psi(\vv i, \vv b, o')$
together with a focused 
proof with conclusion
  \[\Theta(\vv i, \vv a, r); \; \varphi(\vv i, \vv a, r), \psi(\vv i, \vv b, o') \vdash \exists {r' \in_p o'}. \; r \equiv_T r'\]
we can compute in polytime an $\nrc$ expression $E(\vv i)$ such that
  \[ \Theta(\vv i, \vv a, r); \; \varphi(\vv i, \vv a, r), \psi(\vv i, \vv b, o') ~\modelssem~ r \in E(\vv i)\]
\end{restatable}

That is, we can find an $\nrc$ query that ``collects answers''.
Assuming Theorem \ref{thm:mainsetinter}, let's prove the main result.

\begin{proof}[Proof of Theorem \ref{thm:mainset}]
  We assume $o$ has a set type, deferring the simple product and Ur-element cases (the latter using $\nrcget$) to the appendix.
Fix  an implicit definition of $o$ up to extensionality relative
to $\varphi(\vv i, \vv a, o)$ and a focused proof of
  \[\varphi(\vv i, \vv a, o) \; \wedge \; \varphi(\vv i, \vv a', o') ~\proves~ o \equiv_{\sett(T)} o' \]
We can apply a simple variation of Lemma \ref{lem:effnestedproof-down} for the ``empty path'' $p$
to obtain a focused derivation of
  \begin{align}
r \in o ; \; \varphi(\vv i, \vv a, o), \varphi(\vv i, \vv a', o')  
    \vdash~ \exists r' \inq o' \; r \equiv_T r' \label{algn:ent1}
  \end{align}
Then applying Theorem \ref{thm:mainsetinter} gives
a $\nrc$ expression $E(\vv i)$ such that
\[
\varphi(\vv i, \vv a, o) \wedge r \memmac o \wedge \; \varphi(\vv i, \vv a', o') ~  \modelssem ~ r \in E(\vv i)
\]
Thus, the object determined by $\vv i$ is always contained in $E(\vv i)$.
Coming back to (\ref{algn:ent1}), we can obtain a derivation of
  \[
    r \in o; \; \varphi(\vv i, \vv a, o) \vdash \varphi(\vv i, \vv a', o')  \rightarrow \exists r' 
  \inq o'\; r \equiv_T r'\]
and applying interpolation (Theorem \ref{thm:interpolationdeltafocus}) to that gives
  a $\deltazero$ formula $\kappa(\vec i,  r)$ such that the following are valid
\begin{align}
  r \in o \wedge \varphi(\vv i, \vv a, o) \imp \kappa(\vv i, r) \label{algn:inter1}\\
  \kappa(\vv i, r) \wedge \varphi(\vv i, \vv a', o') \imp \exists r' \inq o'. \; r \equiv_T r' \label{algn:inter2}
\end{align}

  We claim that $E_\kappa(\vv i) = \left\{ x \in E(\vv i) \mid \kappa(\vv i, x)\right\}$ is the
desired $\nrc$ expression. To show this, assume $\varphi(\vv i, \vv a, o)$ holds.
  We know already that $o \subseteq E(\vv i)$ and, by~(\ref{algn:inter1}), every $r \in o$ satisfies
  $\kappa(\vv i, o)$, so $o \subseteq E_\kappa(\vv i)$.
  Conversely, if $x \in E_{\kappa}(\vv i)$, we have $\kappa(\vv i,x)$, so by~(\ref{algn:inter2}),
  we have that $x \in o$, so $E_\kappa(\vv i) \subseteq o$. So $E_\kappa(\vv i) = o$, which concludes the proof.
\end{proof}

We now turn to the proof of Theorem \ref{thm:mainsetinter}.

\myparagraph{Proof of Theorem \ref{thm:mainsetinter}}
We prove the theorem %
by induction over the type $T$. 
We only prove the inductive step
for set types: the inductive case for products are straightforward.

For $T = \ur$, the base case of the induction, it is clear that we can take for $E$ an expression computing the set of all $\ur$-elements
in the transitive closure of $\vv i$. This can clearly be done in $\nrc$.

So now, we assume $T = \sett({T'})$ and that Theorem~\ref{thm:mainsetinter} holds up to ${T'}$.
We have a focused derivation of
\begin{align}
\Theta; \; \varphi(\vv i, r), \; \psi(\vv i, o') ~\vdash~ \exists {r' \in_p o'}. \; r \equiv_{\sett(T')} r'
  \label{algn:inputder}
\end{align}
omitting the additional variables for brevity. 

From our input derivation, we can easily see that each element
of $r$ must be equivalent to some element below $o'$.
This is reflected by Lemma~\ref{lem:effnestedproof-down}, which allows us to
efficiently compute a proof of
\[\Theta, z \in r; \; \varphi(\vv i, r), \; \psi(\vv i, o') ~\vdash~ \exists z' \inq_{mp} o'. \; z \equiv_{T'} z'\]
We can then apply the inductive hypothesis of our main theorem at sort ${T'}$,
which is strictly smaller than $\sett({T'})$.
This yields a $\nrc$ expression $E\IH(\vv i)$
of type $\sett({T'})$ such that 
\[\Theta, z \in r; \; \varphi(\vv i,r), \psi(\vv i, o') \modelssem z \in E\IH(\vv i)\]
That is, our original hypotheses
entail $r \subseteq  E\IH(\vv i)$.

Thus, we have used the inductive hypothesis to get a ``superset expression''. But
now we want an expression that has $r$ as an element. We will do this
by  unioning a collection of definable subsets of $E\IH(\vv i)$.
To get these, we come back to our input derivation (\ref{algn:inputder}).
By~\Cref{lem:equivsettoequivalence}, we can efficiently compute a derivation of
\[\Theta; \; \varphi(\vv i, r), \psi(\vv i, o') \vdash \exists r' \inq_p
o'\,  \forall z \inq a  \;
(z \memmac r \equi z \memmac r')\]
where we take $a$ to be a fresh variable of sort $\sett({T'})$. 
Now, applying our $\nrc$ Parameter Collection result (Theorem~\ref{thm:mainflatcasebounded})
we obtain a $\nrc$ expression $E^{\mathrm{coll}}(\vv i, a)$ satisfying
\begin{align*}
\Theta; \; \varphi(\vv i, r), \psi(\vv i, o') ~\modelssem 
  a \cap r  \in E^{\mathrm{coll}}(\vv i, a)
\end{align*}
Now, recalling that we have $r \subseteq E\IH(\vv i)$ and instantiating $a$ to be
$E\IH(\vv i)$, we can conclude that
\begin{align*}
\Theta; \; \varphi(\vv i, r), \psi(\vv i, o') ~\modelssem 
  r  \in E^{\mathrm{coll}}(\vv i, E\IH(\vv i))
\end{align*}
Thus we can take $E^{\mathrm{coll}}(\vv i, E\IH(\vv i))$ as an explicit definition. \qed
\section{Discussion and future work} \label{sec:conc}

Our effective nested Beth result implies that whenever a set of $\nrc$ views determines
an $\nrc$ query, the query is rewritable over the views in $\nrc$.
Further, from a proof witnessing determinacy in our proof system, we can efficiently generate the rewriting. Our result applies to  a setting
where we have determinacy with respect to constraints and views, as in
Example \ref{ex:views}, or to general $\deltazero$ implicit definitions that
may not stem from views.

In terms of impact on databases, a crucial limitation of our work is that we do not yet
know how to find the proofs. In the case of relational data,
we know of many ``islands of decidability'' where proofs of determinacy can be found effectively -- e.g. 
for views and queries in guarded logics \cite{gnfjsl}.
But it remains open to find similar decidability results for views/queries in fragments of $\nrc$.

It is possible to use our proof system without full automation -- simply search for a proof, and then
when one finds one, generate the rewriting. We have had some success with this approach in the relational setting,
where standard theorem proving technology can be applied \cite{usijcai17}. But for the proof systems proposed
here, we do not have either our own theorem prover or a reduction to a system that has been implemented in the past.
The need to find proofs automatically is pressing since
our system is so low-level that it is difficult to do proofs by
hand. Indeed, a formal proof of implicit definability for Example \ref{ex:views}, 
or even the simpler Example~\ref{ex:simplenesting},
 would come  to several pages. 

The implicit-to-explicit methodology requires
a proof of implicit definability, which implies implicit definability over all instances, not just finite ones.
This requirement is necessary:
one cannot hope to convert implicit definitions over finite instances
to explicit  $\nrc$ queries, even ineffectively.  We do not believe that this is a limitation in practice. See the appendix for
details.

Our key proof tool was the $\nrc$ Parameter Collection theorem,
Theorem \ref{thm:mainflatcasebounded}. There is an intuition behind this
theorem  that concerns a general  setting,
where  we have a first-order theory $\Sigma$ that factors into a conjunction
of two formulas $\Sigma_L \wedge \Sigma_R$, and from this we have
a notion of a ''left formula'' (with predicates from $\Sigma_L$),
a ''right formula'' (predicates from $\Sigma_R$), and a ``common formula''
(all predicates occur in both $\Sigma_L$ and $\Sigma_R$). Under the
hypothesis that a left formula $\lambda$ is 
definable from a right formula with parameters, we can conclude that
the left formula must actually be definable from a common formula
with parameters: see the appendix for a formal version and the
corresponding proof.

Our work contributes to the broader topic of proof-theoretic vs model-theoretic techniques for interpolation and definability theorems. For Beth's theorem, 
there are reasonably short model-theoretic \cite{Lyndon59, ck} and proof-theoretic arguments \cite{craig57beth,fittingbook}.
In database terms, you can argue semantically that relational
algebra
is complete
for  rewritings of queries determined by views, and producing
a rewriting from a proof of determinacy is not   that difficult.
But for a number of results  on definability proved in the 60's and 70's 
\cite{changbeth,makkai,kueker,gaifman74}, there
are short model-theoretic arguments, but no proof-theoretic ones.
For our $\nrc$ analog of Beth's theorem, the situation is more similar
to the latter case:
the model-theoretic proof of completeness \cite{benediktpradicpopl}
is relatively short and elementary, but generating explicit
definitions from proofs is much more challenging.
We hope that our results and  tools represent a step towards
providing  effective versions, and towards understanding the relationship between
model-theoretic and proof-theoretic arguments.

\begin{acks}
This work was funded by EPSRC grant EP/T022124/1 and by the Deutsche
Forschungsgemeinschaft (DFG, German Research Foundation) --
Project-ID~457292495.
\end{acks}

\bibliographystyle{ACM-Reference-Format}
\bibliography{obj}

\newpage
\onecolumn
\appendix
\section{Comparison to the situation with finite instances}
Our result concerns a specification $\Sigma(\vec I, O \ldots)$
such that $\vec I$ implicitly defines $O$. This can be defined
``syntactically'' -- via the existence of a proof (e.g. in our own proof system).
Thus, the class of queries that we deal with could be called
the ``provably implicitly definable queries''.
The same class of queries can also be defined semantically, and this is how implicitly
defined queries are often presented. But in order to be equivalent to the
proof-theoretic version, we need the implicit definability of the object
$O$ over $\vec I$ to holds considering
all nested relations $\vec I, O \ldots$, not just finite ones.
Of course, the fact that when you phrase the property
 semantically requires referencing unrestricted instances  does not mean that
our results depend on the existence of infinite nested relations.

Discussion of finite vs. unrestricted instances appears in many other papers
(e.g. \cite{benediktpradicpopl}). And the results in this submission do not raise
any new issues with regard to the topic. But we discuss what
happens if we take the obvious analog of the semantic definition, but
using only finite instances.
Let us say that
a $\deltazero$ specification $\Sigma(\vec I, O \vec A)$  is \emph{implicitly
defines $O$ in terms of $\vec I$ over finite instances} if for any
\emph{finite} nested relations $\vec I, O, \vec A, O', \vec A'$, if
$\Sigma(\vec I, O, \vec A) \wedge \Sigma(\vec I, O', \vec A')$ holds,
then $O=O'$.
If this holds, then $\Sigma$ defines a query, and we call
such a query \emph{finitely implicitly definable}.

This class of queries is reasonably well understood, and 
we summarize what is known about it:

\begin{itemize}
\item \emph{ Can finitely implicitly definable queries always be defined in $\nrc$?}
The answer is a resounding ``no'': one can implicitly define the powerset query over
finite nested relations. Bootstrapping this, one can define iterated
powersets, and show that the expressiveness of implicit definitions
is the same as queries in $\nrc$ enhanced with powerset -- a query
language with non-elementary complexity. Even in the setting of relational
queries, considering only finite instances leads to a query class
that is not known to be in PTIME \cite{kolaitisimpdef}.

\item \emph{Can we generate explicit definitions from specifications $\Sigma$, given
a proof that $\Sigma$   implicitly defines $O$ in terms
of $\vec I$ over finite instances?} It depends on what you mean by ``a proof'', but
in some sense there is no way to make sense of the question:
there is no complete proof system for such definitions. This follows
from the fact that the set of finitely
implicitly definable queries is not computably enumerable.

\item \emph{Is sticking to specifications $\Sigma$ that are implicit definitions
over all inputs -- as we do in this work -- too strong?} Here the answer can
not be definitive. But we know of no evidence that this is too restrictive
 in practice.
Implicit specifications suffice to specify any $\nrc$ query. And
the  answer to the first question above says that if we modified
the definition in the obvious way to get a larger class,  we would allow specification of
 queries that do not admit efficient evaluation. The answer to the second
question above says that we do not have a witness to membership in this
larger class.
\end{itemize}

\section{Capturing $\nrc$ expressions with $\deltazero$ formulas}

In the body of the paper we mentioned that for every $\nrc$ expression $E(\vec i)$,
we can create a $\deltazero$ expression $\phi_E(\vec i, o)$
 such that $E(\vec i)=o$ exactly when  $\phi_E(\vec i,  o)$
holds. These were called ``input-output specifications''.
This conversion is needed to reason about determinacy
of queries by views in our formalism.  If we start with the
views and queries in $\nrc$, we can use this transformation to get a corresponding
$\deltazero$ specification.

Note that $o =E$ is itself an $\nrc$ expression of Boolean  type.
The result then follows from the fact that every  $\nrc$ expression $E(\vec w)$ of Boolean type
can be converted to a $\deltazero$ formula $\phi(\vec w)$.
This conversion be done in polynomial time
for the ``composition-free'' syntax for $\nrc$ \cite{koch} mentioned
briefly in the body: in composition-free $\nrc$, we  restrict
$\bigcup \{E | x \in E'\}$ so that $E'$ must be a variable. One can normalize
every expression to be of this form. The normalization  is exponential,
and under complexity-theoretic hypotheses one cannot do better \cite{koch}.
The conversion from $\nrc$ Boolean expressions to $\deltazero$
 is given in full detail in \cite{benediktpradicpopl}, although it is very similar
to  results on simulating $\nrc$ with flat relations given in prior work (e.g. \cite{simulation}).

\section{Completeness of proof systems} \label{sec:app-complete}
In the body of the paper we mentioned that the completeness
of the proof systems is argued using a standard method.
We outline this for the higher-level system in Figure \ref{fig:dzcalc-2sided}.

One has a sequent $\Theta;\; \Gamma \seq \Delta$ that is not provable.
We want to construct a countermodel: one that satisfying all the formulas in $\Theta$ and $\Gamma$
but none of the formulas in $\Delta$.
We construct a  tree with $\Theta;\; \Gamma \seq \Delta$ at the root by iteratively
applying applicable inference rules in reverse: in ``proof search mode'', generating
subgoals from goals. We apply the  rules whenever possible in a given order,
enforcing some fairness constraints: a rule that is active must be 
eventually applied along an infinite search,
and if a choice of terms must be made (as with the $\exists$-R rule),
all possible choices of terms are eventually made in an application of the rule.
For example, if we have a disjunction $\rho_1 \vee \rho_2$ on the right, we may
immediately ``apply $\vee$-R'': we generate a subgoal where on the
right hand side we add $\rho_1, \rho_2$.
Finite branches leading to
a sequent that does not match the conclusion of any rule or
axiom are artificially extended to infinite branches by
repeating the topmost sequent.

By assumption, this process does not produce a proof,
 and thus we have
an infinite  branch $b$ of the tree. We create a model $M_b$ whose
elements are the variables that appear on the branch, where an element
inherits the type of its variable. The memberships
correspond to the  membership atoms that appear on the left of any
sequent in $b$, and also the atoms that appear negated on the right hand
side of any sequent.

We claim that $M_b$ is the desired countermodel. It suffices to show
that for every sequent $\Theta; \; \Gamma, \seq \Delta$ in $b$,
$M_b$ is a counterexample to the sequent: it satisfies the conjunction of
formulas on the left and none of the formulas on the right.
We prove this by induction on the logical complexity of the formula.
For atoms it is immediate by construction.
Each inductive step will involve the assumptions about inference rules
not terminating proof search.
For example, suppose for some sequent $b_i$ in $b$ of the above form,
 $\Delta$ contains $\rho_1 \vee \rho_2$, we want
to show that $M_b$ satisfies $\neg (\rho_1 \vee \rho_2)$. 
But we know that in some successor of $b$, we would have applied
$\vee$-R, and thus have a descendant with $\rho_1, \rho_2$ within
the left. By induction $M_b$ satisfies $\neg \rho_1$ and $\neg \rho_2$.
Thus $M_b$ satisfies $\neg (\rho_1 \vee \rho_2)$ as desired.
The other connectives and quantifiers are handled similarly.

\section{$\deltazero$ interpolation: proof sketch of Theorem \ref{thm:interpolationdeltafocus}}

We recall the statement:

\medskip

Let $\Theta$ be an $\in$-context and $\Gamma, \Delta$ finite sets of $\deltazero$ formulas.
Then from any proof of $\Theta; \; \Gamma \vdash \Delta$
we can compute in linear time a $\deltazero$ formula $\theta$ with $\FV(\theta) \subseteq \FV(\Theta, \Gamma) \cap \FV(\Delta)$,
such that $\Theta; \; \Gamma \vdash \theta$ and $\emptyset; \; \theta \vdash \Delta$

\medskip

Recall also that we claim this for both the higher-level $2$-sided system
and the $1$-sided system, where
the $2$-sided syntax is a ``macro'':  $\Theta; \; \Gamma \vdash \Delta$
is a shorthand for $\Theta \vdash \neg \Gamma, \Delta$, where $\neg \Gamma$
is itself a macro for dualizing connectives. 
Thus in the $1$-sided version, we are arbitrarily classifying some of the 
$\deltazero$ formulas as Left and the others as Right, and our interpolant
must be common according to that partition.

We stress that there are no new ideas needed in proving Theorem  \ref{thm:interpolationdeltafocus} --- unlike
for our main tool, the Parameter Collection Theorem, or our final result.
The construction for Theorem  \ref{thm:interpolationdeltafocus} proceeds exactly as in prior interpolation theorems 
for similar calculi \cite{takeuti:1987,bpt,smullyan}. Similar constructions
are utilized in works for query reformulation
in databases, so for a presentation geared towards
a database audience one can check   \cite{tomanweddell} or the later
\cite{ interpbook}.

We explain the argument for the higher-level $2$-sided system.
We prove a more general statement, where we partition the context and
the formulas 
on both sides of $\vdash$ into Left and Right.
So we have 
\[
\Theta_L \Theta_R;  \Gamma_L, \Gamma_R \vdash \Delta_L, \Delta_R
\]
And our inductive invariant is that
we will compute in linear time a $\theta$  such that:
\begin{align*}
\Theta_L; \Gamma_L \vdash  \theta, \Delta_L \\
\Theta_R; \Gamma_R, \theta \proves \Delta_R
\end{align*}
And we require that $FV(\theta) \subseteq FV(\Theta_L, \Gamma_L, \Delta_L) \cap
 FV(\Theta_R,\Gamma_R, \Delta_R)$.
This generalization 
is used to handle the negation rules, as we explain below.

We proceed by induction on the depth of the proof tree.

One of the base cases is where we have a trivial proof tree,
which uses rule (Ax) to derive:
\[
\Theta; \Gamma, \phi \proves  \phi, \Delta
\]
We do a case distinction on where the occurrences of $\phi$ sit in our partition.
Assume the occurrence on the left is in $\Gamma_L$ and the occurrence
on the right is in $\Delta_R$. 
Then we can take our interpolant $\theta$ to be $\phi$.
Suppose the occurrence on the left is $\Gamma_L$ and the occurrence on the right
is in $\Delta_L$. Then we can take $\theta$ to be $\bot$.
The other base cases are similar.

The inductive cases for forming the interpolant will work
``in reverse'' for each proof rule. That is,  if we used an inference rule
to derive sequent $S$ from sequents $S_1$ and $S_2$, we will partition
the sequents $S_1$ and $S_2$ based on the partition of $S$.
We will then apply induction to our partitioned
sequent for $S_1$ to get an interpolant $\theta_1$, and also
apply induction to our partitioned version
of $S_2$ to get an interpolant $\theta_2$.
We
then put them together to get the interpolant for the partitioned sequent
$S$. This ``putting together''
will usually reflect the semantics of the connective mentioned in the proof rule.

Consider the case where the last rule applied is the $\neg$-L rule: this is the case
that motivates the more general invariant involving partitions.
We have a partition of the final sequent $\Theta; \Gamma, \phi \proves  \Delta$.
We form a partition of the sequent $\Theta; \Gamma \proves \neg \phi, \Delta$
by placing  $\neg \phi$ on the same side (Left, Right) as $\phi$ was in the original
partition. We then get an interpolant $\theta$ by induction. We just
use $\theta$ for the final interpolant.

We  consider the inductive case for $\wedge$-R.  We have
two top sequents, one for each conjunct. We partition them in the obvious way:
each $\phi_i$ in the top is in the same partition that $\phi_1 \wedge \phi_2$ was
in the bottom. Inductively we take the interpolants
$\theta_1$ and $\theta_2$ for each sequent. We again do a case analysis based on
whether $\phi_1 \wedge \phi_2$ was in $\Delta_L$ or in $\Delta_R$. 

Suppose
$\phi_1 \wedge \phi_2$ was in $\Delta_R$, so $\Delta_R=\phi_1 \wedge \phi_2, \Delta'_R$.
Then we arranged that each $\phi_i$ was
in $\Delta_R$ in the corresponding top sequent. So we know 
that $\Theta_L; \Gamma_L \vdash  \theta_i, \Delta_L$
and $\Theta_R; \Gamma_R, \theta_i \proves \phi_i, \Delta'_R$ for $i=1,2$.
Now can set set the interpolant $\theta$ to be 
 $\theta_1 \wedge \theta_2$.

In the other case, $\phi_1 \wedge \phi_2$ was in $\Delta_L$, say $\Delta_L=\phi_1 \wedge \phi_2, \Delta'_L$.
Then we would arrange each $\phi_i$ to be ``Left'' in the corresponding top sequent, so we know that
$\Theta_L; \Gamma_L \vdash  \theta_i, \phi_i, \Delta'_L$
and $\Theta_R; \Gamma_R, \theta_i \proves \Delta_R$ for $i=1,2$.
We set $\theta=\theta_1 \vee \theta_2$ in this case.

With the $\exists$ rule, a term in the inductively-assumed
   $\theta'$ for the top sequent may become illegal for the
   $\theta$ for the bottom sequent, since it has a free variable that is not common. In this case, the term 
   in $\theta$ is replaced by a quantified variable, where the
   quantifier is existential or universal, depending on the
   partitioning, and bounded according to the requirements
   for deltazero formulas.

\pagebreak
\section{Details for the proof of Theorem~\ref{thm:mainflatcasebounded}, Nested Parameter Collection}

Recall the major result on collecting parameters stated in the body of the paper:

\mainflatcasebounded*

As we mentioned in the body of the paper,
this s a corollary of the following lemma.

\lemnrcparamcoll*

We now give the full details of the proof of the lemma.

We prove this by induction over the size of the proof of 
$\Theta_L, \Theta_R \vdash \Delta_L, \Delta_R,
\exists y \in_p r  ~ \forall z \in c.  ~ (\lambda(z) \equi \rho(z, y))$,
making a case distinction according to which
rule is applied last.
The way $\theta$ will be built will, perhaps unsurprisingly, be very reminiscent
of the way interpolations are normally constructed in standard proof systems
\cite{fittingbook,smullyan}.

For readability, we adopt the
following conventions:
\begin{itemize}
\item We write $\cG$ for the formula
$\exists y \in_p r  ~ \forall z \in c.  ~ (\lambda(z) \equi \rho(z, y))$
    and $\Lambda$ for the expression $\{ z \in c \mid \lambda(z) \}$.
\item For formulas and $\nrc$ expressions obtained by applying the
induction hypothesis, we use the names $\theta\IH$ and $E\IH$
    (or $\theta_1\IH, \theta_2\IH$ and $E_1\IH, E_2\IH$) when the induction hypothesis
is applied several times). In each subcase, our goal will be to build suitable
$\theta$ and $E$.
\item We will color pairs of terms, formulas and sets of  formulas
according to whether they are %
part of either $\colL{\Theta_L; \Delta_L}$
or $\colR{\Theta_R; \Delta_R}$ either at the start of the case analysis or when
we want to apply the induction hypothesis.
In particular, the last sequent of the proof
under consideration will be depicted as
\[ \colL{\Theta_L}, \colR{\Theta_R} \vdash \colL{\Delta_L}, \colR{\Delta_R}, \cG \]
\item Unless it is non-trivial, we leave checking that the free variables in our
proposed definition for $E$ and $\Theta$ are taken among variables of $C$ 
to the reader.
\end{itemize}

With these convention in mind, let us proceed.

\begin{itemize}
  \item If the last rule applied is the $\top$ rule, in both cases we are
    going to take $E \eqdef \emptyset$, but pick $\theta$ to be $\bot$ or $\top$
    according to whether $\top$ occurs in $\colL{\Delta_L}$ or $\colR{\Delta_R}$;
    we leave checking the details to the reader.
  \item If the last rule applied is the $\wedge$ rule, we have two cases according
    to the position of the principal formula $\phi_1 \wedge \phi_2$. In both
    cases, $E$ will be obtained by unioning $\nrc$ expressions obtained from the
    induction hypothesis, and $\theta$ will be either a disjunction or a conjunction.
    \begin{itemize}
      \item If we have $\colL{\Delta_L} = \colL{\phi_1 \wedge \phi_2, \Delta_L'}$,
        so that the proof has shape
        \begin{mathpar}
          \inferrule*{
            \colL{\Theta_L}, \colR{\Theta_R} \vdash \colL{\phi_1, \Delta_L'}, \colR{\Delta_R}, \cG
            \and
            \colL{\Theta_L}, \colR{\Theta_R} \vdash \colL{\phi_2, \Delta_L'}, \colR{\Delta_R}, \cG
          }
          {
            \colL{\Theta_L}, \colR{\Theta_R} \vdash \colL{\phi_1, \Delta_L'}, \colR{\Delta_R}, \cG
          }
        \end{mathpar}
        by the induction hypothesis, we have $\nrc$ expressions
        $E_1\IH$, $E_2\IH$ and formulas $\theta_1\IH, \theta_2\IH$ such that
        \[
          \begin{array}{c !\qquad!{\text{and}}!\qquad c}
            \colL{\Theta_L} \modelssem \colL{\phi_1, \Delta_L'},
            \theta_1\IH \vee \Lambda \in E_1\IH
            &
            \colR{\Theta_R} \modelssem \colR{\Delta_R}, \neg \theta_1\IH
\\
            \colL{\Theta_L} \modelssem \colL{\phi_2, \Delta_L'},
            \theta_2\IH \vee \Lambda \in E_1\IH
            &
            \colR{\Theta_R} \modelssem \colR{\Delta_R}, \neg \theta_2\IH
          \end{array}
        \]
        In that case, we take $E = E_1\IH \cup E_2\IH$ and $\theta \eqdef \theta_1\IH \vee \theta_2\IH$.
        Weakening the properties on the left column, we have
        \[
            \colL{\Theta_L} \modelssem \colL{\phi_i, \Delta_L'},
            \theta \vee \Lambda \in E\]
        for both $i \in \{1,2\}$, so we have
\[
            \colL{\Theta_L} \modelssem \colL{\phi_1 \wedge \phi_2, \Delta_L'},
            \theta \vee \Lambda \in E
\]
as desired. Since $\neg \theta = \neg \theta_1\IH \wedge \neg\theta_2\IH$, we get
        \[
          \colR{\Theta_R} \modelssem \colR{\Delta_R}, \neg \theta
        \]
        by combining both properties from the right column.
    \end{itemize}
  \item Suppose the last rule applied is $\vee$ with principal formula $\phi_1 \vee \phi_2$.
    Depending on whether
    $\colL{\Delta_L} = \colL{\phi_1 \vee \phi_2, \Delta_L'}$
    or
    $\colR{\Delta_R} = \colR{\phi_1 \vee \phi_2, \Delta_R'}$, the proof will
    end with one of the following steps
\[
\dfrac{
\colL{\Theta_L}, \colR{\Theta_R} \vdash
\colL{\phi_1, \phi_2, \Delta_L'},\colR{\Delta_R},\cG}
{\colL{\Theta_L}, \colR{\Theta_R} \vdash
\colL{\phi_1 \vee \phi_2, \Delta_L'},\colR{\Delta_R},\cG}
\qquad \text{or} \qquad
\dfrac{
\colL{\Theta_L}, \colR{\Theta_R} \vdash
\colL{\Delta_L'},\colR{\phi_1,\phi_2,\Delta_R},\cG}
{\colL{\Theta_L}, \colR{\Theta_R} \vdash
\colL{\Delta_L'},\colR{\phi_1 \vee \phi_2, \Delta_R},\cG}
\]
In both cases, we apply the inductive hypothesis according to the obvious
splitting of contexts and sets of formulas, to get a $\nrc$ definition $E\IH$
along with  a formula $\theta\IH$ that satisfy the desired semantic property.
We set $E \eqdef E\IH$ and $\theta \eqdef \theta\IH$.
\item Suppose the last rule applied is $\forall$ with principal formula $\forall x \in b. \phi$. As in the previous
case,  
    depending on  whether
    $\colL{\Delta_L} = \colL{\forall x \in b. \phi, \Delta_L'}$
    or
    $\colR{\Delta_R} = \colR{\forall x \in b. \phi, \Delta_R'}$, the proof will
    end with one of the following steps (assuming $y$ is fresh below)
\[
\dfrac{
\colL{\Theta_L, y \in b}, \colR{\Theta_R} \vdash
\colL{\phi[y/x], \Delta_L'},\colR{\Delta_R},\cG}
{\colL{\Theta_L}, \colR{\Theta_R} \vdash
\colL{\forall x \in b. \phi, \Delta_L'},\colR{\Delta_R},\cG}
\qquad \text{or} \qquad
\dfrac{
\colL{\Theta_L}, \colR{\Theta_R, y \in b} \vdash
\colL{\Delta_L'},\colR{\phi[y/x],\Delta_R},\cG}
{\colL{\Theta_L}, \colR{\Theta_R} \vdash
\colL{\Delta_L'},\colR{\forall x \in b. \phi, \Delta_R},\cG}
\]
In both cases, we again apply the inductive hypothesis according to the obvious
splitting of contexts and sets of formulas to get a $\nrc$ definition $E\IH$
and a formula $\theta\IH$ that satisfy the desired semantic property.
We set $E \eqdef E\IH$ and $\theta \eqdef \theta\IH$.
\item Now we consider  the case where the
last rule applied is $\exists$. Here we have two main subcases,
  according to whether $\cG$ is the main formula or not.
The case where the main formula of $\cG$ was sketched in the body
of the paper, but we repeat it in more detail here.
  \begin{itemize}
    \item If $\cG = \exists y \in_p r  ~ \forall z \in c.  ~ (\lambda(z) \equi \rho(z, y))$
      is the main formula, $\colL{\Delta_L}, \colR{\Delta_R}$ is necessarily existential-leading and
      the proof necessarily has shape 
      \[
        \text{
          \AXC{$\colL{\Theta_L}, \colR{\Theta_R}, \colL{x \in c} \vdash \colL{\Delta_L}, \colR{\Delta_R},
      \colR{\neg \rho(x,w)}, \colL{\lambda(x)},
      \cG$}
\myLeftLabel{$\vee$}
      \UIC{$\colL{\Theta_L}, \colR{\Theta_R}, \colL{x \in c} \vdash \colL{\Delta_L}, \colR{\Delta_R},
\rho(x, w) \imp \lambda(x),
      \cG$}
      \AXC{$\colL{\Theta_L}, \colR{\Theta_R}, \colL{x \in c} \vdash \colL{\Delta_L}, \colR{\Delta_R},
      \colL{\neg \lambda(x)}, \colR{\rho(x, w)},
      \cG$}
\myLeftLabel{$\vee$}
\UIC{$\colL{\Theta_L}, \colR{\Theta_R}, x \in c \vdash \colL{\Delta_L}, \colR{\Delta_R},
\lambda(x) \imp \rho(x, w),
      \cG$}
\myLeftLabel{$\wedge$}
\BIC{$\colL{\Theta_L}, \colR{\Theta_R}, x \in c \vdash \colL{\Delta_L}, \colR{\Delta_R},
\lambda(x) \equi \rho(x, w),
      \cG$}
\myLeftLabel{$\forall$}
\UIC{$\colL{\Theta_L}, \colR{\Theta_R} \vdash \colL{\Delta_L}, \colR{\Delta_R},
\forall z \in c.  ~ (\lambda(z) \equi \rho(z, w)),
      \cG$}
\myLeftLabel{$\exists$}
\UIC{$\colL{\Theta_L}, \colR{\Theta_R} \vdash \colL{\Delta_L}, \colR{\Delta_R}, \cG$}
\DisplayProof
        }\]
where $x$ is a fresh variable
So in particular, we have two strict subproofs with respective conclusions
      \[
        \colL{\Theta_L, x \in c}, \colR{\Theta_R, x \in c} \vdash \colL{\lambda(x), \Delta_L}, \colR{\neg \rho(x,w), \Delta_R},\cG
    \qquad \text{and} \qquad
      \colL{\Theta_L, x \in c}, \colR{\Theta_R, x \in c} \vdash \colL{\neg\lambda(x),\Delta_L}, \colR{\rho(x,w), \Delta_R},
      \cG
\]
Applying the inductive hypothesis, we obtain $\nrc$ expressions $E_1\IH$, $E_2\IH$
and formulas $\theta_1\IH, \theta_2\IH$ which contain free variables in $C \cup \{x\}$ such that all of the following hold
\begin{align}
&
  \colL{\Theta_L, x \in c} \modelssem \phantom{\neg}\colL{\lambda(x), \Delta_L}, \theta_1\IH \vee \Lambda  \in E_1\IH
\label{align:1} \\
  \text{and} \qquad &
  \colL{\Theta_L, x \in c} \modelssem \colL{\neg \lambda(x), \Delta_L}, \theta_2\IH \vee \Lambda  \in E_2\IH
  \label{align:2} \\
  \text{and} \qquad &
  \phantom{x \in c,}\colR{\Theta_R} \modelssem \colR{\neg\rho(x,w),\Delta_R}, \neg \theta_1\IH
  \label{align:3} \\
  \text{and} \qquad &
  \phantom{x \in c,} \colR{\Theta_R} \modelssem \phantom{\neg}\colR{\rho(x,w), \Delta_R}, \neg \theta_2\IH
  \label{align:4}
\end{align}
With this in hand, we set
\[ \theta \eqdef \exists x \in c. ~ \theta_1\IH \wedge \theta_2\IH \qquad \text{and} \qquad
         E \eqdef \left\{\left\{ x \in c \mid \theta_2\IH\right\}\right\} ~\cup ~ \bigcup\left\{ E_1\IH \cup E_2\IH \mid x \in c \right\}
      \]
Note in particular that the free variables of $E$ and $\theta$ are contained in $C$, since we bind $x$.
The bindings of $x$ have radically different meaning across the two main components 
$E_1 \eqdef \left\{\left\{ x \in c \mid \theta_2\IH\right\}\right\}$ and $ E_2 \eqdef \left\{ E_1\IH \cup E_2\IH \mid x \in c \right\}$
of $E = E_1 \cup E_2$. $E_1$ consists of a single definition corresponding to the restriction of $c$ to $\theta_2\IH$, and there $x$
plays the role of an element being defined. On the other hand, $E_2$ corresponds to the joining of all the definitions obtained inductively,
which may contain an $x \in c$ as a parameter. So  we have two families of potential definitions for $\Lambda$ indexed by $x \in c$ that we join together.
Now let us show that we have the desired semantic properties. First we need to show that
$E$ contains a definition for $\Lambda$ under the right hypotheses, i.e.,
\begin{align}
\colL{\Theta_L} \modelssem \colL{\Delta_L}, \exists x \in c. \; \theta_1\IH \wedge \theta_2\IH, \Lambda \in \left( 
\left\{\left\{ x \in c \mid \theta_2\IH\right\}\right\} ~\cup ~ \bigcup\left\{ E_1\IH \cup E_2\IH\mid x \in c \right\}\right)
\label{align:5}
\end{align}
which can be rephrased as
\begin{align*}
\colL{\Theta_L} \modelssem \colL{\Delta_L}, \exists x \in c. \; \theta_1\IH \wedge \theta_2\IH,
\Lambda = \left\{ x \in c \mid \theta_2\IH\right\}, \exists x \in c. \; \Lambda \in E_1\IH \cup E_2\IH
\end{align*}
Now concentrate on the statement
$\Lambda = \left\{x \in c \mid \theta_2\IH\right\}$.
It would follow from 
the two inclusions
$\Lambda \subseteq \left\{x \in c \mid \theta_2\IH\right\}$
and
$\left\{x \in c \mid \theta_2\IH\right\} \subseteq \Lambda$, so, recalling that
$\Lambda = \{ x \in c \mid \lambda(x)\}$, the overall conclusion would follow from
having
\begin{align*}
&
\colL{\Theta_L}, x \in c \modelssem \colL{\Delta_L}, \exists x \in c. \; \theta_1\IH \wedge \theta_2\IH,
  \lambda(x) \imp \theta_2\IH, \exists x \in c. \; \Lambda \in E_1\IH \cup E_2\IH
\\
  \text{and}\qquad
  &
\colL{\Theta_L}, x \in c \modelssem \colL{\Delta_L}, \exists x \in c. \; \theta_1\IH \wedge \theta_2\IH,
  \theta_2\IH \imp \lambda(x), \exists x \in c. \; \Lambda \in E_1\IH \cup E_2\IH
\end{align*}
Those in turn follow from the following two statements
\begin{align*}
&
\colL{\Theta_L}, x \in c \modelssem \colL{\Delta_L}, \theta_1\IH \wedge \theta_2\IH,
  \neg \lambda(x), \theta_2\IH, \Lambda \in E_1\IH \cup E_2\IH
\\
  \text{and}\qquad
  &
\colL{\Theta_L}, x \in c \modelssem \colL{\Delta_L}, \theta_1\IH \wedge \theta_2\IH,
  \neg\theta_2\IH, \lambda(x), \Lambda \in E_1\IH \cup E_2\IH
\end{align*}
which are straightforward consequences of~\ref{align:2} and~\ref{align:1} respectively.
This concludes the proof of~\ref{align:5}.

Now we only need to prove a final property, which is
\begin{align*}
\colR{\Theta_R} \modelssem \colR{\Delta_R}, \forall x \in c. \; \neg \theta_1\IH \vee \neg \theta_2\IH 
\end{align*}
which is equivalent to the validity of
\begin{align*}
\colR{\Theta_R}, x \in c \modelssem \colR{\Delta_R}, \neg \theta_1\IH \vee \neg \theta_2\IH
\end{align*}
which can be obtained by combining~\ref{align:3} and~\ref{align:4} with excluded
middle for $\rho(x,w)$.

\item If $\cG$ is not the main formula, then we have two subcases corresponding to
  whether the main formula under consideration occurs in $\colL{\Delta_L}$ or $\colR{\Delta_R}$.
  In both cases, the restriction on the $\exists$-rule imposed by the focused proof system
  is not particularly relevant and the block quantification hinders readability.
  Let us treat the equivalent of the case of a more general rule with a single quantifier in
  an auxiliary lemma.
  \begin{lemma}
    \label{lem:mainflatcaseboundedexleft}
    Suppose that we have a formula $\theta\IH$ and a $\nrc$ expression $E\IH$ such that
    \[ \colL{\Theta_L} \modelssem \colL{\phi[w/x], \exists x \in r. \; \phi, \Delta_L}, \theta\IH \vee \Lambda \in E\IH \qquad \text{and}
       \qquad \colR{\Theta_R} \modelssem \colR{\Delta_R}, \neg \theta\IH
    \]
    with $\FV(E\IH, \theta\IH) \subseteq \FV(\Theta_L,\Delta_L,\phi[w/x]) \cap \FV(\Theta_R, \Delta_R)$.
    and that we additionally that $w \in t$ is part of $\colL{\Theta_L}, \colR{\Theta_R}$.
    Then we can define $\theta$ and $E$ with
    $\FV(E, \theta) \subseteq \FV(\Theta_L,\Delta_L,\phi) \cap \FV(\Theta_R, \Delta_R)$ and
    \[ \colL{\Theta_L} \modelssem \exists x \in t. \; \phi, \Delta_L, \theta \vee \Lambda \in E \qquad \text{and}
       \qquad \colR{\Theta_R} \modelssem \colR{\Delta_R}, \neg \theta
    \]
  \end{lemma}
  \begin{proof}
    There are two subcases to consider:
    \begin{itemize}
      \item If $w \in t$ is part of $\colL{\Theta_L}$, we can conclude immediately by setting $\theta \eqdef \theta\IH$ and $E \eqdef E\IH$.
      \item Otherwise $w \in t$ is part of $\colR{\Theta_R}$, and it might be the case that $w$ is a free
        variable which is not part of $\colL{\Theta_L; \Delta_L, \exists x \in t. \phi}$, so $\theta\IH$ and $E\IH$ may feature $w$ as a free variable.
        In that case, we know that the free variables of $t$ are common and we can set $\theta \eqdef \forall x \in t. \; \theta[x/w]$ and $E \eqdef \bigcup
        \left\{ E \mid w \in t \right\}$.
    \end{itemize}
  \end{proof}
  A dual lemma where the existential formula is located on the right side of the partition can be proven in a completely analogous way. 
We can use those lemmas to prove the result by induction on the size of the quantifier block for the focused $\exists$ rule.
\end{itemize}
\item The case of the $=$ rule can be handled exactly as the $\top$ rule.
\item For the $\neq$ rule, we distinguish several subcases:
  \begin{itemize}
    \item If we have $\colL{\Delta_L} = \colL{y \neq_\ur z, \alpha[y/x], \Delta_L'}$
      or $\colR{\Delta_R} = \colR{y \neq_\ur z, \alpha[y/x], \Delta_R'}$, so that the
      last step has either shape
      \[
        \dfrac{
          \colL{\Theta_L}, \colR{\Theta_R} \vdash \colL{y \neq_\ur z, \alpha[y/x], \alpha[z/x], \Delta_L'
        }, \colR{\Delta_R'}, \cG
      }{
          \colL{\Theta_L}, \colR{\Theta_R} \vdash \colL{y \neq_\ur z, \alpha[y/x], \Delta_L'
        }, \colR{\Delta_R'}, \cG
      }
      \qquad \text{or} \qquad
        \dfrac{
          \colL{\Theta_L}, \colR{\Theta_R} \vdash \colL{\Delta_L'}, \colR{y \neq_\ur z, \alpha[y/x], \alpha[z/x], \Delta_R'
        }, \cG
      }{
        \colL{\Theta_L}, \colR{\Theta_R} \vdash \colL{\Delta_L'}, \colR{y \neq_\ur z, \alpha[y/x], \Delta_R'
        }, \cG
      }
    \]
    we can apply the induction hypothesis to obtain some $\theta\IH$ and $E\IH$ such that setting $\theta \eqdef \theta\IH$ and $E \eqdef E\IH$ solves this
      subcase; we leave checking the additional properties to the reader.
      \item Otherwise, if we have
        $\colL{\Delta_L} = \colL{y \neq_\ur z, \Delta_L'}$,
        $\colR{\Delta_R} = \colR{\alpha[y/x], \Delta'_R}$ and a last step of
        shape
      \[
        \dfrac{
          \colL{\Theta_L}, \colR{\Theta_R} \vdash \colL{y \neq_\ur z, \Delta_L'
        }, \colR{\alpha[y/x], \alpha[z/x],\Delta_R'}, \cG
      }{
          \colL{\Theta_L}, \colR{\Theta_R} \vdash \colL{y \neq_\ur z, \Delta_L'
        }, \colR{\alpha[y/x], \Delta_R'}, \cG
      }
    \]
      In that case, the inductive hypothesis gives $\theta\IH$ and $E\IH$ with
      free variables in $C \cup \{z\}$ such that
      \[
        \colL{\Theta_L} \modelssem \colL{y \neq_\ur z, \Delta_L'}, \theta\IH, \Lambda \in E\IH
      \qquad \text{and} \qquad \colR{\Theta_R} \modelssem \colR{\alpha[y/x], \alpha[z/x], \Delta_R'}, \neg \theta\IH \]
    We then have two subcases according to whether $z \in C$ or not
    \begin{itemize}
        \item If $z \in C$, we can take $\theta \eqdef \theta\IH \wedge y =_\ur z$ and $E \eqdef E\IH$. Their free variables are in $C$
          and we only need to check
      \[
        \colL{\Theta_L} \modelssem \colL{y \neq_\ur z, \Delta_L'}, \theta\IH \wedge y =_\ur z, \Lambda \in E\IH
      \qquad \text{and} \qquad \colR{\Theta_R} \modelssem \colR{\alpha[y/x], \Delta_R'}, \neg \theta\IH, y \neq_\ur z \]
which follow easily from the induction hypothesis.
        \item Otherwise, we take $\theta \eqdef \theta\IH[y/z]$ and $E \eqdef E\IH[y/z]$.
          In that case, note that we have $\alpha[z/x][y/z] = \alpha[y/x]$ (which would not be necessarily the case if $z$ belonged to $C$).
          This allows to conclude that we have
      \[
        \colL{\Theta_L} \modelssem \colL{y \neq_\ur z, \Delta_L'}, \theta\IH[y/z], \Lambda \in E\IH[y/z]
        \qquad \text{and} \qquad \colR{\Theta_R} \modelssem \colR{\alpha[y/x], \Delta_R'}, \neg \theta\IH[y/z] \]
          directly from the induction hypothesis.
      \end{itemize}
      \item Otherwise, if we have
        $\colL{\Delta_L} = \colL{\alpha[y/x], \Delta_L'}$,
        $\colR{\Delta_R} = \colR{y \neq_\ur z, \Delta'_R}$ and a last step of
        shape
      \[
        \dfrac{
          \colL{\Theta_L}, \colR{\Theta_R} \vdash
          \colL{\alpha[y/x], \alpha[z/x],\Delta_L'},
          \colR{y \neq_\ur z, \Delta_R'},
\cG
      }{
  \colL{\Theta_L}, \colR{\Theta_R} \vdash
          \colL{\alpha[y/x],\Delta_L'},
          \colR{y \neq_\ur z, \Delta_R'},
\cG
        }
    \]
      In that case, the inductive hypothesis gives $\theta\IH$ and $E\IH$ with
      free variables in $C \cup \{z\}$ such that
      \[
        \colL{\Theta_L} \modelssem \colL{\alpha[y/x],\alpha[z/x], \Delta_L'}, \theta\IH, \Lambda \in E\IH
      \qquad \text{and} \qquad \colR{\Theta_R} \modelssem \colR{y \neq_\ur z, \Delta_R'}, \neg \theta\IH \]
    We then have two subcases according to whether $z \in C$ or not
    \begin{itemize}
        \item If $z \in C$, we can take $\theta \eqdef \theta\IH \vee y \neq_\ur z$ and $E \eqdef E\IH$. Their free variables are in $C$
          and we only need to check
      \[
        \colL{\Theta_L} \modelssem \colL{\alpha[y/x], \Delta_L'}, \theta\IH \vee y \neq_\ur z, \Lambda \in E\IH
      \qquad \text{and} \qquad \colR{\Theta_R} \modelssem \colR{y \neq_\ur z, \Delta_R'}, \neg \theta\IH \wedge y =_\ur z \]
which follow easily from the induction hypothesis.
        \item Otherwise, we take $\theta \eqdef \theta\IH[y/z]$ and $E \eqdef E\IH[y/z]$.
          In that case, note that we have $\alpha[z/x][y/z] = \alpha[y/x]$ (which would not be necessarily the case if $z$ belonged to $C$).
          This allows to conclude that we have
      \[
        \colL{\Theta_L} \modelssem \colL{\alpha[y/x], \Delta_L'}, \theta\IH[y/z], \Lambda \in E\IH[y/z]
        \qquad \text{and} \qquad \colR{\Theta_R} \modelssem \colR{y \neq_\ur z, \Delta_R'}, \neg \theta\IH[y/z] \]
          directly from the induction hypothesis.
      \end{itemize}
  \end{itemize}
\item If the last rule applied is $\times_\eta$, the proofs has shape
\[
  \dfrac{\colL{\Theta_L}[\tuple{x_1,x_2}/x],
         \colR{\Theta_R}[\tuple{x_1,x_2}/x] \vdash 
         \colL{\Delta_L}[\tuple{x_1,x_2}/x],
         \colR{\Delta_R}[\tuple{x_1,x_2}/x], \cG[\tuple{x_1,x_2}/x]}{\colL{\Theta_L}, \colR{\Theta_R} \vdash \colL{\Delta_L}, \colR{\Delta_R}, \cG}
\]
and one applies the inductive hypothesis as expected to get $\theta\IH$ and
$E\IH$ such that
\[\colL{\Theta_L}[\tuple{x_1,x_2}/x] \modelssem
         \colL{\Delta_L}[\tuple{x_1,x_2}/x], \theta\IH, \Lambda[\tuple{x_1,x_2}/x] \in E
\qquad \text{and} \qquad
         \colR{\Theta_R}[\tuple{x_1,x_2}/x] \modelssem 
         \colR{\Delta_R}[\tuple{x_1,x_2}/x], \neg\theta\IH
\]
and with free variables included in $C$ if $x \notin C$ or
$\left(C \cup \{x_1, x_2\}\right) \setminus \{x\}$ otherwise.
In both cases, it is straightforward to check that taking
$\theta \eqdef \theta\IH[\pi_1(x)/x_1,\pi_2(x)/x_2]$ and
$E \eqdef E\IH[\pi_1(x)/x_1,\pi_2(x)/x_2]$ will yield the desired result.
\item Finally, if the last rule applied is the $\times_\beta$ rule, it has shape
\[
  \dfrac{
         \colL{\Theta_L}[x_i/x],
         \colR{\Theta_R}[x_i/x] \vdash 
         \colL{\Delta_L}[x_i/x],
         \colR{\Delta_R}[x_i/x],\cG'[x_i/x]}{
         \colL{\Theta_L}[\pi_i(\tuple{x_1,x_2})/x],
         \colR{\Theta_R}[\pi_i(\tuple{x_1,x_2})/x] \vdash 
         \colL{\Delta_L}[\pi_i(\tuple{x_1,x_2})/x],
         \colR{\Delta_R}[\pi_i(\tuple{x_1,x_2})/x],\cG'[\pi_i(\tuple{x_1,x_2})/x]}
\]
and we can apply the induction hypothesis to get satisfactory
$\theta\IH$ and $E\IH$ (moving from $\cG$ to $\cG'[x_i/x]$ is unproblematic, as we can assume the
lemma works for $\cG$ with arbitrary subformulas $\lambda$ and $\rho$); it is easy to see that we can set $\theta \eqdef \theta\IH$
and $E \eqdef E\IH$ and conclude.
\end{itemize}

This completes the proof of the lemma.

\section{Proofs of polytime admissibility}

The goal of this section is to prove most claims of polytime
admissibility made in the body of the paper, crucially those of
Section~\ref{sec:proof:ain}.
Recall that a rule
\[\dfrac{\Theta \vdash \Delta}{\Theta' \vdash \Delta'}\]
is \emph{polytime admissible} if we can compute in polynomial time 
a proof of the conclusion $\Theta' \vdash \Delta'$ from a proof of
the antecedent $\Theta \vdash \Delta$, and \emph{polytime derivable}
if there is a polynomial-size poof tree with dangling leaves labelled by the
antecedent.

Throughout this section we deal with the focused proof system of
Figure \ref{fig:dzcalc-1sided}.

\subsection{Standard rules}

Here we collect some useful standard sequent calculi rules, which are all polytime
admissible in our system. 
The arguments for these rules are straightforward.

\begin{lemma}
  \label{lem:wkadm}
  The following weakening rule is polytime admissible:
\[
  \dfrac{\Theta \vdash \Delta}{\Theta' \vdash \Delta, \Delta'}
\]
\end{lemma}

\begin{lemma}
  \label{lem:wedgeprojadm}
  The following inference, witnessing the invertibility of the $\wedge$ rule, is polytime admissible for both $i \in \{1,2\}$:
\[
  \dfrac{\Theta \vdash \phi_1 \wedge \phi_2, \Delta}{\Theta \vdash \phi_i, \Delta}\]
\end{lemma}

\begin{lemma}
  \label{lem:allinvadm}
  The following, witnessing the invertibility of the $\forall$ rule, is polytime admissible:
  \[ \dfrac{\Theta \vdash \forall x \in t. \phi, \Delta}{\Theta, x \in t \vdash \phi, \Delta}\]
\end{lemma}

\begin{lemma}
  \label{lem:existsblockadm}
  The following generalization of the $\exists$ rule is polytime admissible:
  \[\dfrac{\Theta, \Theta_2 \vdash \phi', \phi, \Delta^\pospolaritysuper \quad \text{$\phi'$ is a spec~of $\phi$ wrt an ordering of $\Theta_2$}}{\Theta, \Theta_2 \vdash \phi, \Delta^\pospolaritysuper}
  \]
\end{lemma}

\begin{lemma}
  \label{lem:substitutivity}
  The following substitution rule is polytime admissible:
  \[\dfrac{\Theta \vdash \Delta}{\Theta[t/x] \vdash \Delta[t/x]}
  \]
\end{lemma}

\subsection{Admissibility of generalized congruence}

Recall the admissibility claim concerning the rule related to congruence:

\begin{restatable}{lemma}{gencongruence}
  \label{lem:gencongruence}
The following generalized congruence rule is polytime admissible:
\[ \dfrac{\Theta[t/x,t/y]\seq \Delta[t/x,t/y]}{\Theta[t/x,u/y] \seq \neg (t \equiv u), \Delta[t/x,u/y]}
  \]
\end{restatable}

Recall that in a two-side reading  of this, the hypothesis
is 
$t \equiv u;\Theta[t/x,u/y] \seq \Delta[t/x,u/y]$.
So the rule says that if we $\Theta$ entails $\Delta$ where both
contain $t$, then if we assume $t \equiv u$ and substitute some occurrences of $t$ with
$u$ in $\Theta$, we can conclude the corresponding substitution of $\Delta$.

To prove Lemma \ref{lem:gencongruence} in the case where the terms $t$ and $u$ are of type
$\sett(T)$, we will need
a more general statement.
We are going to generalize the statement to treat tuples of terms and
use $\pospolarity$ formulas instead of $\neg (t \equiv u)$ to simplify the inductive
invariant.

Given two terms $t$ and $u$ of type $t$, define by induction the set of formulas
$\cE_{t,u}$:
\begin{itemize}
    \item If $T = \ur$, then $\cE_{t,u}$ is $t \neq_\ur u$ 
    \item If $T = T_1 \times T_2$, then
      $\cE_{t,u}$ is $\cE_{\pi_1(t),\pi_1(u)}, \cE_{\pi_2(t),\pi_2(u)}$
    \item If $T = \sett(T')$, $\cE_{t,u}$ is $\neg (t \subseteq_T u), \neg(u \subseteq_T t)$
\end{itemize}

The reader can check that $\cE_{t,u}$ is essentially $\neg (t \equiv u)$.

\begin{lemma}
\label{lem:cEder}
The following rule is polytime admissible:
  \[
    \dfrac{\Theta \vdash \Delta, \cE_{t,u}}{\Theta \vdash \Delta, \neg (t \equiv u)}
  \]
\end{lemma}
\begin{proof}
  Straightforward induction over $T$.
\end{proof}

Since we will deal with multiple equivalences, we will adopt vector notation
$\vv t = t_1, \ldots, t_n$ and $\vv x = x_1, \ldots, x_n$ for lists of terms and
variables. Call $\cE_{\vv t, \vv u}$ the union of the $\cE_{t_i,u_i}$.
We can now state our more general lemma:

\begin{lemma}
  \label{lem:gencongruencenary}
  The following generalized $n$-ary congruence rule for set variables is polytime admissible:
  \[
    \dfrac{\Theta[\vv t/\vv x, \vv t /\vv y]\seq \Delta[\vv t/\vv x,\vv t/\vv y]}{
   \Theta[\vv t/\vv x, \vv u/\vv y]\seq
  \Delta[\vv t/\vv x, \vv u/\vv y], \cE_{\vv t, \vv u}
}\]
\end{lemma}

\begin{proof}[Proof of Lemma~\ref{lem:gencongruencenary}]
We proceed by induction over the proof of $\Theta[\vv t/\vv x, \vv t/\vv y] \seq \Delta[\vv t/\vv x,\vv u/\vv y]$,
\begin{itemize}
\item If the last rule applied is the $=$ rule, i.e. we have
  \[ \dfrac{}{\Theta[\vv t/\vv x, \vv t/\vv y] \vdash
    a[\vv t/\vv x, \vv t/\vv y] =_\ur b[\vv t/\vv x, \vv t/\vv y],
    \Delta[\vv t/\vv x, \vv t/ \vv y]}\]
    and $a[\vv t/\vv x, \vv t/\vv y] = b[\vv t/\vv x, \vv t/\vv y]$, with $a,b$ variables.
    Now if $a$ and $b$ are equal, or if they belong both to either $\vv x$ or $\vv y$,
    it is easy to derive the desired conclusion with a single application of the $=$ rule.
    Otherwise, assume $a = x_i$ and $b = y_j$ (the symmetric case is handled similarly).
    In such a case, we have that $\cE_{\vv t, \vv u}$ contains $t_i \neq_\ur t_j$. So the desired
    proof
  \[ \dfrac{}{\Theta[\vv t/\vv x, \vv t/\vv y] \vdash
    t_i =_\ur u_j,
    \Delta[\vv t/\vv x, \vv t/ \vv y], \cE_{\vv t, \vv u}, t_i \neq_\ur t_j}\]
    follows from the polytime admissibility of the axiom rule.
\item Suppose the last rule applied is the $\top$ rule:
  \[ \dfrac{}{\Theta[\vv t/\vv x, \vv t/\vv y] \vdash \top, \Delta[\vv t/\vv x, \vv t/ \vv y]}\]
    Then we do not need to  apply the induction hypothesis. Instead we can immediately
apply  the $\top$ rule to obtain
    \[ \dfrac{}{\Theta[\vv t/\vv x, \vv u/\vv y] \vdash \top, \Delta[\vv t/\vv x, \vv u/ \vv y], \cE_{\vv t, \vv u}}\]

\item If the last rule applied is the $\wedge$ rule
  \[ \dfrac{
    \Theta[\vv t/\vv x, \vv t/\vv y] \vdash \phi_1[\vv t/\vv x, \vv t/\vv y], \Delta[\vv t/\vv x, \vv t/ \vv y]
\qquad
    \Theta[\vv t/\vv x, \vv t/\vv y] \vdash \phi_2[\vv t/\vv x, \vv t/\vv y], \Delta[\vv t/\vv x, \vv t/ \vv y]
  }{
    \Theta[\vv t/\vv x, \vv t/\vv y] \vdash (\phi_1 \wedge \phi_2)[\vv t/\vv x, \vv t/\vv y], \Delta[\vv t/\vv x, \vv t/ \vv y]
}
    \]
    then the induction hypothesis gives us proofs of
    \[    \Theta[\vv t/\vv x, \vv u/\vv y] \vdash \phi_i[\vv t/\vv x, \vv u/\vv y], \Delta[\vv t/\vv x, \vv u/ \vv y], \cE_{\vv t, \vv u} \]
    for both $i \in \{1,2\}$. So we can apply the $\wedge$ rule to conclude that we have
    \[    \Theta[\vv t/\vv x, \vv u/\vv y] \vdash (\phi_1 \wedge \phi_2)[\vv t/\vv x, \vv u/\vv y], \Delta[\vv t/\vv x, \vv u/ \vv y], \cE_{\vv t, \vv u} \]
    as desired.
  \item The cases of the rules $\vee$,$\forall$ and $\times_\eta$ are equally straightforward and left to the reader.
  \item Now, let us handle the case of the $\exists$ rule. To simplify notation, 
we just treat the case where there is only one leading existential in the formula.
    \[ \dfrac{\Theta, t \in u \vdash \phi[t/x], \Delta^\pospolaritysuper}{\Theta, t \in u \vdash \exists x \in u. \; \phi, \Delta^\pospolaritysuper }\]
    The generalization to multiple existentials can be obtained
    by translating the general $\exists$ rule into a sequence of applications of
    the rule above and applying the same argument to each instance (note that this
    is possible because we do not require $\phi$ to be $\negpolarity$ in the rule
    for a single existential).

    So assume that $z$ is fresh wrt $\vv x, \vv y, \vv t, \vv u, a, b$ and
    that the last step of the proof is
  \[ \dfrac{
    \Theta[\vv t/\vv x, \vv t/\vv y], a[\vv t/\vv x, \vv t/\vv y] \in b[\vv t/\vv x, \vv t/\vv y]
    \vdash \phi[a/z][\vv t/\vv x, \vv t/\vv y],
    \exists z \in c[\vv t/\vv x, \vv t/\vv y].\; \phi[\vv t/\vv x, \vv t/\vv y],
    \Delta[\vv t/\vv x, \vv t/ \vv y]
  }{
    \Theta[\vv t/\vv x, \vv t/\vv y], a[\vv t/\vv x, \vv t/\vv y] \in b[\vv t/\vv x, \vv t/\vv y]
    \vdash
    \exists z \in c[\vv t/\vv x, \vv t/\vv y].\; \phi[\vv t/\vv x, \vv t/\vv y],
    \Delta[\vv t/\vv x, \vv t/ \vv y]
    }
    \]
    with $b[\vv t/\vv x, \vv t/\vv y] = c[\vv t/\vv x, \vv t/\vv y]$.
    Set $a' = a[\vv t/\vv x, \vv u/\vv y]$, $b' = b[\vv t/\vv x, \vv u/\vv y]$, $c'=c[\vv t/\vv x, \vv u/\vv y]$.
    We have three subcases:
    \begin{itemize}
    \item If we have that $b' = c'$,
    using the induction hypothesis, we have a proof of
    \[    \Theta[\vv t/\vv x, \vv u/\vv y], a' \in b'
    \vdash \phi[a/z][\vv t/\vv x, \vv u/\vv y],
    \exists z \in c'.\; \phi[\vv u/\vv x, \vv u/\vv y],
    \Delta[\vv t/\vv x, \vv u/ \vv y], \cE_{\vv t, \vv u}\]
    we can simply apply an $\exists$ rule to that proof and we are done.
  \item Otherwise, if we have that $b' = t_i$ and $c' = u_j$ for some $i,j \le n$.
    In that case,
    extending the tuples $\vv t$ and $\vv u$ with $a'$ and a fresh variable
        $z'$ (the substitutions under consideration would be,
        we can apply the induction hypothesis to obtain a proof of
\[    \Theta[\vv t/\vv x, \vv u/\vv y, a'/z], a' \in t_i, z' \in u_j
    \vdash \phi[\vv t/\vv x, \vv u/\vv y,z'/z],
    \exists z \in u_j.\; \phi[\vv u/\vv x, \vv u/\vv y],
        \Delta[\vv t/\vv x, \vv u/ \vv y], \cE_{\vv t, \vv u}, \cE_{a',z'}\]
      Note that $\cE_{\vv t, \vv u}$
      contains an occurence of $\neg (t_i \subseteq u_j)$, which expands to
        $\exists z \in t_i. \forall z' \in u_j. \neg (z \equiv z')$, so we can construct the
        partial derivation
    \[
      \text{
        \AXC{$
    \Theta[\vv t/\vv x, \vv u/\vv y, a'/z], a' \in t_i, z' \in u_j
    \vdash \phi[\vv t/\vv x, \vv u/\vv y,z'/z'],
    \exists z \in u_j.\; \phi[\vv u/\vv x, \vv u/\vv y],
        \Delta[\vv t/\vv x, \vv u/ \vv y], \cE_{\vv t, \vv u}, \cE_{a',z'}$
        }
        \myLeftLabel{$\exists$}
        \UIC{$
    \Theta[\vv t/\vv x, \vv u/\vv y], a' \in t_i, z' \in u_j
    \vdash \exists z \in u_j. \phi[\vv t/\vv x, \vv u/\vv y],
        \Delta[\vv t/\vv x, \vv u/ \vv y], \cE_{\vv t, \vv u}, \cE_{a',z'}$
        }
        \myLeftLabel{Lemma~\ref{lem:cEder}}
        \UIC{$
    \Theta[\vv t/\vv x, \vv u/\vv y, a'/z], a' \in t_i, z' \in u_j
    \vdash 
    \exists z \in u_j.\; \phi[\vv u/\vv x, \vv u/\vv y],
        \Delta[\vv t/\vv x, \vv u/ \vv y], \cE_{\vv t, \vv u}, a' \equiv z'$}
        \myLeftLabel{$\forall$}
        \UIC{$
    \Theta[\vv t/\vv x, \vv u/\vv y, a'/z], a' \in t_i
    \vdash
    \exists z \in c'.\; \phi[\vv u/\vv x, \vv u/\vv y],
        \Delta[\vv t/\vv x, \vv u/ \vv y], \cE_{\vv t, \vv u}, \forall z' \in u_j. a' \equiv z'$}
        \myLeftLabel{$\exists$}
        \UIC{$
    \Theta[\vv t/\vv x, \vv u/\vv y, a'/z], a' \in t_i
    \vdash
    \exists z \in c'.\; \phi[\vv t/\vv x, \vv u/\vv y],
        \Delta[\vv t/\vv x, \vv u/ \vv y], \cE_{\vv t, \vv u}$}
      \DisplayProof
      }
    \]
    whose conclusion matches what we want.
    \item Otherwise, we are in a similar case where $b' = u_j$ and $c' = t_i$.
      We proceed similarly, except that we use the formula $\neg (u_j \subseteq t_i)$
        of $\cE_{t_i,u_j}$ instead of $\neg(t_i \subseteq u_j)$.
    \item For the rule $\times_\beta$, which has general shape
      \[\dfrac{\Theta[z_i/z] \vdash \Delta[z_i/z]}{\Theta[\pi_i(\tuple{z_1,z_2})/z] \vdash \Delta[\pi_i(\tuple{z_1,z_2})/z]}\]
      we can assume, without loss of generality, that $z$ occurs only once in $\Theta, \Delta$.
      Let us only sketch the case where $z$ occurs in a formula $\phi$ and the rule has shape
      \[\dfrac{\Theta \vdash \phi[z_i/z], \Delta
      }{\Theta \vdash \phi[\pi_i(\tuple{z_1,z_2})/z], \Delta}\]
      We have that $\phi[\pi_i(\tuple{z_1,z_2})/z)]$ is also of the shape $\psi[\vv t/\vv x, \vv u/\vv y]$ in our situation.
      We can also assume without loss of generality that each variable in $\vv x$ and $\vv y$ occur each a single time in $\Theta, \Delta$.
      Now if we have a couple of situations:
      \begin{itemize}
        \item If the occurence of $z$ do not interfere with the substitution $[\vv t/\vv x, \vv u/\vv y]$, i.e., there is a formula $\theta$ such that
          \[\phi[\pi_i(\tuple{z_1,z_2})/z] = \psi[\vv t/\vv x, \vv u] = \theta[\vv t/\vv x, \vv u/\vv y, \tuple{z_1,z_2}/z]\]
            we can simply apply the induction hypothesis on the subproof and conclude with one application of $\times_\eta$.
          \item If we have that $z$ clashes with a variable of $\vv x, \vv y$, say $x_j$, but that $t_j = v[\pi_i(\tuple{z_1,z_2}/x_j)]$ for some term $v$.
            Then we can apply the induction hypothesis with the matching tuples of terms $\vv t, x_i$ and $\vv u, t_i$ and conclude by applying the $\beta$ rule.
          \item Otherwise, the occurence of $z$ does interfere with the substitution in such a way that we have,
            say $t_j = \tuple{z_1, z_2}$. Then we can apply the induction hypothesis on the subproof with the matching tuples of terms $\vv x, z_i$ and $\vv y, \pi_i(u_j)$
            and conclude by applying the $\beta$ rule.
      \end{itemize}
    \end{itemize}
\end{itemize}
\end{proof}

One easy consequence of the above is Lemma \ref{lem:gencongruence}:

\begin{proof}[Proof of Lemma~\ref{lem:gencongruence}]
Combine Lemma~\ref{lem:gencongruencenary} and Lemma~\ref{lem:cEder}.
\end{proof}

Another consequence is the following corollary, which will be used later in this section:

\begin{restatable}{corollary}{memctxadm}
  \label{cor:memctxadm}
The following rule is polytime admissible:
\[
  \dfrac{\Theta, t \in u \seq \Delta}{\Theta \seq \neg t\memmac u, \Delta}
\]
\end{restatable}

\begin{proof}
Recall that $t \memmac u$ expands to $\exists x \in u. \; x \equiv t$,
so that $\neg t\memmac u$ is $\forall x \in u. \; \neg (x \equiv t)$.
So we have
\[
    \text{
      \AXC{$\Theta, t \in u \vdash \Delta$}
      \myLeftLabel{Lemma~\ref{lem:gencongruence}}
      \UIC{$\Theta, x \in u \vdash \neg (x \equiv t), \Delta$}
      \myLeftLabel{$\forall$}
      \UIC{$\Theta \vdash \neg t \memmac u, \Delta$}
\DisplayProof
    }\]
\end{proof}

\subsection{Proof of  Lemma \ref{lem:effnestedproof-down}}

We now recall the claim of admissibility concerning rules
for ``moving down in an equivalence''. Recall that these use
the notation for quantifying on a path below an object, defined
in the body.

\effnestedproofdown*

\begin{proof}
  We proceed by induction over the input proof of
  $\Theta \vdash \Delta, \exists r' \inq_p o'. \; r \equiv_{\sett(T')} r'$
and make a case distinction
  according to which rule was applied last. All cases are straightforward, save for one: when a $\exists$ rule is applied on the formula  $\exists r'
  \inq_p o'. \; r \equiv_{\sett(T')} r'$.
Let us only detail that one.

In that case, the last step has shape
  \[
    \dfrac{
      \Theta \vdash \Delta^\pospolaritysuper, r \equiv_{\sett(T')} w, \exists r' \inq_p o'. \; r \equiv_{\sett(T')} r'
    }{
\Theta \vdash \Delta^\pospolaritysuper, \exists r' \inq_p o'. \; r \equiv_{\sett(T')} r'
    }\]
and because $r \equiv_{\sett(T')} w$ is $\negpolarity$, we can further infer that the corresponding
proof tree starts as follows
  \[
    \text{
      \AXC{
        $\Theta, z \in r \vdash \Delta^\pospolaritysuper, z \memmac w, \exists r' \inq_p o'. \; r \equiv_{\sett(T')} r'$
      }
      \UIC{
        $\Theta \vdash \Delta^\pospolaritysuper, r \subseteq w, \exists r' \inq_p o'. \; r \equiv_{\sett(T')} r'$
      }
      \AXC{
        $\vdots$
      }
      \UIC{
        $\Theta \vdash \Delta^\pospolaritysuper, w \subseteq r, \exists r' \inq_p o'. \; r \equiv_{\sett(T')} r'$
      }
      \BIC{
        $\Theta \vdash \Delta^\pospolaritysuper, r \equiv_{\sett(T')} w, \exists r' \inq_p o'. \; r \equiv_{\sett(T')} r'$
      }
      \myLeftLabel{$\exists$}
      \UIC{
$\Theta \vdash \Delta^\pospolaritysuper, \exists r' \inq_p o'. \; r \equiv_{\sett(T')} r'$
      }
      \DisplayProof
    }
\]

In particular, we have a strictly smaller subproof of 
\[\Theta, z \in r \vdash \Delta^\pospolaritysuper, z \memmac w, \exists r' \inq_p o'. \; r \equiv_{\sett(T')} r'\]
Applying the induction hypothesis, we get a proof of
  \[\Theta, z \in r \vdash \Delta^\pospolaritysuper, z \memmac w, \exists z' \inq_{mp} o'. \; z \equiv_{\sett(T')} z'\]
  Recall that there is a  max specialization of $\exists r' \inq_p o'. \; r \equiv_{\sett(T') r'}$. Using that same specialization and the admissibility of
  the $\exists$ rule that allows to perform a non-maximal specialization (Lemma~\ref{lem:existsblockadm}),
  we can then obtain a proof with conclusion
  \[
    \Theta, z \in r \vdash \Delta^\pospolaritysuper, \exists r \inq_p o'. \; z \memmac r, \exists z' \inq_{mp} o'. \; z \equiv_{\sett(T')} z'\]
  which concludes our argument, since $\exists r \inq_p o'. \; z \memmac r$ and $\exists z' \inq_{mp} o'. \; z \equiv_{\sett(T')} z'$ are syntactically the
  same.
\end{proof}

\subsection{Proof of  Lemma \ref{lem:equivsettoequivalence}}

\equivsettoequivalence*

\begin{proof}
By induction over the shape of the input derivation, making a case distinction according
  to the last rule applied.
All cases are trivial, except of the case of the rule $\exists$ where
  $\exists r' \inq_p o'. r \equiv_{\sett(T')} r'$ is the main formula. So let us focus on that one.

In that case, the proof necessarily has shape
\[\text{
\AXC{$\Theta, y \in w \vdash \Delta^\pospolaritysuper, y \memmac r, \exists r' \inq_p o'. \; r \equiv_{\sett(T')} r'$}
\myLeftLabel{$\forall$}
\UIC{$\Theta \vdash \Delta^\pospolaritysuper, w \subseteq_{T'} r, \exists r' \inq_p o'. \; r \equiv_{\sett(T')} r'$}
\AXC{$\Theta, x \in r \vdash \Delta^\pospolaritysuper, x \memmac w, \exists r' \inq_p o'. \; r \equiv_{\sett(T')} r'$}
\myLeftLabel{$\forall$}
\UIC{$\Theta \vdash \Delta^\pospolaritysuper, r \subseteq_{T'} w, \exists r' \inq_p o'. \; r \equiv_{\sett(T')} r'$}
  \myLeftLabel{$\wedge$}
\BIC{$\Theta \vdash \Delta^\pospolaritysuper, r \equiv_{\sett(T')} w, \exists r' \inq_p o'. \; r \equiv_{\sett(T')} r'$}
  \myLeftLabel{$\exists$}
\UIC{$\Theta \vdash \Delta^\pospolaritysuper, \exists r' \inq_p o'. \; r \equiv_{\sett(T')} r'$}
\DisplayProof
  }\]
Applying the induction hypothesis to the leaves of that proof, we obtain two proofs
of
\[
  \Theta, y \in w \vdash \Delta^\pospolaritysuper, y \memmac r, \exists r' \inq_p o'. \forall z \in a, \; z \memmac r \equi z \memmac r'
  \qquad \text{and} \qquad 
  \Theta, x \in r \vdash \Delta^\pospolaritysuper, x \memmac w, \exists r' \inq_p o'. \forall z \in a, \; z \memmac r \equi z \memmac r'
\]
  and we can conclude using the admissibility of weakening and Corollary~\ref{cor:memctxadm} twice and replaying the $\wedge/\forall/\exists$ steps in the
  appropriate order (half of the proof derivation is elided below to save space):
\[\text{
  \AXC{
    $\Theta, y \in w \vdash \Delta^\pospolaritysuper, y \memmac r, \exists r' \inq_p o'. \forall z \in a, \; z \memmac r \equi z \memmac r'$}
  \myLeftLabel{Corollary~\ref{cor:memctxadm}}
  \UIC{
    $\Theta \vdash \Delta^\pospolaritysuper, \neg (y \memmac w), y \memmac r, \exists r' \inq_p o'. \forall z \in a, \; z \memmac r \equi z \memmac r'$}
  \myLeftLabel{$\vee$}
  \UIC{$
    \Theta \vdash \Delta^\pospolaritysuper, {y \memmac w} \imp {y \memmac r}, \exists r' \inq_p o'. \forall z \in a, \; z \memmac r \equi z \memmac r'$}
  \AXC{
    $\Theta, x \in r \vdash \Delta^\pospolaritysuper, x \memmac w, \exists r' \inq_p o'. \forall z \in a, \; z \memmac r \equi z \memmac r'$}
  \myLeftLabel{Corollary~\ref{cor:memctxadm}}
  \UIC{
    $\Theta \vdash \Delta^\pospolaritysuper, \neg (x \memmac r), x \memmac w, \exists r' \inq_p o'. \forall z \in a, \; z \memmac r \equi z \memmac r'$}
  \myLeftLabel{$\vee$}
  \UIC{$
    \Theta \vdash \Delta^\pospolaritysuper, {x \memmac r} \imp {x \memmac w}, \exists r' \inq_p o'. \forall z \in a, \; z \memmac r \equi z \memmac r'$}
  \myLeftLabel{$\wedge$}
   \BIC{$
    \Theta \vdash \Delta^\pospolaritysuper, {y \memmac r} \equi {y \memmac w}, \exists r' \inq_p o'. \forall z \in a, \; z \memmac r \equi z \memmac r'$} 
  \myLeftLabel{Lemma~\ref{lem:wkadm}}
  \UIC{$
    \Theta, y \in a \vdash \Delta^\pospolaritysuper, {y \memmac r} \equi {y \memmac w}, \exists r' \inq_p o'. \forall z \in a, \; z \memmac r \equi z \memmac r'$} 
  \myLeftLabel{$\forall$}
  \UIC{$
    \Theta \vdash \Delta^\pospolaritysuper, \forall z \in a. \; {z \memmac r} \equi {z \memmac w}, \exists r' \inq_p o'. \forall z \in a, \; z \memmac r \equi z \memmac r'$} 
  \myLeftLabel{$\exists$}
  \UIC{$
    \Theta \vdash \Delta^\pospolaritysuper, \exists r' \inq_p o'. \forall z \in a, \; z \memmac r \equi z \memmac r'$} 
\DisplayProof
  }\]
\end{proof}

\section{Proof of the main theorem for non-set types}

Recall again our main result:

\thmmainset*

 In the body of the paper, we gave a proof for the case where the type of the defined object is $\sett(T)$ for any $T$.
We now discuss the remaining cases: the base case and the inductive
case for product types.

So assume we are given an implicit
  definition $\varphi(\vv i, \vv a, o)$ and a focused witness, and proceed by induction over the type $o$:
  \begin{itemize}
    \item If $o$ has type $\unit$, then, since there is only one inhabitant in
      type $\unit$, then we can take our explicit definition to be the corresponding $\nrc$ expression $()$.
    \item If $o$ has type $\ur$, then using interpolation on the entailment $\varphi(\vv i, \vv a, o) \imp (\varphi(\vv i, \vv a', o') \imp o =_\ur o')$,
      we obtain $\theta(\vv i, o)$ with $\varphi(\vv i, \vv a, o) \imp \theta(\vv i, \vv a, o)$ and $\theta(\vv i, o) \wedge \varphi(\vv i, \vv a, o') \imp
      o =_\ur o'$. 
      But then we know that $\phi$ implies
$o$ is a subobject of $\vv i$: otherwise we could find a model that contradicts
the entailment. 
There is a $\nrc$ definition $A(\vv i)$ that collects all of the $\ur$-elements 
lying beneath $\vv i$. 
      We can then take $E(\vv i) = \nrcget_T(\{ x \in A(\vv i) \mid \theta(\vv i, x)\})$ as our
      $\nrcwget$ definition of $o$.  The correctness of $E$ follows
from the properties of $\theta$ above.
    \item If $o$ has type $T_1 \times T_2$, recalling the definition of $\equiv_{T_1 \times T_2}$, we have a derivation of
      \[ \varphi(\vv i, \vv a, o) \wedge \varphi(\vv i, \vv a, o') \vdash
      \pi_1(o) \equiv_{T_1} \pi_1(o') \wedge
      \pi_2(o) \equiv_{T_2} \pi_2(o')\] 
      By Lemma~\ref{lem:wedgeprojadm}, we have proofs of
      \[
         \varphi(\vv i, \vv a, o) \wedge \varphi(\vv i, \vv a, o') \vdash
        \pi_i(o) \equiv_{T_i} \pi_i(o') \]
      for $i \in \{1,2\}$.
      Take $o_1$ and $o_2$ to be fresh variables of types $T_1$ and $T_2$. Take $\tilde\phi(\vv i, \vv a, o_1, o_2)$ to be $\phi(\vv i, \vv a, \tuple{o_1,
      o_2})$. By substitutivity (the admissible rule given
by Lemma~\ref{lem:substitutivity}) and applying the $\times_\beta$ rule, we have focused proofs of
      \[ \varphi(\vv i, \vv a, \tuple{o_1,o_2}) \wedge \varphi(\vv i, \vv a, \tuple{o_1', o_2'}) \vdash
      o_i \equiv_{T_i} o'_i 
      \]
      We can apply our inductive hypothesis to obtain a definition $E\IH_i(\vv i)$ for both $i \in \{1,2\}$.
      We can then take our explicit definition to be $\tuple{E\IH_1(\vv i),E\IH_2(\vv i)}$.
  \end{itemize}

\section{Variant of Parameter Collection Theorem, 
Theorem \ref{thm:mainflatcasebounded}, for parameterized definability
in first-order logic} \label{sec:parambethfo}

Our paper has focused on the setting of nested relations, phrasing our results
in terms of the language $\nrc$. We indicated in the conclusion of the paper that
there is a variant of the $\nrc$ parameter collection theorem, Theorem
\ref{thm:mainflatcasebounded}, for the broader context of first-order logic.
In fact, this  first-order version of the result provided the intuition
for the theorem.
In this section we present this variant.

We consider first-order logic with equality and without function symbols,
which also excludes nullary function symbols, that is, individual constants,
whose role is just taken by free individual variables. Specifically, we consider
first-order formulas with the following syntax
\[ \varphi, \psi ~~\bnfeq ~~ P(\vv x) \bnfalt \neg P(\vv x) \bnfalt x = y \bnfalt x \neq y
\bnfalt \top \bnfalt \bot
\bnfalt \varphi \wedge \psi \bnfalt \varphi \vee \psi
\bnfalt \forall x\, \varphi \bnfalt \exists x\, \varphi.
\]
On top of this, we give some ``syntactic sugar''. We define $\neg \varphi$ by
induction over $\varphi$, dualizing every connective, including the
quantifiers, and removing doubled negation. We define implication $\varphi
\imp \psi$ as an abbreviation of $\neg \varphi \vee \psi$ and bi-implication
$\varphi \equi \psi$ as an abbreviation of $(\varphi \imp \psi) \wedge (\psi
\imp \varphi)$. The set of free variables occurring in a formula~$\varphi$ is
denoted by $\FV(\varphi)$ and the set of predicates occurring in~$\varphi$ by
$\pred(\varphi)$.

Figure~\ref{fig:rhflk} shows our proof system for first-order logic. 
It is identical to a system from the prior literature
\footnote{
 G3c+\textsc{Ref}+\textsc{Repl} \cite{negri:plato:structural:2001,bpt},
in the one-sided form of GS3, discussed in Chapter~3 of \cite{bpt}, which
reduces the number of rules.} %
Like the focused proof system we used in the body of the paper for $\deltazero$ formulas,
it is a $1$-sided calculus.
The formulas other than $\Gamma$ in the premise are the \name{active
  formulas} of the rule, while the \name{principal formulas} are the other
formulas in its conclusion. The complementary principal formulas in
\textsc{Ax} have to be literals. The replacement of symbols induced by
equality with \textsc{Repl} is only performed on negative literals.
\begin{figure}[h]
\
\input{rhflk_nosorts.tex}
\caption{One-sided sequent calculus for first-order logic with equality.}
\label{fig:rhflk}
\end{figure}

As in the body of the paper,   a proof tree or derivation is a tree whose nodes are
labelled with sequents, such that the labels of the children of a given node
and that of the node itself are the premises and conclusion, resp., of an
instance of a rule from Figure~\ref{fig:rhflk}. The \emph{conclusion} of a
proof tree is the sequent that labels its root. The proof system is closed
under cut, weakening and contraction. Closure under contraction in particular
makes it suited as basis for ``root-first'' proof search. Read in this
``bottom-up'' way, the $\exists$ rule states that a disjunction with an
existentially quantified formula can be proven if the extension of the
disjunction by a copy of the formula where the formerly quantified
variable~$x$ is instantiated with an arbitrary variable~$t$ can be proven. The
existentially quantified formula is retained in the premise and may be used to
add further instances by applying~$\exists$ again in the course of the proof.

Soundness of the rules is straightforward. For example the $\exists$ rule
could be read as stating that if we deduce a disjunction in which one disjunct
is a formula $\phi$ with $t$ in it, then we can deduce the same disjunction
but with some occurrences of $t$ replaced in that disjunct with an
existentially quantified variable. Completeness of the proof system can also
be proven by a standard Henkin-style construction: indeed, since this
is really ordinary first-order logic, there are proofs in the literature
for systems that are very similar to this one:
\cite{bpt,negri:plato:equality:1998}. 

The system in Figure ~\ref{fig:rhflk} is the analog of our higher-level system
in the body of the paper. We also have a restricted notion of proof, which
admit more efficient algorithms, that we refer to as \emph{focused proofs}. They are
analogous to the focused proof system for
$\deltazero$ formulas in the body of the paper.
But in this first-order context  we can define focused more easily, as an extra condition on proofs
in the system.
We characterize a proof as \name{\folfocused} if
no application of \rulename{AX}, $\top$, $\exists$, \rulename{Ref},
\rulename{Repl} contains in its conclusion a formula whose top-level connective
is $\lor$, $\land$ or $\forall$. This property may be either incorporated
directly into a ``root-first'' proof procedure by constraining rule
applications or it may be ensured by converting an arbitrary given proof tree
to a \folfocused proof tree with the same ultimate consequence. This
conversion is quite straightforward, but may increase the proof size
exponentially.

We now discuss our generalization of the $\nrc$ Parameter Collection Theorem from the body of the paper to this 
first-order setting.
The concept of explicit definition can
be generalized to \name{definition up to parameters and disjunction}: A family
of formulas $\chi_i(\vv z, \vv y, \vv r)$, $1 \leq i \leq n$, provides an
\name{explicit definition up to parameters and disjunction} of a formula
$\lambda(\vv z, \vv l)$ relative to a formula $\phi$ if
\begin{equation}
  \label{eq-def-pd}
\phi \entails \bigvee_{i=n}^n
\exists \vv y \forall \vv z\, (\lambda(\vv z, \vv l) \equi \chi_i(\vv z, \vv
y, \vv r)).\tag{$\star$}
\end{equation}
The entailment (\ref{eq-def-pd}) is considered with restrictions on the
predicates and variables permitted to occur in the $\chi_i$. In the simplest
case, $\lambda(\vv z,\vv l)$ is an atomic formula $p(\vv z)$ with a predicate
that is permitted in $\phi$ but not in the $\chi_i$. The predicate~$p$ is then
said to be \name{explicitly definable up to parameters and disjunction} with
respect to $\phi$ \cite{ck}. 

The disjunction over
a finite family of formulas $\chi_i$ can be consolidated into a single
quantified biconditional as long as the domain 
has size at least $2$ in every model of $\phi$. 
Notice that if $\phi$ has only finite models, then by the
compactness theorem of first-order logic, the size of models must be bounded.
In such cases every formula $\lambda$ is definable 
with sufficiently many parameters.

We can now state our analog of the Parameter Collection Theorem,  Theorem~\ref{thm:mainflatcasebounded}.
\begin{restatable}{theorem}{mainflatcase} \label{thm:mainflatcase}
  Let $\phi$, $\psi$, $\lambda(\vv z, \vv l)$, and $\rho(\vv
  z, \vv y, \vv r)$ be first-order formulas such that
  \[\phi \land \psi
  \entails \exists \vv y \forall \vv z\, (\lambda(\vv z, \vv l) \equi \rho(\vv
  z, \vv y, \vv r)).\] Then there exist first-order formulas $\chi_i(\vv z, {\vv v}_i,
  {\vv c}_i)$, $1 \leq i \leq n$, such that
  \begin{enumerate}
  \item $\phi \land \psi \entails \bigvee_{i=1}^{n} \exists {\vv v}_i \forall
    z (\lambda(\vv z, \vv l) \equi \chi_i(\vv z, {\vv v}_i, {\vv c}_i))$,
  \item ${\vv c}_i \subseteq (\FV(\phi) \cup \vv l) \cap (\FV(\psi) \cup \vv
    r)$,
  \item $\pred(\chi_i) \subseteq (\pred(\phi) \cup \pred(\lambda)) \cap
    (\pred(\psi) \cup \pred(\rho))$.
  \end{enumerate}

  \noindent
  Moreover, given a \folfocused proof of the precondition with the system of
  Fig.~\ref{fig:rhflk}, a family of formulas $\chi_i$, $1 \leq i \leq n$, with
  the claimed properties can be computed in polynomial time in the size of the
  proof tree.
\end{restatable}

In the theorem statement, 
the free variables of $\lambda$ are
$\vv z$ and $\vv l$, and the free variables of $\rho$ are $\vv
z$, $\vv y$, and $\vv r$.
The precondition supposes an explicit
definition~$\rho$ of $\lambda$ up to parameters with respect to a conjunction
$\phi \land \psi$. The conclusion then claims that one can effectively compute
another definition of $\lambda$ with respect to $\phi \land \psi$ that is up
to parameters and disjunction and has a constrained signature: free variables
and predicates must occur in at least one of the ``left side'' formulas $\phi$
and $\lambda$ and also in at least one of the ``right side'' formulas $\psi$
and $\rho$. In other words, %
the theorem states that if
$\sig_L$ and $\sig_R$ are ``left'' and ``right'' signatures such that $\phi$
and $\lambda$ are over $\sig_L$, $\psi$ and $\rho$ are over $\sig_R$, and
$\rho$ provides an explicit definition of $\lambda$ up to parameters with
respect to $\phi \land \psi$, then one can effectively compute another
definition of $\lambda$ with respect to $\phi \land \psi$ that is up to
parameters and disjunction and is just over the intersection of the signatures
$\sig_L$ and $\sig_R$.

We now prove Theorem~\ref{thm:mainflatcase} by induction on the depth of the
proof tree, generalizing the constructive proof method for Craig interpolation
often called Maehara's method \cite{takeuti:1987,bpt,smullyan}. To simplify
the presentation we assume that the tuples $\vv z$ and $\vv y$ in the theorem
statement each consist of a single variable $z$ and $y$, respectively. The
generalization of our argument to tuples of variables is straightforward.

To specify conveniently the construction steps of the family of formulas
$\chi_i$ we introduce the following concept: A \defname{pre-defining
  equivalence up to parameters and disjunction} (briefly $\pdepd$) for a
formula~$\lambda(z, \vv l)$ is a formula~$\delta$ built up from formulas of
the form $\forall z\, (\lambda(z,\vv l) \equi \chi(z, \vec p))$ (where the
left side is always the same formula $\lambda(z,\vv l)$ but the right sides
$\chi(z, \vec p)$ may differ) and a finite number of applications of
disjunction and existential quantification upon variables from the
vectors~$\vec p$ of the right sides. The empty disjunction $\bot$ is allowed
as a special case of a $\pdepd$. By rewriting with the equivalence $\exists
v\, (\delta_1 \lor \delta_2) \equiv \exists v\, \delta_1 \lor \exists v\,
\delta_2$, any $\pdepd$ for $\lambda$ can be efficiently transformed into the
form $\bigvee_{i=1}^n \exists \vs_i \forall z\, (\lambda(z, \vv l) \equi
\chi_i(z,\vs_i,\vec r_i))$ for some natural number $n \geq 0$. That is,
although a $\pdepd$ has in general not the syntactic form of the disjunction
of quantified biconditionals in the theorem statement (thus ``pre-''), it is
corresponds to such a disjunction. The more generous syntax will be convenient
in the induction.
The sets of additional variables $\vv l$ and $\vec p$ in the biconditionals
$\lambda(z,\vv l) \equi \chi(z, \vec p)$ can overlap, but the overlap will be
top-level variables that never get quantified. Although we have defined the
notion of $\pdepd$ for a general $\lambda$, in the proof we just consider
$\pdepd$s for the formula $\lambda$ from the theorem statement.

For $\pdepd$s $\delta$ we provide analogs to $\FV$ and $\pred$ that only yield
free variables and predicates of $\delta$ that occur in a right side of its
binconditionals, which helps to express the restrictions by definability
properties that constrains the signature of exactly those right sides. Recall
that we refer of these right sides as subformulas $\chi(z,\vv p)$. For
$\pdepd$~$\delta$, define $\PREDD(\delta)$ as the set of the predicate symbols
that occur in a subformula~$\chi(z,\vv p)$ of $\delta$ and define
$\ADDLVARS(\delta)$ as the set of all variables that occur in a subformula
$\chi(z,\vv p)$ of $\delta$ and are free in $\delta$. In other words,
$\ADDLVARS(\delta)$ is the set of all variables~$p$ in the vectors $\vv p$ of
the subformulas $\chi(z,\vv p)$ that have an occurrence in $\delta$ which is
not in the scope of a quantifier $\exists p$. If, for example $\delta =
\bigvee_{i=1}^n \exists \vs_i \forall z\, (\lambda(z, \vv l) \equi
\chi_i(z,\vs_i,\vec r_i))$, then $\ADDLVARS(\delta)$ is the set of all
variables in the vectors $r_i$, for $1 \leq i \leq n$. 

To build up $\pdepd$s we provide an operation that only affects the right
sides of the biconditionals, a restricted form of existential quantification.
It is for use in interpolant construction to convert a free variable in right
sides that became illegal into an existentially quantified parameter. For
variables $p,y$ define $\delta\substdef{p}{y}$ as $\delta$ after substituting
all occurrences of $p$ that are within a right side $\chi(z,\vec p)$ and are
free in $\delta$ with $y$. Define $\exdef p\, \delta$ as shorthand for
$\exists y\, \delta\substdef{p}{y}$, where $y$ is a fresh variable. Clearly
$\delta \entails \exdef p\, \delta$ and $p$ has no free occurrences in $\exdef
p\, \delta$ that are in any of the right side formulas $\chi(z,\vec p)$, i.e.,
$p \notin \ADDLVARS(\exdef p\, \delta)$. Occurrences of $p$ in the left sides
$\lambda(z,\vv l)$, if $p$ is a member of $\vv l$, are untouched in $\exdef
p\, \delta$. If $p$ is not in $\vv l$, then $\exdef p\, \delta$ reduces to
ordinary existential quantification $\exists y\, \delta[y/p]$.

We introduce the following symbolic shorthand for the parametric definition on
the right side in the theorem's precondition.
  \[
  \begin{array}{lcl@{\hspace{1em}}lcl}
    \cG & \eqdef & \exists y \forall z\, (\lambda(z,\vv l) \equi \rho(z,y, \vv r)).\\
  \end{array}
  \]
  Note that our proof rules are such that if we have a proof (\folfocused or
  not) that our original top-level ``global'' parametric definition is implied
  by some formula, then every one-sided sequent in the proof must include that
  parametric definition in it. This is because the rules that eliminate a
  formula when read ``bottom-up'' cannot apply to that parametric definition,
  whose outermost logic operator is the existential quantifier. Thus, in our
  inductive argument, we can assume that $\cG$ is always present.

  We write
  \[\seq \SL\sep \SR\sep \cG \ipol{\theta}{\D},\]
  where $\seq \SL, \SR, \cG$ is a sequent, partitioned into three components,
  multisets $\SL$ and $\SR$ of formulas and the formula~$\cG$ from the
  theorem's hypothesis, $\theta$ is a formula and $\D$ is a $\pdepd$, to
  express that the following properties hold:
  \begin{enumerate}[{I}1.]
  \item \label{prop-r-theta} $\valid \SR \lor \theta$.
  \item \label{prop-l-theta} $\valid \lnot \theta \lor \SL \lor \D$.
  \item \label{prop-sig-theta} $\pred(\theta) \subseteq (\pred(\SL) \cup \ADDPREDSLAM)
    \cap (\pred(\SR) \cup \ADDPREDSRHO)$.
  \item \label{prop-fv-theta} $\FV(\theta) \subseteq
    (\FV(\SL) \cup \ADDVARSLAM ) \cap (\FV(\SR) \cup \ADDVARSRHO)$.
  \item \label{prop-form-d} $\D$ is a $\pdepd$ for $\lambda(z,\vv l)$.
  \item \label{prop-sig-d} $\PREDD(\D) \subseteq (\pred(\SL) \cup
    \ADDPREDSLAM) \cap (\pred(\SR) \cup \ADDPREDSRHO)$.
  \item \label{prop-fv-d} $\ADDLVARS(\D) \subseteq (\FV(\SL) \cup \ADDVARSLAM) \cap
    (\FV(\SR) \cup \ADDVARSRHO)$.
  \end{enumerate}

  For a given proof with conclusion $\seq \lnot \phi, \lnot \psi, \cG$,
  corresponding to the hypothesis $\phi \land \psi \entails \cG$ of the
  theorem, we show the construction of a formula $\theta$ and $\pdepd$ $\D$
  such that
  \[\seq \lnot \phi\sep \lnot \psi \sep \cG \ipol{\theta}{\D}.\]
  From properties~\ref{prop-r-theta}--\ref{prop-l-theta} and
  \ref{prop-form-d}-- \ref{prop-fv-d} it is then straightforward to read off
  that the formula $\bigvee_{i = 1}^{n} \exists \vv {v_i} \forall \vv z
  (\lambda(\vv z, \vv l) \equi \chi_i(\vv z,{\vv v}_i,{\vv r}_i))$ obtained
  from $\D$ by propagating existential quantifiers inwards is as claimed in
  the theorem's conclusion.

  Formula $\theta$ plays an auxiliary role in the induction. For the overall
  conclusion of the proof it a side result that is like a Craig interpolant
  of~$\psi$ and $\phi \imp D$, but slightly weaker syntactically constrained
  by taking $\lambda$ and $\rho$ into account: $\FV(\theta) \subseteq
  ((\FV(\phi) \cup \vv l) \cap (\FV(\psi) \cup \vv r))$ and $\pred(\theta)
  \subseteq (\pred(\phi) \cup \pred(\lambda)) \cup \pred(\psi) \cup
  \pred(\rho)$.

  As basis of the induction, we have to show constructions of $\theta$ and
  $\D$ such that $\seq \SL\sep \SR\sep \cG \ipol{\theta}{\D}$ holds for
  \rulename{Ax} and \rulename{$\top$}, considering each possibility in which
  the principal formula(s) can be in $\SL$ or $\SR$. For the induction step,
  there are a number of subcases, according to which rule is last applied and
  which of the partitions $\SL$, $\SR$ or $\cG$ contain the principal
  formula(s). We first discuss the most interesting case, the induction step
  where $\cG$ is the principal formula. This case is similar to the most
interesting case in the NRC Parameter Collection Theorem, covered in the body
of the paper.

  \paragraph*{Case where the principal formula is $\cG$.}

   We now give more detail on the most complex case. If the principal formula
   of a conclusion $\seq \SL\sep\SR\sep\cG$ is $\cG$, then the rule that is
   applied rule must be~$\exists$. From the \folfocused property of the proof
   it follows that the derivation tree ending in $\seq \SL\sep\SR\sep\cG$ must
   have the following shape, for some $u \notin \FV(\SL,\SR,\cG)$ and $w \neq
   u$. Note that $u$ could be either a top-level variable from $\lambda$,
   i.e., a member of $\vv l$, or one introduced during the proof.
    \[
    \begin{array}{rccrc}
      \tabrulename{$\lor$} & \seq \SL, \lnot \lambda(u,\vv l), \SR,
      \rho(u,w,\vv r), \cG
      & \hspace*{2em} &
      \tabrulename{$\lor$}
      & \seq \SL, \lambda(u, \vv l), \SR, \lnot \rho(u,w,\vv r) , \cG\\\cmidrule{2-2}\cmidrule{5-5}
      \tabrulename{$\land$} & \seq \SL, \SR, \lambda(u, \vv l) \imp
      \rho(u,w,\vv r), \cG & &
      & \seq \SL, \lambda(u, \vv l), \SR, \rho(u,w,\vv r) \imp \lambda(u, \vv l), \cG\\\cmidrule{2-5}
      \multicolumn{5}{c}{
        \begin{array}{rc}
          \tabrulename{$\forall$} & \seq \SL, \SR, \lambda(u, \vv l) \equi
          \rho(u,w,\vv r), \cG\\\cmidrule{2-2}
          \tabrulename{$\exists$} & \seq \SL, \SR,
          \forall z\, (\lambda(z, \vv l) \equi \rho(z,w,\vv r)), \cG\\\cmidrule{2-2}
          & \seq \SL,\SR,\cG
      \end{array}}
    \end{array}
    \]
The important point is that the two ``leaves'' of the above tree are both sequents 
where we can apply our induction hypothesis.
    Taking into account the partitioning of the sequents at the bottom
    conclusion and the top premises in this figure, we can express the
    induction step in the form of a ``macro'' rule that specifies the
    how we constructed the required $\theta$ and $\D$ for the conclusion,  making use of
the $\theta_1, \D_1$ and $\theta_2, \D_2$ that we get by applying  the induction hypothesis to
each of the two premises.
    \[
    \begin{array}{c@{\hspace{1em}}cl}
      \seq \SL, \lnot \lambda(u, \vv l);\, \SR, \rho(u,w,\vv r);\, \cG \ipolnarrow{\theta_1}{\D_1} &
      \seq \SL, \lambda(u, \vv l);\, \SR, \lnot \rho(u,w,\vv r);\, \cG \ipolnarrow{\theta_2}{\D_2}
      \\\cmidrule{1-2}
      \multicolumn{2}{c}{\seq \SL \sep  \SR\sep \cG \ipol{\theta}{\D},}
    \end{array}
    \]
        where $u$ is as above. The values of $\theta$ and $\D$ -- the new formula and definition
that we are building --  will depend on occurrences of~$w$, and
we give their construction in cases below:
    
    \begin{enumerate}[(i)]
    \item
      \label{eq-case-g-u}
       If $w \notin \FV(\SL) \cup \ADDVARSLAM$ or $w \in \FV(\SR) \cup \ADDVARSRHO$, then
      \[\begin{array}{lcl}
      \theta & \eqdef & \forall u\, (\theta_1 \lor \theta_2).\\
         \D & \eqdef & \exdef u\, \D_1 \lor \exdef u\, \D_2 \lor \forall z\,
         (\lambda(z, \vv l) \equi \theta_2[z/u]).
       \end{array}\]
     \item \label{eq-case-g-wu}
       Else it holds that $w \in \FV(\SL) \cup \ADDVARSLAM$ and $w \notin \FV(\SR) \cup
       \ADDVARSRHO$. Then
       \[\begin{array}{lcl}
       \theta & \eqdef & \forall w \forall u\, (\theta_1 \lor \theta_2).\\
       \D & \eqdef & \exdef w\, (\exdef u\, \D_1 \lor \exdef u\, \D_2 \lor \forall z\,
       (\lambda(z, \vv l) \equi \theta_2[z/u])).
       \end{array}\]
    \end{enumerate}

    We now verify that $\seq \SL \sep \SR\sep \cG \ipol{\theta'}{\D'}$, that
    is, properties~I\ref{prop-r-theta}--I\ref{prop-fv-d}, hold. The proofs for
    the individual properties are presented in tabular form, with explanations
    annotated in the side column, where \name{IH} stands for \name{induction
      hypothesis}. We concentrate on the case~(\ref{eq-case-g-u}) and indicate
    the modifications of the proofs for case~(\ref{eq-case-g-wu}) in remarks,
    where we refer to the values of $\theta$ and $\D$ for that case in terms
    of the values for the case~(\ref{eq-case-g-u}) as $\forall w\, \theta$ and
    $\exdef w\, \D$. In the proofs of the semantic
    properties~I\ref{prop-r-theta} and~I\ref{prop-l-theta} we let sequents
    stand for the disjunction of their members.
    
    \smallskip
\noindent
Property~I\ref{prop-r-theta}:
\prlReset{r-theta}

\smallskip\noindent
$
\begin{arrayprf}
  \prl{r-theta-1} & \valid \SR \lor \rho(u,w,\vv r) \lor \theta_1.
  & \textrm{IH}\\
  \prl{r-theta-2} & \valid \SR \lor \lnot \rho(u,w,\vv r) \lor \theta_2.
  & \textrm{IH}\\
  \prl{r-theta-3} & \valid \SR \lor \theta_1 \lor \theta_2.
  & \textrm{by~\pref{r-theta-1} and~\pref{r-theta-2}}\\
 \prl{r-theta-4} & \valid \SR \lor \forall u\, (\theta_1 \lor \theta_1).
 & \textrm{by~\pref{r-theta-3} since } u \notin \FV(\SR)\\
 \prl{r-theta-5} & \valid \SR \lor \theta.
 & \textrm{by~\pref{r-theta-4} and def. of } \theta\\
\end{arrayprf}
$

\medskip
For case~(\ref{eq-case-g-wu}), it follows from the precondition $w \notin
\FV(\SR)$ and step~\pref{r-theta-5} that $\valid \SR \lor \forall w\, \theta$.

\medskip

\noindent
Property~I\ref{prop-l-theta}:
\prlReset{l-theta}

\smallskip\noindent
$
\begin{arrayprf}
  \prl{l-theta-i1} & \valid \lnot \theta_1 \lor \SL \lor \lnot \lambda(u, \vv l) \lor \D_1.
  & \textrm{IH}\\
  \prl{l-theta-i2} & \valid \lnot \theta_2 \lor \SL \lor \lambda(u, \vv l) \lor \D_2.
  & \textrm{IH}\\
  \prl{l-theta-3} &
  \valid \lnot (\theta_1 \lor \theta_2) \lor \SL \lor \lnot \lambda(u, \vv l) \lor \D_1 \lor
  \theta_2. & \textrm{by \pref{l-theta-i1}}\\
  \prl{l-theta-4} &
  \valid \lnot (\theta_1 \lor \theta_2) \lor \SL \lor \lnot \lambda(u, \vv l) \lor \D_1 \lor
  \D_2 \lor \theta_2. & \textrm{by \pref{l-theta-3}}\\
  \prl{l-theta-5} &
  \valid \lnot \forall u\, (\theta_1 \lor \theta_2) \lor \SL \lor \lnot \lambda(u, \vv l) \lor \D_1 \lor \D_2 \lor \theta_2. & \textrm{by \pref{l-theta-4}}\\
  \prl{l-theta-6} &
\valid \lnot \forall u\, (\theta_1 \lor \theta_2) \lor \lnot \theta_2 \lor \SL \lor \lambda(u, \vv l) \lor \D_1 \lor \D_2. & \textrm{by \pref{l-theta-i2}}\\
  \prl{l-theta-7} &
  \valid \lnot \forall u\, (\theta_1 \lor \theta_2) \lor \SL \lor \D_1 \lor \D_2
  \lor (\lambda(u, \vv l) \equi \theta_2).
  & \textrm{by \pref{l-theta-6} and \pref{l-theta-5}}\\
  \prl{l-theta-8} &
  \valid \lnot \forall u\, (\theta_1 \lor \theta_2) \lor \SL \lor \exdef u\,
  \D_1 \lor \exdef u\, \D_2
  \lor (\lambda(u, \vv l) \equi \theta_2).
  & \textrm{by \pref{l-theta-7}}\\
  \prl{l-theta-9} &
  \valid \lnot \forall u\, (\theta_1 \lor \theta_2) \lor \SL \lor \exdef u\,
  \D_1 \lor \exdef u\, \D_2
  \lor \forall z\, (\lambda(z, \vv l) \equi \theta_2[z/u]). \hspace*{2em}
  & \textrm{by \pref{l-theta-8} since } u \notin \FV(\SL,\lambda(z, \vv l))\\  
  \prl{l-theta-10} &
  \valid \lnot \theta \lor \SL \lor \D. 
  & \textrm{by \pref{l-theta-9} and defs. of } \theta, \D\\
\end{arrayprf}
$

\smallskip
That $u \notin \FV(\lambda(z, \vv l))$ follows from the precondition $u \notin
\FV(\cG)$. It is used in step~\pref{l-theta-9} to justify that the substitution
$[z/u]$ has only to be applied to $\theta_2$ and not to $\lambda(z, \vv l)$ and, in
addition, to justify that $u \notin \FV(\exdef u\, \D_1)$ and $u \notin
\FV(\exdef u\, \D_2)$, which follow from $u \notin \FV(\lambda(z, \vv l))$ and the
induction hypotheses that property~I\ref{prop-form-d} applies to~$\D_1$
and~$\D_2$.

For case~(\ref{eq-case-g-wu}), it follows from step~\pref{l-theta-10} that
$\valid \lnot \forall w\, \theta \lor \SL \lor \exdef w\, \D$.

\pagebreak
\noindent
Property~I\ref{prop-sig-theta}:
\prlReset{sig-theta}

\smallskip\noindent
$
\begin{arrayprf}
  \prl{sig-theta-1} & \pred(\theta_1) \subseteq
  (\pred(\SL,\lnot \lambda(u, \vv l)) \cup \ADDPREDSLAM) \cap
  (\pred(\SR,\rho(u,w,\vv r)) \cup \ADDPREDSRHO). \hspace*{1em} & \textrm{IH}\\
  \prl{sig-theta-2} &
  \pred(\theta_2) \subseteq (\pred(\SL,\lambda(u, \vv l)) \cup \ADDPREDSLAM) \cap
  (\pred(\SR,\lnot \rho(u,w,\vv r)) \cup \ADDPREDSRHO). & \textrm{IH}\\
  \prl{sig-theta-3} &
  \pred(\forall u\, (\theta_1 \lor \theta_1))
  \subseteq (\pred(\SL) \cup \ADDPREDSLAM) \cap (\pred(\SR) \cup \ADDPREDSRHO).
  & \textrm{by \pref{sig-theta-1}, \pref{sig-theta-2}}\\
  \prl{sig-theta-4} &
  \pred(\theta)
    \subseteq (\pred(\SL) \cup \ADDPREDSLAM) \cap (\pred(\SR) \cup \ADDPREDSRHO).
    & \textrm{by \pref{sig-theta-3} and def. of } \theta
\end{arrayprf}
$

\medskip
For case~(\ref{eq-case-g-wu}) the property follows since $\pred(\theta) =
\pred(\forall w\, \theta)$.

\medskip

\noindent
Property~I\ref{prop-fv-theta}:
\prlReset{fv-theta}

\smallskip\noindent
$
\begin{arrayprf}
  \prl{fv-theta-1} & \FV(\theta_1) \subseteq (\FV(\SL, \lnot \lambda(u, \vv l))
  \cup \ADDVARSLAM)
  \cap (\FV(\SR, \rho(u,w,\vv r)) \cup \ADDVARSRHO). \hspace*{2em}& \textrm{IH}\\
  \prl{fv-theta-2}  & \FV(\theta_2) \subseteq
  (\FV(\SL, \lambda(u, \vv l)) \cup \ADDVARSLAM) \cap
  (\FV(\SR, \lnot \rho(u,w,\vv r)) \cup \ADDVARSRHO).& \textrm{IH}\\
  \prl{fv-theta-x1} & \FV(\theta_1) \subseteq (\FV(\SL) \cup \ADDVARSLAM \cup \{u\})
  \cap (\FV(\SR) \cup \ADDVARSRHO \cup \{u,w\}).
  & \textrm{by \pref{fv-theta-1}}\\
  \prl{fv-theta-x2} & \FV(\theta_1) \subseteq (\FV(\SL) \cup \ADDVARSLAM \cup \{u\})
  \cap (\FV(\SR) \cup \ADDVARSRHO \cup \{u,w\}).
  &\textrm{by \pref{fv-theta-2}}\\
  \prl{fv-theta-3} & \FV(\forall u\, (\theta_1 \lor \theta_2))
  \subseteq (\FV(\SL) \cup \ADDVARSLAM) \cap
  (\FV(\SR) \cup \ADDVARSRHO).
  & \textrm{by \pref{fv-theta-x1}, \pref{fv-theta-x2}
      and the precond.
      $w \notin \FV(\SL) \cup \ADDVARSLAM$ or $w \in \FV(\SR) \cup \ADDVARSRHO$}\\
  \prl{fv-theta-4} &
  \FV(\theta)
    \subseteq (\FV(\SL) \cup \ADDVARSLAM) \cap
    (\FV(\SR) \cup \ADDVARSRHO).
    & \textrm{by \pref{fv-theta-3} and def. of } \theta\\
\end{arrayprf}
$

\medskip

For case~(\ref{eq-case-g-wu}) step~\pref{fv-theta-3} has to be replaced by
\[\FV(\forall w \forall u\, (\theta_1 \lor \theta_2)) \subseteq (\FV(\SL) \cup \ADDVARSLAM) \cap
(\FV(\SR) \cup \ADDVARSRHO),\] which follows just from \pref{fv-theta-x1} and
\pref{fv-theta-x2}. Instead of step~\pref{fv-theta-4} we then have
$\FV(\forall w\, \theta) \subseteq (\FV(\SL) \cup \ADDVARSLAM) \cap (\FV(\SR) \cup
\ADDVARSRHO)$.

\medskip
\noindent
Property~I\ref{prop-form-d}: Immediate from the induction hypothesis and the
definition of $\D$.

\medskip\noindent
Property~I\ref{prop-sig-d}:

\prlReset{sig-d}

\smallskip\noindent
$
\begin{arrayprf}
  \prl{sig-d-1} &
  \PREDD(\D_1) \subseteq (\pred(\SL, \lnot \lambda(u, \vv l)) \cup
  \ADDPREDSLAM) \cap (\pred(\SR, \rho(u,w,\vv r)) \cup \ADDPREDSRHO). \hspace*{1em} &  \textrm{IH}\\
  \prl{sig-d-2} &
  \PREDD(\D_2) \subseteq (\pred(\SL, \lambda(u, \vv l)) \cup
  \ADDPREDSLAM) \cap (\pred(\SR, \lnot \rho(u,w,\vv r)) \cup \ADDPREDSRHO). &  \textrm{IH}\\
  \prl{sig-d-3} &
  \pred(\theta_2) \subseteq (\pred(\SL,\lambda(u, \vv l)) \cup \ADDPREDSLAM) \cap
  (\pred(\SR,\lnot \rho(u,w,\vv r)) \cup \ADDPREDSRHO). & \textrm{IH}\\
  \prl{sig-d-4} &
  \PREDD(\exdef u\, \D_1 \lor \exdef u\, \D_2 \lor
  \forall z\, (\lambda(z, \vv l) \equi
  \theta_2[z/u]))\; \subseteq\\
  & (\pred(\SL) \cup \ADDPREDSLAM) \cap (\pred(\SR) \cup \ADDPREDSRHO)
  & \textrm{by \pref{sig-d-1}--\pref{sig-d-3}}\\
  \prl{sig-d-5} &
  \PREDD(\D) \subseteq
  (\pred(\SL) \cup \ADDPREDSLAM) \cap (\pred(\SR) \cup \ADDPREDSRHO)
  & \textrm{by \pref{sig-d-4} and def. of } \D
\end{arrayprf}
$

\medskip
For case~(\ref{eq-case-g-wu}) the property follows since $\PREDD(\D) =
\PREDD(\exdef w\, \D)$.

\medskip

\medskip\noindent
Property~I\ref{prop-fv-d}:

\prlReset{fv-d}

\smallskip\noindent
$
\begin{arrayprf}
  \prl{fv-d-1} &
  \ADDLVARS(\D_1) \subseteq (\FV(\SL, \lnot \lambda(u, \vv l)) \cup
  \ADDVARSLAM) \cap (\FV(\SR, \rho(u,w,\vv r)) \cup \ADDVARSRHO). %
  & \textrm{IH}\\
  \prl{fv-d-2} &
  \ADDLVARS(\D_2) \subseteq (\FV(\SL, \lambda(u, \vv l)) \cup
  \ADDVARSLAM) \cap (\FV(\SR, \lnot \rho(u,w,\vv r)) \cup \ADDVARSRHO).
  & \textrm{IH}\\
  \prl{fv-d-3} &
  \FV(\theta_2) \subseteq
  (\FV(\SL, \lambda(u, \vv l)) \cup \ADDVARSLAM) \cap
  (\FV(\SR, \lnot \rho(u,w,\vv r)) \cup \ADDVARSRHO).
  & \textrm{IH}\\
  \prl{fv-d-1x} &
  \ADDLVARS(\D_1) \subseteq (\FV(\SL) \cup
  \ADDVARSLAM \cup \{u\}) \cap (\FV(\SR) \cup \ADDVARSRHO \cup \{u,w\}).
  & \textrm{by \pref{fv-d-1}}\\
  \prl{fv-d-2x} &
  \ADDLVARS(\D_2) \subseteq (\FV(\SL) \cup
  \ADDVARSLAM \cup \{u\}) \cap (\FV(\SR) \cup \ADDVARSRHO \cup \{u,w\}).
  & \textrm{by \pref{fv-d-2}}\\
  \prl{fv-d-3x} &
  \FV(\theta_2) \subseteq (\FV(\SL) \cup
  \ADDVARSLAM \cup \{u\}) \cap (\FV(\SR) \cup \ADDVARSRHO \cup \{u,w\}).
  & \textrm{by \pref{fv-d-3}}\\
  \prl{fv-d-4} &
  \ADDLVARS(\exdef u\, \D_1 \lor \exdef u\, \D_2 \lor
  \forall z\, (\lambda(z, \vv l) \equi
  \theta_2[z/u]))\; \subseteq &
   \textrm{by \pref{fv-d-1x}, \pref{fv-d-2x}, \pref{fv-d-3x}
      and the precond. $w \notin \FV(\SL) \cup \ADDVARSLAM$ or}\\
  & (\FV(\SL) \cup \ADDVARSLAM) \cap (\FV(\SR) \cup \ADDVARSRHO).
  &  \hspace{2em} \textrm{$w \in \FV(\SR) \cup \ADDVARSRHO$}\\
  \prl{fv-d-5} &
  \ADDLVARS(\D) \subseteq (\FV(\SL) \cup \ADDVARSLAM) \cap (\FV(\SR) \cup \ADDVARSRHO).
  & \textrm{by \pref{fv-d-4} and def. of } \D
\end{arrayprf}
$

\medskip

For case~(\ref{eq-case-g-wu}) step~\pref{fv-d-4} has to be replaced by
\[\begin{array}{l}
\ADDLVARS(\exdef w (\exdef u\, \D_1 \lor \exdef u\, \D_2 \lor
\forall z\, (\lambda(z, \vv l) \equi \theta_2[z/u])))\; \subseteq\\
(\FV(\SL) \cup \ADDVARSLAM) \cap (\FV(\SR) \cup \ADDVARSRHO),
  \end{array}
  \] which follows just from
  \pref{fv-d-1x}, \pref{fv-d-2x}, \pref{fv-d-3x}. Instead of
  step~\pref{fv-d-5} we then have $\ADDLVARS(\exdef w\, \D) \subseteq (\FV(\SL)
  \cup \ADDVARSLAM) \cap (\FV(\SR) \cup \ADDVARSRHO)$.

  \medskip
  
This completes the verification of correctness, and thus ends our discussion of this
case.

We now turn to the base of the induction along with the other inductive cases.

\paragraph*{Cases where the principal formulas are in the
  $\SL$ or $\SR$ partition.}

The inductive cases where the principal formulas are in the $\SL$ or $\SR$
partition can be conveniently specified in the form of rules that lead from
induction hypotheses of the form $\seq \SL\sep \SR\sep \cG \ipol{\theta}{\D}$
as premises to an induction conclusion of the same form. Base cases can there
be taken just as such rules without premises. The axioms and rules shown below
correspond to the those of the calculus, but replicated for each possible way
in which the partitions $\SL$, $\SR$ or $\cG$ of the conclusion may contain the
principal formula(s). To verify that
properties~I\ref{prop-r-theta}--I\ref{prop-fv-d} are preserved by each of the
shown constructions is in general straightforward, such that we only have
annotated a few subtleties that may not be evident.

\begin{enumerate}[(1)]
\item
  $\begin{array}{rcl}
    \tabrulename{\rulename{Ax}} && \hspace{2em} \tabrulecond{\phi \textit{ an atom}}\\
    \cmidrule{2-2} & \seq \SL, \phi, \lnot \phi\sep  \SR\sep \cG
    \ipol{\top}{\bot}\\
  \end{array}$
  \smallskip

\item
  $\begin{array}{rcl}
    \tabrulename{\rulename{Ax}} && \hspace{2em} \tabrulecond{\phi \textit{ an atom}}\\
  \cmidrule{2-2} & \seq \SL\sep  \SR, \phi, \lnot
  \phi\sep \cG \ipol{\bot}{\bot}
\end{array}$
  \smallskip
  
\item
  $\begin{array}{rcl}
    \tabrulename{\rulename{Ax}} && \hspace{2em} \tabrulecond{\phi \textit{ a literal}}\\
   \cmidrule{2-2} & \seq \SL, \phi\sep  \SR, \lnot \phi\sep \cG
   \ipol{\phi}{\bot}
\end{array}$
  \smallskip

\item
  $\begin{array}{rc}
    \tabrulename{$\top$}\\    
    \cmidrule{2-2} & \seq \SL, \top\sep  \SR\sep \cG
    \ipol{\top}{\bot}
\end{array}$
  \smallskip

\item  
  $\begin{array}{rc}
    \tabrulename{$\top$}\\    
    \cmidrule{2-2} & \seq \SL\sep  \SR, \top\sep \cG
    \ipol{\bot}{\bot}
\end{array}$
  \bigskip

\item
  $\begin{array}[t]{rc}
    \tabrulename{$\lor$} &
    \seq \SL, \phi_1, \phi_2\sep  \SR\sep \cG \ipol{\theta}{\D}\\
    \cmidrule{2-2} & \seq \SL, \phi_1 \lor \phi_2\sep  \SR\sep \cG
  \ipol{\theta}{\D}
\end{array}$
  \bigskip

\item
  $\begin{array}[t]{rc}
    \tabrulename{$\lor$} &
    \seq \SL\sep  \SR, \phi_1, \phi_2 \sep \cG \ipol{\theta}{\D}\\
    \cmidrule{2-2} & \seq \SL\sep  \SR, \phi_1 \lor \phi_2\sep \cG
    \ipol{\theta}{\D}
\end{array}$
  \bigskip

\item
  \label{rule-and-1}
  $\begin{array}[t]{rc@{\hspace{2em}}c}
    \tabrulename{$\land$} &
    \seq \SL, \phi_1 \sep  \SR\sep \cG \ipol{\theta_1}{\D_1}
  & \seq \SL, \phi_2\sep  \SR\sep \cG \ipol{\theta_2}{\D_2}\\\cmidrule{2-3} &
  \multicolumn{2}{c}{\seq \SL, \phi_1 \land \phi_2\sep  \SR\sep \cG
    \ipol{\theta_1 \land \theta_2}{\D_1 \lor \D_2}}
\end{array}$
  \bigskip

\item
  $\begin{array}[t]{rc@{\hspace{2em}}c}
    \tabrulename{$\land$} &    
    \seq \SL \sep  \SR, \phi_1\sep \cG \ipol{\theta_1}{\D_1}
  & \seq \SL \sep \SR, \phi_2\sep \cG \ipol{\theta_2}{\D_2}\\\cmidrule{2-3} &
  \multicolumn{2}{c}{\seq \SL\sep  \SR, \phi_1 \land \phi_2 \sep \cG
 \ipol{\theta_1 \lor \theta_2}{\D_1 \lor \D_2}}
\end{array}$
  \bigskip

\item  \label{rule-exists-1}
  $\begin{array}[t]{rc}
    \tabrulename{$\exists$} &    
    \seq \SL, \phi\subst{x}{t}, \exists x\, \phi\sep  \SR\sep \cG
  \ipol{\theta}{\D}\\\cmidrule{2-2} &
  \seq \SL, \exists x\, \phi\sep  \SR\sep \cG \ipol{\theta'}{\D'},
\end{array}$ \par
  where the values of $\theta'$ and $\D'$ depend on occurrences of~$t$:
  \begin{itemize}
  \item If $t \in \FV(\SL,\exists x\, \phi) \cup \ADDVARSLAM$, then $\theta' \eqdef
    \theta$ and $\D' \eqdef \D$.
  \item Else it holds that $t \notin \FV(\SL,\exists x\, \phi) \cup \ADDVARSLAM$.
    Then $\theta' \eqdef \exists t\, \theta$ and $\D' \eqdef \exdef t\, \D$.
  \end{itemize}
  \bigskip

\item \label{eq-case-ex-r}
  $\begin{array}[t]{rc}
    \tabrulename{$\exists$} &    
    \seq \SL\sep  \SR, \phi\subst{x}{t}, \exists x\, \phi\sep \cG \ipol{\theta}{\D}\\
  \cmidrule{2-2} & \seq \SL\sep \SR, \exists x\, \phi \sep \cG
  \ipol{\theta'}{\D'},
\end{array}$ \par
  where the values of $\theta'$ and $\D'$ depend on occurrences of~$t$:
  \begin{itemize}
  \item If $t \in \FV(\SR,\exists x\, \phi) \cup \ADDVARSRHO$, then
    $\theta' \eqdef \theta$ and $\D' \eqdef \D$.
  \item Else it holds that $t \notin \FV(\SR,\exists x\, \phi) \cup \ADDVARSRHO$. Then
    $\theta' \eqdef \forall t\, \theta$ and $\D' \eqdef \exdef t\, \D$.
  \end{itemize}
  \bigskip

\item
  $\begin{array}[t]{rcl}
    \tabrulename{$\forall$} &    
    \seq \SL, \phi[y/x]\sep  \SR\sep \cG \ipol{\theta}{\D}
  & \hspace{2em} \tabrulecond{y \notin \FV(\SL,\forall x\, \phi,\SR,\cG)}\\
    \cmidrule{2-2} & \seq \SL, \forall x\, \phi\sep  \SR\sep \cG
    \ipol{\theta}{\D}
\end{array}$
  \bigskip
  
\item
  $\begin{array}[t]{rcl}
    \tabrulename{$\forall$} &    
    \seq \SL\sep  \SR,\phi[y/x]\sep \cG \ipol{\theta}{\D}
  & \hspace{2em} \tabrulecond{y \notin \FV(\SL,\SR,\forall x\, \phi,\cG)}\\
  \cmidrule{2-2} & \seq \SL\sep  \SR, \forall x\, \phi\sep \cG
 \ipol{\theta}{\D}
\end{array}$
  \bigskip

\item \label{eq-rule-ref-l}
  $\begin{array}[t]{rcl}
    \tabrulename{\rulename{Ref}} &    
    \seq t \neq t, \SL\sep \SR\sep \cG \ipol{\theta}{\D}
    & \hspace{2em} \tabrulecond{t \in \FV(\SL) \cup \ADDVARSLAM}\\
    \cmidrule{2-2} &
    \seq \SL\sep \SR\sep \cG \ipol{\theta}{\D}
\end{array}$
  \bigskip
  
\item \label{eq-rule-ref-r}
  $\begin{array}[t]{rcl}
    \tabrulename{\rulename{Ref}} &    
    \seq \SL\sep t \neq t, \SR\sep \cG \ipol{\theta}{\D}
    & \hspace{2em} \tabrulecond{t \notin \FV(\SL) \cup \ADDVARSLAM}\\    
    \cmidrule{2-2} &
    \seq \SL\sep \SR\sep \cG \ipol{\theta}{\D}
  \end{array}$
  \bigskip
  
\item
  $\begin{array}[t]{rcl}
    \tabrulename{\rulename{Repl}} &    
    \seq t \neq u, \phi[u/x], \phi[t/x], \SL\sep \SR\sep \cG \ipol{\theta}{\D}
    & \hspace{2em} \tabrulecond{\phi \textit{ a negative literal}}
    \\\cmidrule{2-2} &
    \seq t \neq u, \phi[t/x], \SL\sep \SR\sep \cG \ipol{\theta}{\D}
\end{array}$
  \bigskip

\item
  $\begin{array}[t]{rcl}
  \tabrulename{\rulename{Repl}} &    
  \seq t \neq u, \SL\sep \phi[u/x], \phi[t/x], \SR\sep \cG \ipol{\theta}{\D}
  & \hspace{2em} \tabrulecond{\phi \textit{ a negative literal}}
  \\\cmidrule{2-2} &
  \seq t \neq u, \SL\sep \phi[t/x], \SR\sep \cG \ipol{\theta'}{\D'},
\end{array}$ \par
  where the values of $\theta'$ and $\D'$ depend on occurrences of $t$ and $u$:
  \begin{itemize}
  \item If $t \notin \FV(\phi[t/x], \SR) \cup \ADDVARSRHO$, then $\theta' \eqdef
    \theta$ and $\D' \eqdef \D$. In this subcase the precondition $t
    \notin \FV(\phi[t/x])$ implies that $x \notin \FV(\phi)$ and thus
    $\phi[u/x] = \phi[t/x]$.
  \item If $t,u \in \FV(\phi[t/x], \SR) \cup \ADDVARSRHO$, then $\theta' \eqdef
    \theta \lor t \neq u$ and $\D' \eqdef \D$.
  \item Else it holds that $t \in \FV(\phi[t/x], \SR) \cup \ADDVARSRHO$ and $u
    \notin \FV(\phi[t/x], \SR) \cup \ADDVARSRHO$. Then $\theta' \eqdef \theta[t/u]$
    and $\D' \eqdef \D\substdef{u}{t}$. For this subcase, to derive
    property~I\ref{prop-r-theta} it is used that the precondition $u \notin
    \phi[t/x]$ implies that $\phi[u/x][t/u] = \phi[t/x]$.
  \end{itemize}
  \medskip
  
\item
  $\begin{array}[t]{rcl}
    \tabrulename{\rulename{Repl}} &    
    \seq \SL\sep t \neq u, \phi[u/x], \phi[t/x], \SR\sep \cG \ipol{\theta}{\D}
    & \hspace{2em} \tabrulecond{\phi \textit{ a negative literal}}
    \\\cmidrule{2-2} &
    \seq \SL\sep t \neq u, \phi[t/x], \SR\sep \cG \ipol{\theta}{\D}
  \end{array}$
  \bigskip

\item
  $\begin{array}[t]{rcl}
    \tabrulename{\rulename{Repl}} &    
    \seq \phi[u/x], \phi[t/x], \SL\sep t \neq u; \SR\sep \cG \ipol{\theta}{\D}
    & \hspace{2em} \tabrulecond{\phi \textit{ a negative literal}}
    \\\cmidrule{2-2} &
    \seq \phi[t/x], \SL\sep t \neq u, \SR\sep \cG \ipol{\theta'}{\D'},
  \end{array}$ \par
  where the values of $\theta'$ and $\D'$ depend on occurrences of $t$ and $u$:
  \begin{itemize}
  \item If $t \notin \FV(\phi[t/x], \SL) \cup \ADDVARSLAM$, then $\theta' \eqdef
    \theta$ and $\D' \eqdef \D$. In this subcase the precondition $t
    \notin \FV(\phi[t/x])$ implies that $x \notin \FV(\phi)$ and thus
    $\phi[u/x] = \phi[t/x]$.
  \item If $t,u \in \FV(\phi[t/x], \SL) \cup \ADDVARSLAM$, then $\theta' \eqdef
    \theta \land t = u$ and $\D' \eqdef \D$.
  \item Else it holds that $t \in \FV(\phi[t/x], \SL) \cup \ADDVARSLAM$ and $u
    \notin \FV(\phi[t/x], \SL) \cup \ADDVARSLAM$. Then $\theta' \eqdef \theta[t/u]$
    and $\D' \eqdef \D\substdef{u}{t}$. To derive property~I\ref{prop-l-theta}
    for this subcase, that is, $\valid \lnot \theta[t/u] \lor \phi[t/x], \SL
    \lor \D\substdef{u}{t}$, it is required that $u \notin
    \FV(\D\substdef{u}{t})$, which follows from the precondition $u \notin
    \ADDVARSLAM$ of the subcase.
  \end{itemize}
\end{enumerate}

This completes the proof of Theorem~\ref{thm:mainflatcase}.

\section{From unrestricted proofs to focused proofs}

In the body of the paper we mentioned that our focused proof system is complete for semantic entailment
$\models$.
We explained that this can be argued directly via a  Henkin construction, or by translating the rules for the higher-level
system, shown in  Figure \ref{fig:dzcalc-2sided}
 into the focused system of Figure \ref{fig:dzcalc-1sided}.
Then completeness of the latter system follows from  completeness of the former, which is completely standard.
We now discuss the approach of translating general proofs to focused proofs in more detail.
We explain it  for the first-order proof system outlined in the prior appendix section,
Section \ref{sec:parambethfo}, where the high-level idea can be conveyed independently  of
the additional syntactic restrictions that come from $\deltazero$ formulas.  A similar approach can be applied
in the $\deltazero$ context.

Recall that in Section \ref{sec:parambethfo}
we gave a standard first-order proof system --- see
Figure~\ref{fig:rhflk}. We then 
introduced the first-order version of focused proofs.
A proof is called focused if no application of \rulename{Ax}, $\top$,
$\exists$, \rulename{Ref}, \rulename{Repl} contains in its conclusion a
formula whose top-level connective is $\lor$, $\land$ or $\forall$.

We now make precise the translation between these:

\begin{theorem} \label{thm:focused}
Any proof tree (in the calculus presented in Fig.~\ref{fig:rhflk}) can be
turned into a focused proof tree, in exponential time.
\end{theorem}

We sketch a proof of this theorem.

A counterexample to the focused property is a node in the proof tree where
one of the rules \rulename{Ax}, $\top$, $\exists$, \rulename{Ref} or
\rulename{Repl} is applied and the conclusion contains a formula with $\lor$,
$\land$ or $\forall$ as top-level operator. We eliminate the counterexample by
converting the proof of the conclusion depending on the rule and the top-level
operator. We iterate this until there is no counterexample. The steps
involving connectives other than $\land$ increase the tree size only linearly,
while converting a $\wedge$ step may double the size. The exponential time
bound follows from this. In the following, we show the particular conversion
for some chosen cases of counterexamples. For the remaining cases they are
analogous. First, we consider the case where $\land$ appears as top-level
operator in the conclusion of an application of $\exists$. The proof of the
conclusion then has the form shown on the left below and is converted to the
form shown on the right.
\[\small
\begin{array}{rc}
  & P\\[-1ex]
  & \prooftreesymbol\\
  \tabrulename{$\exists$} &
  \seq \Gamma, \psi_1 \land \psi_2, \phi[t/x], \exists x\, \phi\\
  \cmidrule{2-2} &
  \seq \Gamma, \psi_1 \land \psi_2, \exists x\, \phi
\end{array}
\;\textrm{ is converted to }\;
\begin{array}{rcrc}
  & P' & & P''\\[-1ex]
  & \prooftreesymbol && \prooftreesymbol\\
  \tabrulename{$\exists$} &
  \seq \Gamma, \psi_1, \phi[t/x], \exists x\, \phi &
  \hspace{1em}\tabrulename{$\exists$}
  & \seq \Gamma, \psi_2, \phi[t/x], \exists x\, \phi\\
  \cmidrule{2-2}\cmidrule{4-4}
  \tabrulename{$\land$} &
  \seq \Gamma, \psi_1, \exists x\, \phi &&
  \seq \Gamma, \psi_2, \exists x\, \phi\\
  \cmidrule{2-4}
  & \multicolumn{3}{c}{\seq \Gamma, \psi_1 \land \psi_2, \exists x\, \phi},\\
\end{array}
\]
where proof~$P'$ is like proof~$P$ except that all applications of $\land$
where the principal formula is $\psi_1 \land \psi_2$ and corresponds (i.e., is
passed down) to the shown occurrence of $\psi_1 \land \psi_2$ are removed and
all occurrences of $\psi_1 \land \psi_2$ that correspond to the shown
occurrence are replaced by $\psi_1$. Proof~$P''$ is defined like $P'$, except
that the occurrences of $\psi_1 \land \psi_2$ that correspond to the shown
occurrence are replaced by $\psi_2$ instead of $\psi_1$.
In the case where $\forall$ appears as top-level operator in the conclusion of
an application of $\exists$, the conversion is as follows.
\[\small
\begin{array}{rc}
  & P\\[-1ex]
  & \prooftreesymbol\\
  \tabrulename{$\exists$} &
  \seq \Gamma, \forall u\, \psi, \phi[t/x], \exists x\, \phi\\
  \cmidrule{2-2} &
  \seq \Gamma, \forall u\, \psi, \exists x\, \phi
\end{array}
\;\;\;\textrm{ is converted to }\;\;\;
\begin{array}{rc}
  & P'\\[-1ex]
  & \prooftreesymbol\\
  \tabrulename{$\exists$} &
  \seq \Gamma, \psi[y/u], \phi[t/x], \exists x\, \phi\\
  \cmidrule{2-2}
  \tabrulename{$\forall$} &
  \seq \Gamma, \psi[y/u], \exists x\, \phi\\
  \cmidrule{2-2} &
  \seq \Gamma, \forall u\, \psi, \exists x\, \phi\\
\end{array}
\]
where $y$ is a fresh variable and proof~$P'$ is like proof~$P$ except that all
applications of $\forall$ where the principal formula is $\forall u\, \psi$
and corresponds to the shown occurrence of $\forall u\, \psi$ are removed and
all occurrences of $\forall u\, \psi$ that correspond to the shown occurrence
are replaced by $\psi[y/u]$. Finally, we show the case where $\lor$ appears as
top-level operator in the conclusion of an application of \rulename{Ax}.
\vspace{-2ex}
\[\small
\begin{array}{rc}
  \tabrulename{\rulename{Ax}}\\
  \cmidrule{2-2} &
  \seq \Gamma, \psi_1 \lor \psi_2, \phi, \lnot \phi
\end{array}
\;\;\;\textrm{ is converted to }\;\;\;
\begin{array}{rc}
  \tabrulename{\rulename{Ax}}\\
  \cmidrule{2-2}
  \tabrulename{$\lor$}
  &
  \seq \Gamma, \psi_1, \psi_2, \phi, \lnot \phi\\
  \cmidrule{2-2}
  &
  \seq \Gamma, \psi_1 \lor \psi_2, \phi, \lnot \phi
\end{array}
\]

One can observe that the translation step increases the size of a proof by a factor of at most $2$. The exponential
bound follows from this.

\end{document}